\documentclass[aps,pra,groupedaddress,amsmath,amssymb,tightenlines,notitlepage,nofootinbib,onecolumn,superscriptaddress,11pt]{revtex4-1}
\pdfoutput=1

\PassOptionsToPackage{pdftitle={Postselected quantum hypothesis testing}}{hyperref}
\usepackage{bm,bbm}% 
\usepackage{braket}
\usepackage{amsthm,amsmath,amssymb}

\usepackage[utf8]{inputenc}

\usepackage[dvipsnames,svgnames]{xcolor}
\usepackage{amsmath,amssymb}
\usepackage{changepage,afterpage}
\usepackage{etruscan}
\usepackage{tikz}
\usepackage{enumerate}
\usepackage{mathtools}
\usepackage[colorlinks=true, allcolors=MidnightBlue!70!black!70!TealBlue]{hyperref}

\usepackage{array,booktabs,multirow}
\usepackage{siunitx}

\usepackage{newpxtext,newpxmath}

\newtheorem{theorem}{Theorem}

\newtheorem{proposition}[theorem]{Proposition}
\newtheorem{corollary}[theorem]{Corollary}

\newtheorem{lemma}[theorem]{Lemma}

\theoremstyle{definition}
\newtheorem*{remark}{Remark}

\newtheoremstyle{red}{}{}{\normalfont}{}{\color{red!80!black}\bfseries}{.}{ }{}
\theoremstyle{red}

\let\nc\newcommand

\nc{\NegW}{W_\tau}

\nc{\RW}{\Omega_\FF}

\renewcommand{\*}{\textup{*}}
\newcommand{\aster}{{\mbox{$*$}}}
\newcommand{\<}{\left\langle}
\renewcommand{\>}{\right\rangle}

\renewcommand{\bar}{\;\rule{0pt}{9.5pt}\right|\;}

\newcommand{\lset}{\left\{\left.}
\newcommand{\rset}{\right\}}
\nc{\lsetr}{\left\{\,}
\nc{\rsetr}{\right.\right\}}
\nc{\barr}{\;\rule{0pt}{9.5pt}\left|\;}

\DeclareMathOperator{\Tr}{Tr}
\DeclareMathOperator{\supp}{supp}

\newcommand{\norm}[2]{\left\lVert#1\right\rVert_{\,#2}}
\newcommand{\proj}[1]{\ket{#1}\!\bra{#1}}
\nc{\id}{\mathbbm{1}}

\newcommand{\E}{\mathcal{E}}

\newcommand{\N}{\mathcal{N}}

\newcommand{\D}{\mathcal{D}}

\newcommand{\M}{\mathcal{M}}
\renewcommand{\O}{\mathcal{O}}

\newcommand{\V}{\mathcal{V}}

\newcommand{\C}{\mathcal{C}}
\renewcommand{\S}{\mathcal{S}}

\newcommand{\omin}{\!\underset{\smash{\min}}{\otimes}\!}
\newcommand{\omax}{\!\underset{\smash{\max}}{\otimes}\!}

\newcommand{\RR}{\mathbb{R}}

\newcommand{\DD}{\mathbb{D}}
\newcommand{\NN}{\mathbb{N}}

\newcommand{\FF}{\mathcal{F}}

\nc{\Dmax}{D_{\max}}

\let\O\OO

\nc{\SEP}{\mathrm{SEP}}
\nc{\STAB}{\mathrm{STAB}}
\nc{\PPT}{\mathrm{PPT}}
\nc{\PPTP}{\mathrm{PPTP}}
\nc{\SEPP}{\mathrm{SEPP}}
\nc{\SP}{\mathrm{KP}}
\nc{\FP}{\mathrm{FP}}
\nc{\CPTP}{{\mathrm{CPTP}}}
\let\wt\widetilde

\nc{\Op}{\mathcal{O}}

\nc{\idc}{\mathrm{id}}
\nc{\ve}{\varepsilon}

\nc{\Omax}{\O_{\mathrm{max}}}

\usepackage{stackengine,moresize}
\stackMath

\let\texteq\relax

\newcommand{\texteq}[1]{\stackrel{\mathclap{\scriptsize \mbox{#1}}}{=}}

\nc{\sminfty}{{\infty,\bullet}}

\usepackage[most,breakable]{tcolorbox}
\renewenvironment{boxed}[1]%
	{\expandafter\ifstrequal\expandafter{#1}{orange}{\begin{tcolorbox}[colback=orange!5,colframe=orange!15,breakable,enhanced]}{\begin{tcolorbox}[colback=white,colframe=gray!10,breakable,enhanced]}}%
	{\end{tcolorbox}}

\nc{\regrob}{\DD_{s,\FF}^{\infty}}
\nc{\regrobs}{\DD_{s,\FF}^{\sminfty}}

\makeatletter 
\renewcommand\onecolumngrid{% 
\do@columngrid{one}{\@ne}%
\def\set@footnotewidth{\onecolumngrid}% <<<<<<<<<<<<<<<<
\def\footnoterule{\kern-6pt\hrule width 1.5in\kern6pt}%
}
\makeatother

\usepackage{fmtcount,refcount}

\let\oldproofname\proofname
\renewcommand{\proofname}{\rm\bf{\oldproofname}}

\newcommand{\deff}[1]{\textbf{\emph{#1}}}

\usepackage{pifont}

\newcommand{\para}[1]{\textbf{{#1}.} }

\newcommand{\cbeta}{\overline{\beta}}
\newcommand{\calpha}{\overline{\alpha}}
\newcommand{\cperr}{\overline{p}_{\rm err}}
\newcommand{\cDH}{\overline{D}_{H}}

\newcommand{\cbetaAny}{\cbeta_\ve^{\,\rm any}}
\newcommand{\cbetaPar}{\cbeta_\ve^{\,\rm parallel}}
\newcommand{\cperrAny}{\cperr^{\,\rm any}}
\newcommand{\cperrPar}{\cperr^{\,\rm parallel}}

\newcommand{\Thomp}{\Xi}
\newcommand{\pbar}{\,|\,}

\def\DD{D}

\begin{document}

\title{Postselected quantum hypothesis testing}
 
\author{Bartosz Regula}
\email{bartosz.regula@gmail.com}
\affiliation{Mathematical Quantum Information RIKEN Hakubi Research Team, RIKEN Cluster for Pioneering Research (CPR) and RIKEN Center for Quantum Computing (RQC), Wako, Saitama 351-0198, Japan}
\affiliation{Department of Physics, Graduate School of Science, The University of Tokyo, Bunkyo-ku, Tokyo 113-0033, Japan}

\author{Ludovico Lami}
\email{ludovico.lami@gmail.com}
\affiliation{Institut f\"{u}r Theoretische Physik und IQST, Universit\"{a}t Ulm, Albert-Einstein-Allee 11, D-89069 Ulm, Germany}

\author{Mark M. Wilde}
\email{wilde@cornell.edu}
\affiliation{School of Electrical and Computer Engineering, Cornell University, Ithaca, New York 14850, USA}

%%%%

\begin{abstract}
We study a variant of quantum hypothesis testing wherein an additional `inconclusive' measurement outcome is added, allowing one to abstain from attempting to discriminate the hypotheses. The error probabilities are then conditioned on a successful attempt, with inconclusive trials disregarded. 
We completely characterise this task in both the single-shot and asymptotic regimes, providing exact formulas for the optimal error probabilities. In particular, we prove that the asymptotic error exponent of discriminating any two quantum states $\rho$ and $\sigma$ is given by the Hilbert projective metric $D_{\max}(\rho\|\sigma) + D_{\max}(\sigma \| \rho)$ in asymmetric hypothesis testing, and by the Thompson metric $\max\! \big\{ D_{\max}(\rho\|\sigma),\, D_{\max}(\sigma \| \rho) \big\}$ in symmetric hypothesis testing. This endows these two quantities with fundamental operational interpretations in quantum state discrimination. Our findings extend to composite hypothesis testing, where we show that the asymmetric error exponent with respect to any convex set of density matrices is given by a regularisation of the Hilbert projective metric.
We apply our results also to quantum channels, showing that no advantage is gained by employing adaptive or even more general discrimination schemes over parallel ones, in both the asymmetric and symmetric settings. Our state discrimination results make use of no properties specific to quantum mechanics and are also valid in general probabilistic theories.
\end{abstract}

\maketitle

%%%%%%%%%%%%%%%%%%%%%%%%%%%%%%%%%%%%%%%%%%%%%%%%%%%%%%%%%%%%%%%%%%
%%%%%%%%%%%%%%%%%%%%%%%%%%%%%%%%%%%%%%%%%%%%%%%%%%%%%%%%%%%%%%%%%%
%%%%%%%%%%%%%%%%%%%%%%%%%%%%%%%%%%%%%%%%%%%%%%%%%%%%%%%%%%%%%%%%%%

\section{Introduction}\label{sec:intro}

\subsection{Background}

Quantum hypothesis testing --- a generalisation of the crucial statistical primitive of hypothesis testing to the non-commutative setting --- is a fundamental task that directly underlies many protocols in quantum information science~\cite{hayashi_2016,watrous_2018,wilde_2017}. It is concerned with the problem of distinguishing different quantum states or channels, and in particular asks about the ultimate limits of such distinguishability.

In quantum state discrimination, after receiving one of two possible states, $\rho$ (null hypothesis) or $\sigma$ (alternative hypothesis), one aims to determine which of the two was actually obtained. This is typically done by performing a quantum measurement, or positive operator-valued measure (POVM), given by the tuple $M \coloneqq \{M_1, M_2 \}$ where $M_2 = \id - M_1$. If measurement outcome 1 is obtained (corresponding to the operator $M_1$), one guesses that the state is $\rho$; if outcome 2 is obtained, one guesses $\sigma$. There are then two types of errors that can occur: the \emph{type I error} with probability
\begin{equation}\begin{aligned}
  \alpha(M) \coloneqq \Tr M_2 \rho,
\end{aligned}\end{equation}
which corresponds to incorrectly guessing $\sigma$ when $\rho$ was true (false positive), and the \emph{type II error} with probability
\begin{equation}\begin{aligned}
  \beta(M) \coloneqq \Tr M_1 \sigma,
\end{aligned}\end{equation}
which is the case of guessing $\rho$ when $\sigma$ was true (false negative). Two settings naturally emerge here: one is \emph{asymmetric hypothesis testing}, which asks about how small one of the errors can be subject to constraints on the other error; another setting is \emph{symmetric hypothesis testing}, where the two errors are treated equally and one aims to minimise their average. Specifically, the asymmetric case is concerned with the study of the optimised type II error probability
\begin{equation}\begin{aligned}
  \beta_\ve(\rho, \sigma) \coloneqq& \min_{M \in \M_2} \lset \beta(M) \bar \alpha(M) \leq \ve \rset\\
  =&\!\! \min_{M_1,M_2 \geq 0} \lset \Tr M_1 \sigma \bar M_1 + M_2 = \id,\; \Tr M_2 \rho \leq \ve \rset,
\end{aligned}\end{equation}
where $\M_2$ denotes the set of all two-outcome measurements. The symmetric error depends also on the prior probabilities associated with the two hypotheses, that is, the probability that the state to be discriminated is $\rho$ or that it is $\sigma$. Denoting the respective priors by $p$ and $q \coloneqq 1-p$, we define the symmetric error probability as
\begin{equation}\begin{aligned}
  p_{\rm err} (\rho, \sigma \pbar p, q) \coloneqq& \min_{M \in \M_2} \,\big[ p \alpha(M) + q \beta(M)\big]\\
  =&\!\! \min_{M_1,M_2 \geq 0} \lset p \Tr M_2 \rho + q \Tr M_1 \sigma \bar M_1 + M_2 = \id \rset.
\end{aligned}\end{equation}
The two optimised error probabilities defined above are not too difficult to characterise: they are both efficiently computable semidefinite programs, and the symmetric error even admits an exact solution given by the celebrated Helstrom--Holevo theorem~\cite{helstrom_1969,holevo_1972} in terms of the trace distance (total variation distance),
\begin{equation}\begin{aligned}
  p_{\rm err} (\rho, \sigma \pbar p, q) = \frac12 \left( 1 - \norm{p \rho - q \sigma}{1} \right),
\end{aligned}\end{equation}
where $\norm{\cdot}{1}$ denotes the trace norm (Schatten 1-norm).

However, a limitation of the above formulations is that they assume access to only a single copy of the state to be discriminated. This is not reflective of situations encountered in practice, where one typically deals with sources that emit many copies of quantum states, which we can use to our advantage in attempting to discriminate them. That is, we can assume that we receive either $\rho^{\otimes n}$ or~$\sigma^{\otimes n}$ and then perform a joint measurement over all $n$ systems. The crucial question then is: how exactly does the performance in discriminating two states improve when we have access to more copies of them? The ultimate limit of such protocols is described by the asymptotic setting with $n \to \infty$, and the figures of merit then are the decay rates of $\beta_\ve(\rho^{\otimes n},\sigma^{\otimes n})$ and $p_{\rm err} (\rho^{\otimes n}, \sigma^{\otimes n} \pbar p, q)$, known as the error exponents. These two limits give rise to two of the most fundamental quantities in quantum information. The first one is the \emph{quantum relative entropy} $D(\rho \| \sigma) \coloneqq \Tr \rho (\log \rho - \log \sigma)$~\cite{umegaki_1962}, which is known to be exactly equal to the error exponent in asymmetric hypothesis testing:
\begin{equation}\begin{aligned}
  \lim_{n \to \infty} - \frac1n \log \beta_\ve(\rho^{\otimes n},\sigma^{\otimes n}) = D(\rho \| \sigma) \qquad \forall \ve \in (0,1).
\end{aligned}\end{equation}
This was shown in the works of Hiai and Petz~\cite{hiai_1991} (achievability and weak converse%
\footnote{In the setting of asymptotic asymmetric hypothesis testing, a `weak converse' for the error exponent is understood as one that holds in the limit $\ve \to 0$, while a `strong converse' holds for every $\ve \in (0,1)$. The motivation for considering the latter is that, for every type~II error exponent larger than a strong converse bound, the type~I error must be prohibitively large --- it must asymptotically converge to~1, making discrimination effectively impossible.}) and Ogawa and Nagaoka~\cite{ogawa_2000} (strong converse).
The above result, known as the quantum Stein's lemma, generalises a seminal result due to Stein and Chernoff~\cite{stein_unpublished,chernoff_1956}. 
It provides an important operational interpretation of the quantum relative entropy, and in a large part motivates its extensive use in quantum information as a quantifier of the distinguishability between two quantum states and a meaningful generalisation of the Kullback--Leibler divergence~\cite{kullback_1951}. 

An analogous result in symmetric hypothesis testing, generalising a finding of Chernoff~\cite{chernoff_1952}, was obtained in the works of Nussbaum and Szko\l{}a~\cite{nussbaum_2009} (converse) and Audenaert \textit{et al}.~\cite{audenaert_2007} (achievability), yielding
\begin{equation}\begin{aligned}
  \lim_{n \to \infty} - \frac1n \log p_{\rm err} (\rho^{\otimes n}, \sigma^{\otimes n} \pbar p, q) = \xi(\rho \| \sigma) \coloneqq - \log \inf_{0\leq s\leq1} \Tr \rho^{1-s} \sigma^{s} \qquad \forall p, q =  1-p \in (0,1).
\end{aligned}\end{equation}
This is commonly known as the quantum Chernoff bound, with the Chernoff divergence $\xi(\rho\|\sigma)$ constituting another important measure of distinguishability between quantum states.

Due to the intricate mathematical methods needed to establish the quantum Stein's lemma and the quantum Chernoff bound, the results have found few generalisations. There are, for instance, important cases where the alternative hypothesis is not given simply by i.i.d.\ quantum states $(\sigma^{\otimes n})_n$, but rather by a sequence $(\sigma_n)_n$ of states taken from some subset of density matrices --- this is known as composite hypothesis testing. A generalised quantum Stein's lemma in this setting was proposed~\cite{brandao_2010-1}, but so far no valid proof of such a relation has been found~\cite{berta_2022} beyond some special cases~\cite{tomamichel_2018,berta_2021} or related but incomparable composite settings~\cite{brandao_2020,berta_2021,mosonyi_2021,berta_2022}.

Perhaps even more challenging is the case of quantum channel discrimination. This task is inherently much more complex than state discrimination: any manipulation of $n$ copies of a quantum state $\rho$ is fully equivalent to having access to the tensor product $\rho^{\otimes n}$, while being able to invoke $n$~uses of  a channel $\mathcal{M}$  is not the same as manipulating $n$~parallel copies $\M^{\otimes n}$, since schemes for channel discrimination much more general than na\"ive parallel protocols can be employed. A common example are the so-called adaptive protocols, which work by processing the channel $n$~times, such that information obtained in previous rounds can be used to enhance the successive rounds of processing~\cite{chiribella_2008-1}. It is known that adaptive discrimination of channels can yield strict improvements over parallel protocols, although a very recent line of work~\cite{wilde_2020,wang_2019-4,fang_2020-2} established that the advantage disappears asymptotically in asymmetric hypothesis testing, where the (weak converse) error exponent is given by the regularised channel relative entropy:
\begin{equation}\begin{aligned}
 \lim_{\ve\to0} \liminf_{n\to\infty} - \frac1n \log \beta_\ve(\M^{\otimes n},\N^{\otimes n}) = D^\infty(\M \| \N) \coloneqq& \lim_{n\to\infty} \frac1n \sup_{\rho} D(\idc \otimes \M^{\otimes n}(\rho) \| \idc \otimes \N^{\otimes n}(\rho))\\
=&{} \sup_{\rho,\omega} \big[ D(\idc \otimes \M(\rho) \| \idc \otimes \N(\omega)) - D(\rho\|\omega) \big].
\end{aligned}\end{equation}
In symmetric hypothesis testing, it is known that adaptive strategies \emph{can} improve over parallel ones both non-asymptotically and asymptotically~\cite{harrow_2010-1,salek_2022}, although in general no exact form of the asymptotic error exponent --- and thus no extension of the Chernoff bound --- is known in the setting of quantum channels~\cite{yu_2021}.
The landscape of channel discrimination is further complicated by the fact that there exist even more general ways of manipulating quantum channels than adaptive protocols, for instance ones without a definite causal order structure~\cite{chiribella_2013,oreshkov_2012}, and they can yield strict improvements in some settings~\cite{chiribella_2013,ebler_2018,quintino_2019,bavaresco_2021}; the asymptotic performance of such strategies, however, has not been characterised yet. 

We note that all definitions of this section reduce to their corresponding classical counterparts when the two states $\rho$ and $\sigma$ commute (corresponding to probability distributions) and when the channels $\M$ and $\N$ are classical~\cite{blahut_1974,hayashi_2009}.
%%%%%%%%%%%%%%%%%%%%%%%%%%%%%%%%%%%%%%%%%%%%%%%%%%%%%%%%%%%%%%%%%%

\subsection{Setting}

In this work, we propose a variant of quantum hypothesis testing that adds an additional, `inconclusive' measurement outcome, allowing us to abstain from discriminating the states %after performing the measurement 
if we are not sufficiently confident in the result. This can be motivated, for instance, by an experimental situation in which a `no detection' event occurs, that is, a detector responsible for distinguishing two states fails to recognise either of them.
More generally, a motivation to consider such a setting could be any situation where avoiding a wrong result is of utmost importance, even if this means having to repeat the trial many times, such as would be the case e.g.\ with medical diagnoses.

To be specific, let us consider the case of quantum state discrimination. Given either of the two states $\rho$ or $\sigma$, we perform a \emph{three}-outcome measurement $M = \{ M_1, M_2, M_? \}$. The first two outcomes are as before --- outcome 1 corresponds to the guess `$\rho$' and outcome 2 corresponds to the guess `$\sigma$' --- while the third outcome `?' is designated as an inconclusive result, that is, a case in which we are unable to decisively distinguish the hypotheses. There is therefore some probability that our attempt to perform hypothesis testing fails: this probability is $\Tr M_? \rho$ when measuring $\rho$, and $\Tr M_? \sigma$ when measuring $\sigma$.

Such an approach dates back to the early works of Ivanovic~\cite{ivanovic_1987}, Dieks~\cite{dieks_1988}, and Peres~\cite{peres_1988}, where it was shown that all linearly independent pure states can be discriminated with no error, but only through inconclusive protocols whose probability of success may be strictly less than one. Extensions of this idea were studied in a number of settings~\cite{chefles_1998, fiurasek_2003, rudolph_2003, croke_2006, herzog_2009, bagan_2012, zhuang_2020, barnett_2009, bae_2015}, but mostly restricted to non-asymptotic, symmetric state discrimination. A related class of problems was also studied in classical hypothesis testing~\cite{nikulin_1986,gutman_1989}, although the way that discrimination errors and the consequent rates were defined there means that the setting is not directly comparable with ours.

Here we establish a complete and rigorous formalisation of inconclusive state discrimination schemes in quantum hypothesis testing. Our approach is based on characterising the corresponding errors conditioned on a conclusive result; that is, assuming that either of the outcomes `1' or `2' was obtained, how likely is it that we made an incorrect guess? This can be understood as postselecting on a conclusive outcome of the measurement and ignoring all measurement results where the outcome `?' is obtained. Depending on whether discussing the asymmetric or symmetric setting, we are concerned with the error probability when measuring a particular state $\rho$ or $\sigma$, or the average probability when measuring either of them.

To clarify how our setting differs from the conventional one of conclusive hypothesis testing, let us first consider the case of asymmetric error. Observe that the conventional type I error, for instance, is defined as $\alpha(M) = P(\text{error}\pbar\rho,M)$, i.e.\ the probability of guessing incorrectly given that the measured state is $\rho$. In our approach here, we further condition the incorrect guess event on the attainment of a conclusive outcome.
We thus define the following counterparts of the conventional hypothesis testing errors: the \emph{conditional type I error} and \emph{conditional type II error} probabilities are given by
\begin{equation}\begin{aligned}
  \calpha(M) &\coloneqq P(\text{error}\pbar\rho,M,\text{conclusive}) = \frac{\alpha(M)}{1-\Tr M_? \rho}= \frac{\Tr M_2 \rho}{\Tr \left[(M_1+M_2) \rho\right]} , \\
  \cbeta(M) &\coloneqq P(\text{error}\pbar\sigma,M,\text{conclusive}) = \frac{\beta(M)}{1- \Tr M_? \sigma} = \frac{\Tr M_1 \sigma}{\Tr \left[(M_1+M_2) \sigma\right]} ,
\end{aligned}\end{equation}
where `conclusive' denotes the event of a conclusive measurement outcome. The minimal type II error probability is then
\begin{equation}\begin{aligned}
  \cbeta_\ve (\rho, \sigma) \coloneqq&\, \inf_{M \in \M_3} \lset \cbeta(M) \bar \calpha(M) \leq \ve \rset\\
  =&\! \inf_{M_1,M_2 \geq 0} \lset \frac{\Tr M_1 \sigma}{\Tr \left[(M_1+M_2) \sigma\right]} \bar M_1 + M_2 \leq \id,\; \frac{\Tr M_2 \rho}{\Tr \left[(M_1+M_2) \rho\right]} \leq \ve \rset,
\end{aligned}\end{equation}
where $\mathcal{M}_3$ denotes the set of all three-outcome measurements of the form $\{M_1,M_2, M_?\}$, and to avoid pathological cases the optimisation is implicitly restricted to be over operators such that $\Tr \left[(M_1+M_2) \sigma\right], \Tr \left[(M_1+M_2) \rho\right] > 0$.
By construction, this error can never be larger than the conventional error $\beta_\ve(\rho,\sigma)$, 
since choosing $M_? = 0$ reduces to the usual definitions of $\alpha(M)$ and $\beta(M)$.

The analogous symmetric error is defined as follows. First, given two states, $\rho$ with prior probability $p$ and $\sigma$ with prior $q$, a measurement $M = \{ M_1, M_2, M_? \}$ has the following probability of obtaining an incorrect outcome:
\begin{equation}\begin{aligned}
  p_{\rm err}(\rho, \sigma \pbar p,q \pbar M) \coloneqq{}& P(\text{error} \pbar M) \\
  ={}& p \Tr M_2 \rho + q \Tr M_1 \sigma,
\end{aligned}\end{equation}
while the probability of obtaining a conclusive outcome is 
\begin{equation}\begin{aligned}
  p_{\rm conc}(\rho, \sigma \pbar p,q \pbar M) \coloneqq{}& P(\text{conclusive} \pbar M)\\
  ={}& 1 - \big( p \Tr M_? \rho + q \Tr M_? \sigma \big)\\
  ={}& p \Tr [(M_1 + M_2)\rho] + q \Tr [(M_1 + M_2)\sigma]\\
   ={} &  \Tr [(M_1 + M_2)(p\rho+  q \sigma)].
\end{aligned}\end{equation}
We then define the postselected symmetric probability of error $\cperr$ as the optimal conditional probability of error given a conclusive outcome --- in other words, the smallest ratio of incorrect outcomes to all conclusive outcomes on average:
\begin{equation}\begin{aligned}
    \cperr(\rho,\sigma \pbar p,q) &\coloneqq \inf_{M \in \M_3} \frac{p_{\rm err}(\rho, \sigma \pbar p,q \pbar M)}{p_{\rm conc}(\rho, \sigma \pbar p,q \pbar M)}\\
    &= \inf_{\substack{M_1, M_2 \geq 0,\\ M_1 + M_2 \leq \id}} \, \frac{p\, \Tr M_2 \rho + q\, \Tr M_1 \sigma}{ \Tr [(M_1 + M_2)(p\rho+  q \sigma)]}.
\end{aligned}\end{equation}

Before stating our main results, we need to introduce several quantities which we will find to characterise the hypothesis testing errors defined above.

Underlying all of them is the \emph{max-relative entropy}~\cite{datta_2009}, defined as
\begin{equation}\begin{aligned}\label{eq:dmax_definition}
  D_{\max} (\rho \| \sigma) &\coloneqq \log \inf \lset \lambda \in \RR_+ \bar \rho \leq \lambda \sigma \rset\\
  &\hphantom{:}= \begin{cases} \log \norm{\sigma^{-1/2} \rho \sigma^{-1/2}}{\infty} & \supp(\rho) \subseteq \supp(\sigma)\\
     \infty & \text{otherwise.} \end{cases}
\end{aligned}\end{equation}
Here, the inequality between operators is understood in terms of the positive semidefinite cone, i.e.\ $\rho \leq \lambda \sigma$ if and only if $\lambda \sigma - \rho$ is a positive semidefinite operator, and $\norm{\cdot}{\infty}$ denotes the operator norm (largest singular value). Eq.~\eqref{eq:dmax_definition} is typically defined for normalised quantum states $\rho$ and $\sigma$, although we will extend the same definition also to unnormalised positive semidefinite operators. This quantity can be understood as the sandwiched R\'enyi divergence of order $\infty$~\cite{muller-lennert_2013}; in the case where $\rho = P$ and $\sigma = Q$ are classical discrete probability distributions, this reduces to the $\infty$--R\'enyi divergence $D_{\max} (P\|Q) = D_{\infty}(P\|Q) = \log \sup_x P(x) / Q(x)$.

The max-relative entropy gives rise to two important quantities: the \deff{Hilbert projective metric}~\cite{bushell_1973}
\begin{align}
  \DD_\Omega(\rho \| \sigma) &\coloneqq D_{\max} (\rho \| \sigma) + D_{\max} (\sigma \| \rho)
  \label{eq:D_Omega}
\end{align}
and the \deff{Thompson metric}~\cite{thompson_1963}
\begin{align}
  \DD_\Thomp(\rho \| \sigma) &\coloneqq \max \big\{ D_{\max} (\rho \| \sigma),\, D_{\max} (\sigma \| \rho) \big\}.
\end{align}
The two quantities have enjoyed a number of applications in various branches of geometry and analysis~\cite{bushell_1973,nussbaum_2004}. Within quantum information theory, the Hilbert projective metric found operational use in characterising the convertibility of quantum states in different scenarios~\cite{reeb_2011,regula_2022,regula_2022-1}, often constituting the unique function that determines the existence of probabilistic transformations between quantum resources~\cite{regula_2021-4,regula_2022-1}; the Thompson metric, on the other hand, was recently identified with the rate of a task known as symmetric distinguishability dilution~\cite{salzmann_2021}, concerned with transformations of pairs of quantum states. To date, they have not been connected with rates of hypothesis testing tasks.

Throughout this work, we assume that the considered quantum systems are finite dimensional.

%%%%%%%%%%%%%%%%%%%%%%%%%%%%%%%%%%%%%%%%%%%%%%%%%%%%%%%%%%%%%%%%%%

\subsection{Results}

Our main results are the exact evaluation of the asymptotic error exponents in the setting of postselected hypothesis testing, both for asymmetric and symmetric discrimination, and both for quantum states and channels. The framework is shown to enjoy remarkably simplified properties compared to conventional hypothesis testing, allowing us to completely characterise the errors at the one-shot level as well as their asymptotic exponents. Notably, we show that the asymptotic exponents in asymmetric and symmetric postselected hypothesis testing are given exactly by the Hilbert and Thompson metrics, respectively, endowing the two quantities with operational meanings mirroring that of the quantum relative entropy and the Chernoff divergence. Our methods straightforwardly extend to the setting of composite hypothesis testing and channel discrimination, establishing results that are strictly stronger and more general than corresponding known results in conventional quantum hypothesis testing.

To our knowledge, this is the first time that the asymptotic setting of postselected hypothesis testing has been studied, even in the classical case; our results are therefore new also when the discrimination of quantum states is replaced with the discrimination of probability distributions (equivalently, commuting states).

We provide an overview of our main findings below.

\vspace{.5\baselineskip}\noindent\para{Asymmetric hypothesis testing} In Section~\ref{sec:asymmetric} we show that the asymptotic exponent of the conditional type II error in discriminating any two quantum states is given by the Hilbert projective metric $\DD_\Omega$ between them; specifically,
\begin{equation}\begin{aligned}
  \lim_{n\to\infty} - \frac1n \log \cbeta_{\ve}(\rho^{\otimes n},\sigma^{\otimes n}) = \DD_{\Omega} (\rho \| \sigma) \qquad \forall \ve \in (0,1).
\end{aligned}\end{equation}
This establishes an equivalent of quantum Stein's lemma in the setting of postselected hypothesis testing. 
We in fact establish a closed formula for the asymmetric hypothesis testing error probability already at the %level of one-shot error,
one-shot level, leading directly to the above asymptotic result. Specifically, our main result in Theorem~\ref{thm:asymmetric} shows that
\begin{equation}\begin{aligned}
  \cbeta_\ve (\rho,\sigma) = \left( \frac{\ve}{1-\ve} 2^{D_\Omega(\rho\|\sigma)} + 1 \right)^{-1} \qquad \forall \ve \in (0,1).
\end{aligned}\end{equation}

In Section~\ref{sec:composite}, we extend the above results to composite hypotheses, showing that when the alternative hypothesis $\sigma^{\otimes n}$ is replaced by a family of convex sets of density matrices $(\FF_n)_n$ that is closed under tensor product, then the asymptotic exponent is given by
\begin{equation}\begin{aligned}
  \lim_{n\to\infty} - \frac1n \log \cbeta_{\ve,\FF_n}(\rho^{\otimes n}) = \DD_{\Omega,\FF}^\infty (\rho) \coloneqq \lim_{n\to\infty} \frac1n \min_{\sigma_n \in \FF_n} \DD_\Omega(\rho^{\otimes n} \| \sigma_n ) \qquad \forall \ve \in (0,1),
\end{aligned}\end{equation}
where $\cbeta_{\ve,\FF_n}$ denotes the least conditional type II error in discriminating a state against \emph{all} states in the set $\FF_n$. 
The above is the %equivalent
analogue of the conjectured generalised quantum Stein's lemma of~\cite{brandao_2010-1}, which is not known to hold in the setting of conventional quantum hypothesis testing, but which we show to be valid in the postselected setting.

Analogous findings are also valid for the discrimination of quantum channels. We show in Section~\ref{sec:channels} that
\begin{equation}\begin{aligned}
  \lim_{n\to\infty} - \frac1n \log \cbeta_{\ve}(\M^{\otimes n}, \N^{\otimes n}) = \lim_{n\to\infty} - \frac1n \log \cbetaAny(\M, \N \pbar n) = \DD_{\Omega} (\M \| \N) \qquad \forall \ve \in (0,1),
\end{aligned}\end{equation}
where $\cbetaAny(\M, \N \pbar n)$ denotes the error exponent obtained using the most general channel discrimination schemes that use $n$ channel evaluations, and where $\DD_{\Omega}(\M\|\N)$ is given by the Hilbert projective metric between the Choi operators of the two channels, which does not require any regularisation and can be readily evaluated.
Importantly, our results show that the performance of channel hypothesis testing is the same regardless of whether the channels are used in parallel or discriminated using broader schemes such as adaptive protocols. This equivalence holds both asymptotically and in the finite-copy regime.

\vspace{.5\baselineskip}\noindent\para{Symmetric hypothesis testing} In Section~\ref{sec:symmetric} we establish that the asymptotic exponent of the postselected symmetric error in distinguishing any two states equals their Thompson metric:
\begin{equation}\begin{aligned}
  \lim_{n\to\infty} - \frac1n \log \cperr(\rho^{\otimes n}, \sigma^{\otimes n} \pbar p, q) = \DD_{\Thomp} (\rho \| \sigma) \qquad \forall p, q = 1-p \in (0,1).
\end{aligned}\end{equation}
This is the equivalent of the Chernoff bound in the setting of postselected quantum hypothesis testing and gives a direct operational interpretation of the Thompson metric in state discrimination. 
As in the asymmetric case, the result is based on an exact quantification of the one-shot symmetric error probability, which we show in Theorem~\ref{thm:symmetric} to be
\begin{equation}\begin{aligned}
    \cperr(\rho,\sigma \pbar p,q) = \big( 2^{D_\Thomp (p \rho \| q \sigma )} + 1 \big)^{-1}.
\end{aligned}\end{equation}
The above is generalised to quantum channels in Section~\ref{sec:channels}, where we show that
\begin{equation}\begin{aligned}
  \lim_{n\to\infty} - \frac1n \log \cperr(\M^{\otimes n}, \N^{\otimes n} \pbar p, q) = \lim_{n\to\infty} - \frac1n \log \cperrAny(\M, \N \pbar p, q \pbar n) = \DD_{\Thomp} (\M \| \N), %\qquad \forall p, q = 1-p \in (0,1)
\end{aligned}\end{equation}
again demonstrating that no advantage is gained by employing protocols more general than parallel ones.

\vspace{.5\baselineskip}Section~\ref{sec:gpts} discusses how the above results in asymmetric and symmetric state discrimination are actually valid in physical theories broader than quantum mechanics, namely, settings under the umbrella of general probabilistic theories. This is, to the best of our knowledge, the first time that asymptotic error exponents in hypothesis testing have been computed in such general theories --- previous such results made explicit use of the properties of quantum mechanics or classical probability theory.

%%%%%%%%%%%%%%%%%%%%%%%%%%%%%%%%%%%%%%%%%%%%%%%%%%%%%%%%%%%%%%%%%%
%%%%%%%%%%%%%%%%%%%%%%%%%%%%%%%%%%%%%%%%%%%%%%%%%%%%%%%%%%%%%%%%%%
%%%%%%%%%%%%%%%%%%%%%%%%%%%%%%%%%%%%%%%%%%%%%%%%%%%%%%%%%%%%%%%%%%

\section{Asymmetric case}\label{sec:asymmetric}

%%%%%%%%%%%%%%%%%%%%%%%%%%%%%%%%%%%%%%%%%%%%%%%%%%%%%%%%%%%%%%%%%%

\subsection{One-shot discrimination error}\label{sec:asymmetric_oneshot}

Let us recall the definition of the minimal postselected type II error,
\begin{equation}\begin{aligned}\label{eq:DH_def}
  \cbeta_\ve (\rho, \sigma) = \inf_{M_1,M_2 \geq 0} \lset \frac{\Tr M_1 \sigma}{\Tr \left[(M_1+M_2) \sigma\right]} \bar M_1 + M_2 \leq \id,\; \frac{\Tr M_2 \rho}{\Tr \left[(M_1+M_2) \rho\right]} \leq \ve \rset.
\end{aligned}\end{equation}
Our first result is to connect it with the Hilbert projective metric. We will frequently encounter the non-logarithmic variant of the latter, for which we introduce the notation
\begin{equation}\begin{aligned}
  \Omega(\rho \| \sigma) \coloneqq 2^{\DD_\Omega(\rho\|\sigma)}.
\end{aligned}\end{equation}
%%%%%%
%\clearpage
%%%%%%
\begin{boxed}{white}
\begin{theorem}\label{thm:asymmetric}
For all $\ve \in (0,1)$, we have
\begin{equation}\begin{aligned}
    %D^\ve_{H,\Omega} (\rho \| \sigma) = \log \left( \frac{\ve}{1-\ve} \Omega(\rho\|\sigma) + 1 \right).
    \cbeta_\ve (\rho,\sigma) = \left( \frac{\ve}{1-\ve}\, \Omega(\rho\|\sigma) + 1 \right)^{-1}.
\end{aligned}\end{equation}
\end{theorem}
\end{boxed}
\begin{proof}
We begin by noticing that the constraint $M_1 + M_2 \leq \id$ in~\eqref{eq:DH_def} is superfluous: for all $M_1, M_2 \geq 0$ that satisfy $\frac{ \Tr M_2 \rho }{ \Tr[(M_1+M_2)\rho] } \leq \ve$, one can always rescale them as $M'_i \coloneqq M_i / \norm{M_1 + M_2}{\infty}$ and they are then feasible for the original program with the same value of the objective function. Thus,%
\footnote{A careful reader might wonder whether there is no danger of dividing by zero in this rewriting. That is not the case, as we can always constrain the optimisation to be over $M_1$ such that $\Tr M_1 \sigma > 0$. This can be seen by noting that, for every $M_1$ such that $\Tr M_1 \sigma = 0$, $M_1 + \delta\id$ is also feasible for all $\delta > 0$, and the error value of the original $M_1$ is recovered in the limit $\delta \to 0$.}
\begin{equation}\begin{aligned}\label{eq:DH_def2}
    \cbeta_\ve(\rho,\sigma)^{-1} = \sup_{M_1,M_2 \geq 0} \lset \frac{ \Tr[(M_1+M_2)\sigma] }{ \Tr M_1 \sigma } \bar \frac{ \Tr M_2 \rho }{ \Tr[(M_1+M_2)\rho] } \leq \ve \rset.
\end{aligned}\end{equation}
In order to relate this optimisation with the Hilbert projective metric, we will use the dual form of the optimisation problem $\Omega(\rho\|\sigma)$. Writing
\begin{equation}\begin{aligned}
\Omega(\rho\|\sigma) = \Omega(\sigma\|\rho) = \inf \lset \gamma \in \RR_+ \bar \sigma \leq \lambda \rho \leq \gamma \sigma, \, \lambda \in \RR_+ \rset,
\end{aligned}\end{equation}
the dual optimisation problem can be obtained as~\cite{regula_2021-4}
\begin{equation}\begin{aligned}\label{eq:omega_dual}
    \Omega(\rho\|\sigma) = \sup \lset \frac{\Tr A \sigma}{\Tr B \sigma} \bar A,B \geq 0,\; \frac{\Tr A \rho}{\Tr B \rho} \leq 1 \rset.
\end{aligned}\end{equation}
The similarity between the two problems can be noticed by rewriting the objective function of~\eqref{eq:DH_def2} as
\begin{equation}\begin{aligned}
  \frac{ \Tr[(M_1+M_2)\sigma] }{ \Tr M_1 \sigma } = 1 + \frac{ \Tr M_2\sigma }{ \Tr M_1 \sigma }.
\end{aligned}\end{equation}
To make the relation explicit, let $M_1, M_2 \geq 0$ be arbitrary operators feasible for~\eqref{eq:DH_def2}, and define
\begin{equation}\begin{aligned}
    A &\coloneqq (1-\ve) M_2,\\
    B &\coloneqq \ve M_1.
\end{aligned}\end{equation}
Notice now that
\begin{equation}\begin{aligned}
    \frac{\Tr A \rho}{\Tr B \rho} = \frac{(1-\ve) \Tr M_2 \rho}{\ve \Tr M_1 \rho} \leq 1 ,
\end{aligned}\end{equation}
using the fact that
\begin{equation}\begin{aligned}
  \frac{ \Tr M_2 \rho }{ \Tr[(M_1+M_2)\rho] } \leq \ve \iff \frac{ \Tr M_1\rho }{ \Tr M_2 \rho } \geq \frac{1-\ve}{\ve}.
\end{aligned}\end{equation}
$A$ and $B$ are therefore feasible operators for the dual optimisation problem for $\Omega(\rho\|\sigma)$ in~\eqref{eq:omega_dual}, giving
\begin{equation}\begin{aligned}\label{eq:omega_sdp}
    \Omega(\rho \| \sigma) \geq \frac{1-\ve}{\ve} \frac{ \Tr M_2\sigma }{ \Tr M_1 \sigma } = \frac{1-\ve}{\ve} \left(  \frac{ \Tr[(M_1+M_2)\sigma] }{ \Tr M_1 \sigma}  - 1\right).
\end{aligned}\end{equation}

Working backwards through the above reasoning, we can analogously show that any feasible solution $\{A,B\}$ for $\Omega(\rho \| \sigma)$ in~\eqref{eq:omega_dual} gives a feasible solution $\{M_1,M_2\}$ for $\cbeta_\ve(\rho,\sigma)^{-1}$ in~\eqref{eq:DH_def2}. The two problems are therefore equivalent, in the sense that
\begin{equation}\begin{aligned}
    \Omega(\rho \| \sigma) = \frac{1-\ve}{\ve} \left( \cbeta_\ve(\rho ,\sigma)^{-1} - 1\right),
\end{aligned}\end{equation}
which is precisely the statement of the theorem.
\end{proof}

\para{Achievability} As long as $\cbeta_\ve(\rho,\sigma) \neq 0$, i.e.\ $\Omega(\rho\|\sigma) < \infty$, a measurement that achieves the optimal error $\cbeta_\ve(\rho,\sigma)$ can be readily constructed using the fact that $D_{\max} (\sigma \| \rho) = \log \norm{\rho^{-1/2} \sigma \rho^{-1/2}}{\infty}$.  Let $\ket{\psi_{\max}}$ and $\ket{\psi_{\min}}$ be eigenvectors corresponding, respectively, to the largest and the smallest non-zero eigenvalue of the operator $\rho^{-1/2} \sigma \rho^{-1/2}$. Then, define
\begin{equation}\begin{aligned}
  \wt M_1 \coloneqq \frac{\rho^{-1/2} \proj{\psi_{\min}} \rho^{-1/2}}{\ve},\qquad
  \wt M_2 \coloneqq \frac{\rho^{-1/2} \proj{\psi_{\max}}\rho^{-1/2}}{1 - \ve}.
\end{aligned}\end{equation}
Now we only need to rescale these operators into a valid POVM, that is, define the three-outcome measurement
\begin{equation}\begin{aligned}\label{eq:measurement_single_copy}
  M \coloneqq \left\{ \frac{\wt M_1}{\norm{\wt M_1 + \wt M_2}{\infty}},\, \frac{\wt M_2}{\norm{\wt M_1 + \wt M_2}{\infty}},\, \id - \frac{\wt M_1 + \wt M_2}{\norm{\wt M_1 + \wt M_2}{\infty}} \right\}.
\end{aligned}\end{equation}
It can be verified that $\calpha(M) = \ve$ and
\begin{equation}\begin{aligned}
  \cbeta(M) =   \left( 1 + \frac{\ve}{1-\ve} \norm{\rho^{-1/2} \sigma \rho^{-1/2}}{\infty} \norm{\sigma^{-1/2} \rho \sigma^{-1/2}}{\infty} \right)^{-1}  = \cbeta_\ve(\rho, \sigma),
\end{aligned}\end{equation}
meaning that the measurement is indeed optimal.

An interesting point for the discrimination of many-copy states, i.e.\ $\rho^{\otimes n}$ and $\sigma^{\otimes n}$, is as follows. An optimal measurement strategy can readily be constructed following the above recipe; remarkably, the operational procedure corresponding to the measurement one obtains in this way is a product strategy, i.e.\ it involves
separate measurements on each individual copy of the state, which is much simpler to implement from a practical and technological perspective. To see this, let us define a single-copy measurement similarly to~\eqref{eq:measurement_single_copy}, except that we replace $\wt M_1$ and $\wt M_2$ with
\begin{equation}\begin{aligned}
  \wt M_1' \coloneqq \frac{\rho^{-1/2} \proj{\psi_{\min}} \rho^{-1/2}}{\ve^{1/n}},\qquad
  \wt M_2' \coloneqq \frac{\rho^{-1/2} \proj{\psi_{\max}}\rho^{-1/2}}{(1 - \ve)^{1/n}}.
\end{aligned}\end{equation}
We then measure each individual copy of $\rho$ or $\sigma$. If the outcome `1' is obtained $n$ times, we guess $\rho^{\otimes n}$; if the outcome `2' is obtained $n$ times, we guess $\sigma^{\otimes n}$; any other combination of outcomes is deemed inconclusive. Conditioned on a conclusive result, the probability of success of such a strategy is clearly the same as if we used the measurement~\eqref{eq:measurement_single_copy} directly.

%%%%%%%%%%%%%%%%%%%%%%%%%%%%%%%%%%%%%%%%%%%%%%%%%%%%%%%%%%%%%%%%%%

\subsection{Asymptotic error exponent}\label{sec:asymmetric_exp}

The extension of Theorem~\ref{thm:asymmetric} to the asymptotic setting is almost immediate. The crucial property of the Hilbert projective metric, and indeed also the Thompson metric, is that they are additive quantities: for all states $\rho$ and $\sigma$,
\begin{equation}\begin{aligned}
  \DD_\Omega(\rho^{\otimes n} \| \sigma^{\otimes n}) = n \, \DD_\Omega(\rho \| \sigma),\\
  \DD_\Thomp(\rho^{\otimes n} \| \sigma^{\otimes n}) = n \, \DD_\Thomp(\rho \| \sigma).
\end{aligned}\end{equation}
This follows from the additivity of $D_{\max}$, which in turn follows directly from the definition and the fact that
\begin{equation}\begin{aligned}
  \norm{\left(\sigma^{\otimes n}\right)^{-1/2} \rho^{\otimes n} \left(\sigma^{\otimes n}\right)^{-1/2}}{\infty} = \norm{\left(\sigma^{-1/2} \rho \sigma^{-1/2}\right)^{\otimes n}}{\infty} = \norm{\sigma^{-1/2} \rho \sigma^{-1/2}}{\infty}^{\, n}.
\end{aligned}\end{equation}

Using this additivity and the %general
two-sided inequality
\begin{equation}\begin{aligned}
    \frac{\ve}{1-\ve}\, \Omega(\rho\|\sigma) \leq \cbeta_\ve (\rho,\sigma)^{-1} \leq \frac{1}{1-\ve}\, \Omega(\rho\|\sigma) ,
    \label{eq:two-sided}
\end{aligned}\end{equation}
itself a consequence of the fact that $\Omega(\rho\|\sigma)\geq 1$ for all $\rho,\sigma$, 
we obtain
\begin{equation}\begin{aligned}\label{eq:stein_secondorder}
   n \DD_\Omega(\rho\|\sigma) + \log\frac{\ve}{1-\ve} \leq  -\log \cbeta_\ve (\rho^{\otimes n} \| \sigma^{\otimes n}) \leq n \DD_\Omega(\rho\|\sigma) + \log\frac{1}{1-\ve}.
\end{aligned}\end{equation}
This immediately yields an expression for the asymptotic error exponent, which is the equivalent of quantum Stein's lemma in the setting of this work.

\begin{boxed}{white}
\begin{corollary}[Asymptotic error exponent of asymmetric postselected hypothesis testing]\label{cor:asymmetric_exp}
$\,$\\For all $\ve \in (0,1)$, we have that
\begin{equation}\begin{aligned}
  \lim_{n\to\infty} - \frac1n \log \cbeta_{\ve}(\rho^{\otimes n},\sigma^{\otimes n}) = \DD_{\Omega} (\rho \| \sigma).
\end{aligned}\end{equation}
\end{corollary}
\end{boxed}

What is noteworthy about the expression in~\eqref{eq:stein_secondorder} is that there are no lower-order terms in $n$: the convergence to the asymptotic limit of Corollary~\ref{cor:asymmetric_exp} 
%does not depend on any terms other than ones linear in $n$, 
happens at least as fast as $O(1/n)$, 
which strongly contrasts with the case encountered in conventional quantum hypothesis testing~\cite{li_2014,tomamichel_2013}.
%%%%%%%%%%%%%%%%%%%%%%%%%%%%%%%%%%%%%%%%%%%%%%%%%%%%%%%%%%%%%%%%%%

\subsection{Properties of the asymmetric error}\label{sec:asymmetric_prop}

As the asymptotic error exponent is given exactly by $\DD_\Omega$, while the one-shot error $\cbeta_\ve$ is an inversely monotonic function of it, it suffices to characterise the properties of $\DD_\Omega$ to fully understand the behaviour of the errors in both the one-shot and exponential settings.

\begin{proposition}\label{prop:asymmetric_prop}
For all states $\rho$ and $\sigma$, and for all $\ve \in (0,1)$:
\begin{enumerate}[(i)]
\item $\DD_\Omega(\rho \| \sigma) \geq 0$, with equality if and only if $\rho = \sigma$. Consequently, $\cbeta_\ve(\rho,\sigma) \leq 1-\ve$, with equality if and only if $\rho=\sigma$.

\item Both $\DD_\Omega(\rho\|\sigma)$ and $\cbeta_\ve(\rho,\sigma)$ are symmetric under the exchange of $\rho$ and $\sigma$, i.e.\ $\DD_\Omega(\rho\|\sigma) = \DD_\Omega(\sigma\|\rho)$ and $\cbeta_\ve(\rho,\sigma) = \cbeta_\ve(\sigma,\rho)$.

\item For all positive numbers $\lambda,\mu$, we have $\DD_\Omega(\lambda \rho \| \mu \sigma) = \DD_\Omega(\rho \| \sigma)$, and analogously $\cbeta_\ve(\lambda\rho, \mu\sigma) = \cbeta_\ve(\rho,\sigma)$.

\item $\DD_\Omega(\rho \| \sigma) = \infty$ if and only if $\supp \rho \neq \supp \sigma$. Hence, $\cbeta_\ve (\rho,\sigma) = 0$  if and only if $\rho$ and $\sigma$ have non-identical supports.

\item $\DD_\Omega$ satisfies the data-processing inequality under all positive maps; specifically, for every positive linear map $\E$,
\begin{equation}\begin{aligned}
  \DD_\Omega\!\left(\left.\frac{\E(\rho)}{\Tr \E(\rho)} \right\| \frac{\E(\sigma)}{\Tr \E(\sigma)}\right) = \DD_\Omega(\E(\rho) \| \E(\sigma)) \leq \DD_{\Omega} (\rho \| \sigma).
\end{aligned}\end{equation}
Hence, $\cbeta_\ve \!\left(\frac{\E(\rho)}{\Tr \E(\rho)} ,\, \frac{\E(\sigma)}{\Tr \E(\sigma)}\right) \geq \cbeta_\ve (\rho, \sigma)$.

\end{enumerate}
\end{proposition}
\begin{proof}
(i)~holds because $\rho \leq \sigma$ and $\sigma \leq \rho$ can only be true if $\rho = \sigma$. (ii)~is an immediate consequence of the manifestly symmetric definition~\eqref{eq:D_Omega}, together with Theorem~\ref{thm:asymmetric}. (iii)~follows from the fact that $D_{\max}(\lambda \rho \| \mu \sigma) = \log\frac{\lambda}{\mu} + D_{\max}(\rho \| \sigma)$, which is immediate from the definition. (iv)~holds because $D_{\max}(\rho \| \sigma) < \infty$ if and only if $\supp \rho \subseteq \supp \sigma$. (v)~is a consequence of~(iii) coupled with the fact that for all $\lambda$ such that $\rho \leq \lambda \sigma$, the positivity of $\E$ ensures that $\lambda \E(\sigma) - \E(\rho) = \E(\lambda \sigma - \rho) \geq 0$, so $D_{\max}(\E(\rho) \| \E(\sigma)) \leq D_{\max}(\rho \| \sigma)$.
\end{proof}

\begin{remark}
    The symmetry shown in Proposition~\ref{prop:asymmetric_prop}(ii) is \textit{a priori} highly non-obvious, and indeed one of the remarkable consequences of Theorem~\ref{thm:asymmetric}. Indeed, we started by looking at the error in \emph{asymmetric} postselected hypothesis testing, but 
    the expression we obtained is \emph{symmetric} under the exchange of the two states. 
    This can be understood as a consequence of the fact that, for all states $\rho, \sigma$ with $\supp \rho = \supp \sigma$, there exists a completely positive map $\E$ such that $\E(\rho) = \sigma$ and $\E(\sigma) = k \rho$ for some $k \in (0,\infty)$~\cite[Theorem~21]{reeb_2011}. Applying the data-processing inequality of Proposition~\ref{prop:asymmetric_prop}(v) together with the scaling invariance~(iii), we see that we must have $\beta_\ve(\rho \| \sigma) = \beta_\ve(\sigma \| \rho)$.
\end{remark}

In all of the above, we have assumed that $\ve \in (0,1)$. The case $\ve = 0$ requires a separate treatment.

\begin{proposition}
It holds that
\begin{equation}\begin{aligned}
    \cbeta_0 (\rho ,\sigma) = %-\log \begin{cases} 1 & D_{\max}(\sigma\|\rho) < \infty, \\ 0 & \text{\rm otherwise} \end{cases} \;=
     \begin{cases} 1 & \supp \sigma \subseteq \supp \rho, \\ 0 & \text{\rm otherwise.} \end{cases}
\end{aligned}\end{equation}
\end{proposition}
\begin{proof}
If $\supp \sigma \subseteq \supp \rho$, then $\Tr M_2 \rho = 0$ implies that $\Tr M_2 \sigma = 0$. Any feasible pair $\{M_1,M_2\}$ such that $\Tr M_1 \sigma \neq 0$ then satisfies
\begin{equation}\begin{aligned}
   \frac {\Tr M_1 \sigma }{ \Tr[(M_1+M_2)\sigma] } = \frac {\Tr M_1 \sigma }{ \Tr M_1\sigma } = 1.
\end{aligned}\end{equation}

For the case $\supp \sigma \not\subseteq \supp \rho$, by hypothesis we know that there exists a pure state $\psi$ such that $\Tr \psi \sigma > 0$ but $\Tr \psi \rho = 0$. Pick $M_1 = \delta \id, M_2 = (1-\delta) \psi$ for some $\delta \in (0,1)$ to conclude that any value of
\begin{equation}\begin{aligned}
   \frac {\Tr M_1 \sigma }{ \Tr[(M_1+M_2)\sigma] } = \frac{\delta}{\delta + (1-\delta) \Tr \psi \sigma}
\end{aligned}\end{equation}
is achievable. Taking the limit $\delta \to 0$ gives the stated result.
\end{proof}
It is a curious fact that $\cbeta_\ve (\rho , \sigma)$ is symmetric in its arguments for $\ve > 0$ (Proposition~\ref{prop:asymmetric_prop}(ii)), but no longer so for $\ve = 0$.
This asymmetry is due to the situation where $\supp \sigma$ is strictly contained in $\supp \rho$. In this case, $\rho$ can be distinguished from $\sigma$, but not the other way around: consider the POVM $\{\id - \Pi_\sigma,0,\Pi_\sigma\}$ which allows us to conclude, whenever we obtain a conclusive outcome, that the state in our possession is $\rho$; however, the probability of a conclusive outcome when measuring $\sigma$ is zero. This problem is avoided for $\ve > 0$.

%%%%%%%%%%%%%%%%%%%%%%%%%%%%%%%%%%%%%%%%%%%%%%%%%%%%%%%%%%%%%%%%%%

\subsection{Composite postselected hypothesis testing}\label{sec:composite}

In the setting of composite hypothesis testing, we wish to distinguish between a given state $\rho$ and a whole set of quantum states, denoted $\FF$. Let us first consider the case of conventional (non-postselected) hypothesis testing. Here, for a given measurement $M$, the performance 
is quantified with the worst-case type II error, i.e.
\begin{equation}\begin{aligned}
  \beta_{\FF} (M) \coloneqq \sup_{\sigma \in \FF} \,\beta(M) = \sup_{\sigma \in \FF} \,\Tr M_1 \sigma,
\end{aligned}\end{equation}
and the optimised error is defined as $\beta_{\ve,\FF}(\rho) \coloneqq \inf \lset \beta_\FF(M) \bar \alpha(M) \leq \ve \rset$.
Brand\~ao and Plenio~\cite{brandao_2010-1}  conjectured a remarkable extension of the quantum Stein's lemma to composite hypothesis testing: namely, that for every family of sets $(\FF_n)_n$ satisfying a small number of regularity properties (including convexity and closedness under tensor product, partial trace, and permutations of systems), it holds that
\begin{equation}\begin{aligned}\label{eq:gen_steins}
  \lim_{n\to\infty} - \frac1n \log \beta_{\ve, \FF_n} (\rho^{\otimes n}) \texteq{?} D_{\FF}^\infty(\rho),
\end{aligned}\end{equation}
where $D_{\FF}^\infty$ is the regularised relative entropy, $D_{\FF}^\infty(\rho) \coloneqq \lim_{n\to\infty} \frac1n \min_{\sigma_n \in \FF_n} D(\rho^{\otimes n}\|\sigma_n)$. Although a proof of this result was claimed in~\cite{brandao_2010-1}, issues with the derivation were recently uncovered~\cite{berta_2022}, and~\eqref{eq:gen_steins} is therefore not known to hold in general.

We show that an %equivalent
analogue of the above result holds true in the setting of postselected hypothesis testing. To this end, define 
\begin{equation}\begin{aligned}
  \cbeta_{\ve,\FF}(\rho) \coloneqq& \inf_{M \in \M_3} \lset \cbeta_\FF(M) \bar \calpha(M) \leq \ve \rset
  \\=& \inf_{M_1,M_2 \geq 0} \lset \sup_{\sigma \in \FF} \frac{\Tr M_1 \sigma}{\Tr \left[(M_1+M_2) \sigma\right]} \bar M_1 + M_2 \leq \id,\; \frac{\Tr M_2 \rho}{\Tr \left[(M_1+M_2) \rho\right]} \leq \ve \rset.
\end{aligned}\end{equation}
We then show that the asymptotic error exponent of composite postselected hypothesis testing is given exactly by the regularisation of the Hilbert projective metric $\DD_\Omega(\rho\|\sigma)$ optimised over the sets~$\FF_n$.

\begin{boxed}{white}
\begin{theorem}\mbox{}\label{thm:postselected_steins}
For every convex and closed set of quantum states $\FF$, 
\begin{equation}\begin{aligned}
  \cbeta_{\ve,\FF} (\rho) = \left(\frac{\ve}{1-\ve} \min_{\sigma \in \FF} \Omega(\rho\|\sigma)  + 1\right)^{-1}.
\end{aligned}\end{equation}

As a consequence, for all  $\ve \in (0,1)$ and for all families of sets $(\FF_n)_n$ that are convex and closed, we have that
\begin{equation}\begin{aligned}\label{eq:generalised_asymmetric1}
  \liminf_{n\to\infty} - \frac1n \log \cbeta_{\ve,\FF_n}(\rho^{\otimes n}) = \DD_{\Omega,\FF}^\infty (\rho),
\end{aligned}\end{equation}
where
\begin{equation}\begin{aligned}\label{eq:generalised_asymmetric2}
  \DD_{\Omega,\FF}^\infty(\rho) = \liminf_{n\to\infty} \frac1n \min_{\sigma_n \in \FF_n} \DD_{\Omega} (\rho^{\otimes n} \| \sigma_n).
\end{aligned}\end{equation}
If the sets $(\FF_n)_n$ are closed under tensor product, i.e.\ if $\sigma_n \otimes \sigma_m \in \FF_{n+m}$ for all $\sigma_n \in \FF_n$ and $\sigma_m \in \FF_m$, then $\liminf$ in~\eqref{eq:generalised_asymmetric1}--\eqref{eq:generalised_asymmetric2} can be replaced with $\lim$, because the limits are achieved.
\end{theorem}
\end{boxed}
We remark here that the quantity $\min_{\sigma \in \FF} \Omega(\rho\|\sigma)$ can be computed as the optimal value of a conic linear optimisation problem~\cite{regula_2021-4}; when membership in the set $\FF$ can be expressed using linear matrix inequalities, this reduces to a semidefinite program.
\begin{proof}
The basic idea is exactly the same as in the proof of Theorem~\ref{thm:asymmetric}. Let us first consider the one-shot case. We have that
\begin{equation}\begin{aligned}
\cbeta_{\ve,\FF} (\rho) &\texteq{(i)} \inf_{M_1,M_2 \geq 0} \lset \sup_{\sigma \in \FF} \frac{\Tr M_1 \sigma}{\Tr \left[(M_1+M_2) \sigma\right]} \bar \frac{\Tr M_2 \rho}{\Tr \left[(M_1+M_2) \rho\right]} \leq \ve \rset\\
&= \inf_{\substack{M_1,M_2 \geq 0\\t > 0}} \lsetr t \barr \frac{\Tr M_1 \sigma}{\Tr \left[(M_1+M_2) \sigma\right]} \leq t \; \forall \sigma \in \FF,\; \frac{\Tr M_2 \rho}{\Tr \left[(M_1+M_2) \rho\right]} \leq \ve\rsetr\\
&\texteq{(ii)} \inf_{\substack{M_1,M_2 \geq 0\\t > 0}} \lsetr t \barr \frac{\Tr M_2 \sigma}{\Tr M_1 \sigma} \geq \frac1t - 1 \; \forall \sigma \in \FF,\; \frac{\Tr M_1 \rho}{\Tr M_2 \rho}  \geq \frac1\ve - 1 \rsetr\\
&\texteq{(iii)} \inf_{\substack{M_1,M'_2 \geq 0\\t > 0}} \lsetr t \barr \frac{\Tr M'_2 \sigma}{\Tr M_1 \sigma} \geq 1 \; \forall \sigma \in \FF,\; \frac{\Tr M_1 \rho}{\Tr M'_2 \rho} \geq \frac{1-\ve}{\ve} \left(\frac1t - 1\right) \rsetr\\
&\texteq{(iv)} \inf_{\substack{M_1,M'_2 \geq 0\\t' > 0}} \lsetr \frac{1}{t'} \barr \frac{\Tr M_1 \sigma}{\Tr M'_2 \sigma} \leq 1 \; \forall \sigma \in \FF,\; \frac{\Tr M_1 \rho}{\Tr M'_2 \rho} \geq \frac{1-\ve}{\ve} \left(t' - 1\right)\rsetr,
\end{aligned}\end{equation}
where: (i) we removed the superfluous constraint $M_1 + M_2 \leq \id$ as before; (ii) we rewrote $\frac{\Tr M_1 \sigma}{\Tr \left[(M_1+M_2) \sigma\right]} = \left( 1 + \frac{\Tr M_2 \rho}{\Tr M_1 \sigma}\right)^{-1}$ and analogously for $\rho$; (iii) we made the change of variables $M'_2 \coloneqq \frac{t}{1-t} M_2$; (iv) we made the change of variables $t' \coloneqq \frac1t$. 
Thus, 
\begin{equation}\begin{aligned}
\cbeta_{\ve,\FF} (\rho)^{-1} &= \sup_{\substack{M_1,M'_2 \geq 0\\t' > 0}} \lsetr t' \barr \frac{\Tr M_1 \sigma}{\Tr M'_2 \sigma} \leq 1 \; \forall \sigma \in \FF,\; \frac{\Tr M_1 \rho}{\Tr M'_2 \rho} \geq \frac{1-\ve}{\ve} \left(t' - 1\right)\rsetr\\
&= \sup_{M_1,M'_2 \geq 0} \lsetr \frac{\ve}{1-\ve}\frac{\Tr M_1 \rho}{\Tr M'_2 \rho} + 1  \barr \frac{\Tr M_1 \sigma}{\Tr M'_2 \sigma} \leq 1 \; \forall \sigma \in \FF\rsetr.
\end{aligned}\end{equation}
Now, for every convex and closed set $\FF$, we have~\cite[Theorem~1]{regula_2021-4}
\begin{equation}\begin{aligned}
  \min_{\sigma \in \FF} \Omega(\rho\|\sigma) &=  \sup_{A, B\geq 0} \lsetr \frac{\Tr A \rho}{\Tr B \rho}  \barr \frac{\Tr A \sigma}{\Tr B \sigma} \leq 1 \; \forall \sigma \in \FF\rsetr,
\end{aligned}\end{equation}
from which we immediately get that
\begin{equation}\begin{aligned}
  \cbeta_{\ve,\FF} (\rho)^{-1} = \frac{\ve}{1-\ve} \min_{\sigma \in \FF} \Omega(\rho\|\sigma)  + 1.
\end{aligned}\end{equation}
For the $n$-copy case, this means that
\begin{equation}\begin{aligned}
  \cbeta_{\ve,\FF_n} (\rho^{\otimes n})^{-1} = \frac{\ve}{1-\ve} \min_{\sigma_n \in \FF_n} \Omega(\rho^{\otimes n}\|\sigma_n)  + 1,
\end{aligned}\end{equation}
implying, as in~\eqref{eq:two-sided}, that
\begin{equation}\begin{aligned}
    \frac{\ve}{1-\ve} \min_{\sigma_n \in \FF_n} \Omega(\rho^{\otimes n}\|\sigma_n) \leq \cbeta_{\ve,\FF_n} (\rho^{\otimes n})^{-1} \leq \frac{1}{1-\ve} \min_{\sigma_n \in \FF_n} \Omega(\rho^{\otimes n}\|\sigma_n) .
\end{aligned}\end{equation}
The result follows by taking the logarithm of the above, dividing by $n$, and taking the $\liminf$.

When the sets $\FF_n$ are closed under tensor product, then $\DD_{\Omega,\FF_{m+n}}(\rho^{\otimes m+n}) \leq \DD_{\Omega,\FF_m}(\rho^{\otimes m}) + \DD_{\Omega,\FF_n}(\rho^{\otimes n})$~\cite[Theorem~1]{regula_2021-4}. By Fekete's lemma~\cite{fekete_1923}, this means that the limit $\lim_{n\to\infty} \frac1n \DD_{\Omega,\FF_n}(\rho^{\otimes n})$
exists.
\end{proof}

In addition to the conjectured generalisation of quantum Stein's lemma, the works of Brand\~ao and Plenio~\cite{brandao_2010-1,brandao_2010} established an interesting connection between quantum hypothesis testing and operational aspects of quantum resource transformations. It was shown, in particular, that the asymptotic error exponent in (conventional) asymmetric hypothesis testing (Eq.~\eqref{eq:gen_steins}) gives exactly the rate at which quantum entanglement can be distilled, i.e.\ converted to pure maximally entangled states, under the class of non-entangling quantum channels. This relation is mirrored exactly in postselected hypothesis testing: the quantity $\DD_{\Omega,\FF}^\infty (\rho)$ in Theorem~\ref{thm:postselected_steins} was recently shown to give the rate of distillation under probabilistic non-entangling operations~\cite{regula_2022-1}, and here we have shown that it is also precisely the asymptotic error exponent.

%%%%%%%%%%%%%%%%%%%%%%%%%%%%%%%%%%%%%%%%%%%%%%%%%%%%%%%%%%%%%%%%%%
%%%%%%%%%%%%%%%%%%%%%%%%%%%%%%%%%%%%%%%%%%%%%%%%%%%%%%%%%%%%%%%%%%
%%%%%%%%%%%%%%%%%%%%%%%%%%%%%%%%%%%%%%%%%%%%%%%%%%%%%%%%%%%%%%%%%%

\section{Symmetric case}\label{sec:symmetric}

%%%%%%%%%%%%%%%%%%%%%%%%%%%%%%%%%%%%%%%%%%%%%%%%%%%%%%%%%%%%%%%%%%

\subsection{One-shot symmetric hypothesis testing error}\label{sec:symmetric_oneshot}

Consider the discrimination of two states, $\rho$ and $\sigma$, with their respective priors being $p$ and $q$. Let us recall the definition of the postselected probability of error $\cperr$:
\begin{equation}\begin{aligned}
    \cperr(\rho,\sigma \pbar p,q) = \inf_{M_1, M_2 \geq 0,\; M_1 + M_2 \leq \id} \, \frac{p\, \Tr M_2 \rho + q\, \Tr M_1 \sigma}{ p \Tr [(M_1 + M_2)\rho] + q \Tr [(M_1 + M_2)\sigma]}.
\end{aligned}\end{equation}
We note that a closely related figure of merit appeared previously in~\cite{fiurasek_2003}, motivated by a framework generalising the setting of unambiguous state discrimination (cf.\ also~\cite{bagan_2012}). There, the authors computed the postselected probability of \emph{success}, which is just $1 - \cperr(\rho,\sigma \pbar p,q)$ in our notation. However, the asymptotic properties of this quantity were not studied, nor was its relation with the Thompson metric observed.

Our first result establishes that the one-shot postselected discrimination error is given by a simple function of the Thompson metric, or rather its non-logarithmic variant
\begin{equation}\begin{aligned}
  \Thomp(\rho \| \sigma) \coloneqq 2^{\DD_\Thomp(\rho\|\sigma)}.
\end{aligned}\end{equation}
Note that below we will apply the Thompson metric also to rescaled quantum states as $\Thomp(p \rho \| q \sigma)$. In such a case, we use the exact same definition of the max-relative entropy as for normalised states, which gives
\begin{equation}\begin{aligned}
  D_{\max} (p \rho \| q \sigma) &= \log \inf \lset \lambda \in \RR_+ \bar p \rho \leq \lambda q \sigma \rset\\
  &= \log \frac{p}{q} + D_{\max}(\rho \| \sigma).
\end{aligned}\end{equation}

\begin{boxed}{white}
\begin{theorem}\label{thm:symmetric}
For all states $\rho, \sigma$ and for all priors $p,q =  1-p \in (0,1)$, we have
\begin{equation}\begin{aligned}
    \cperr(\rho,\sigma \pbar p,q) = \big( \Thomp (p \rho \| q \sigma ) + 1 \big)^{-1}.
\end{aligned}\end{equation}
\end{theorem}
\end{boxed}
\begin{proof}
We first derive the dual formulation of the Thompson metric. Writing
\begin{equation}\begin{aligned}\label{eq:thompson_primal}
   \Thomp (p \rho \| q \sigma ) &= \inf \lsetr \max\{ \lambda_1, \lambda_2 \} \barr \frac{1}{\lambda_1} q \sigma \leq  p \rho \leq \lambda_2 \,q \sigma \rsetr \\
   &= \inf \lsetr \lambda \barr q \sigma \leq \lambda \, p \rho,\; p \rho \leq \lambda q \sigma \rsetr,
\end{aligned}\end{equation}
we obtain the dual program as
\begin{equation}\begin{aligned}\label{eq:thompson_dual}
    \Thomp (p \rho \| q \sigma ) &= \sup \lset q \Tr A\sigma + p \Tr B\rho \bar A, B \geq 0,\; q \Tr B\sigma + p \Tr A\rho = 1 \rset\\
    &= \sup_{A,B \geq 0} \,\frac{q \Tr A\sigma + p \Tr B\rho}{q \Tr B\sigma + p \Tr A\rho}.
\end{aligned}\end{equation}
The fact that we were justified in claiming that the optimal value of~\eqref{eq:thompson_primal} equals the optimal value of~\eqref{eq:thompson_dual} follows from the strong duality of the optimisation problem, easily seen by noticing that $A = B =\id$ are strongly feasible solutions and invoking Slater's theorem~\cite[Thm.~28.2]{rockafellar_1970}.

As before, we see that the constraint $M_1 + M_2 \leq \id$ in the definition of $\cperr$ is superfluous and we can ignore it without loss of generality. Notice then that
\begin{equation}\begin{aligned}
   \cperr(\rho,\sigma \pbar p,q)^{-1} &= \sup_{M_1,M_2\geq 0} \frac{q \Tr [(M_1 + M_2)\sigma] + p \Tr [(M_1 + M_2)\rho]}{q \Tr M_1 \sigma + p \Tr M_2 \rho}\\
   &= \sup_{M_1,M_2\geq 0} \left( 1 + \frac{q \Tr M_2\sigma + p \Tr M_1 \rho}{q \Tr M_1 \sigma + p \Tr M_2 \rho} \right)\\
   &= \big( 1 +\Thomp (p \rho \| q \sigma ) \big),
\end{aligned}\end{equation}
where we used the dual program for the Thompson metric~\eqref{eq:thompson_dual}. This establishes the desired equivalence between the two problems.
\end{proof}

\para{Achievability} A measurement that achieves the error $\cperr(\rho,\sigma \pbar p,q)$ can be constructed again using the fact that $D_{\max} (p \rho \| q \sigma) = \log \left( \frac{p}{q} \norm{\sigma^{-1/2} \rho \sigma^{-1/2}}{\infty} \right)$. Below, we will assume that $\DD_\Thomp (p \rho \| q \sigma) = D_{\max}( p \rho \| q \sigma)$ without loss of generality. Let $\ket{\psi_{\max}}$ be an eigenvector corresponding to the largest eigenvalue of the operator $\sigma^{-1/2} \rho \sigma^{-1/2}$. Then, define
\begin{equation}\begin{aligned}
  W = \sigma^{-1/2} \proj{\psi_{\max}} \sigma^{-1/2}.
\end{aligned}\end{equation}
This operator satisfies $\Tr W \sigma = 1$ and $\Tr W \rho = \norm{\sigma^{-1/2} \rho \sigma^{-1/2}}{\infty}$. Then, define the three-outcome POVM
\begin{equation}\begin{aligned} \label{eq:achieving_symmetric}
  \left\{ \frac{W}{\bra{\psi_{\max}} \sigma^{-1} \ket{\psi_{\max}}},\, 0,\, \id - \frac{W}{\bra{\psi_{\max}} \sigma^{-1} \ket{\psi_{\max}}} \right\},
\end{aligned}\end{equation}
where we used that the largest (and the only non-zero) eigenvalue of $W$ is $\bra{\psi_{\max}} \sigma^{-1} \ket{\psi_{\max}}$. Although the probability of conclusively discriminating $\rho$ is zero, the average error $\cperr$ indeed matches the optimal value of $(1+\frac{p}{q} \norm{\sigma^{-1/2} \rho \sigma^{-1/2}}{\infty})^{-1}$. Note that when $\DD_\Thomp (p \rho \| q \sigma) = D_{\max}( q \sigma \| p \rho)$, the role of $\rho$ and $\sigma$ is interchanged in the above discussion.

%%%%%%%%%%%%%%%%%%%%%%%%%%%%%%%%%%%%%%%%%%%%%%%%%%%%%%%%%%%%%%%%%%

\subsection{Symmetric hypothesis testing error exponent}\label{sec:symmetric_exp}

Using the additivity of $D_{\max}$ and the property that $D_{\max}(p \rho \| q \sigma) = D_{\max}(\rho \| \sigma) +  \log \frac{p}{q}$, we have that
\begin{equation}\begin{aligned}
  \DD_\Thomp(p \rho^{\otimes n} \| q \sigma^{\otimes n}) &= \max \left\{ \log \frac{p}{q} + n D_{\max}(\rho \| \sigma),\,  \log \frac{q}{p} + n D_{\max}(\sigma \| \rho) \right\}.
\end{aligned}\end{equation}
For large enough $n$, the maximum on the right-hand side will be determined only by the maximum between $D_{\max}(\rho \| \sigma)$ and $D_{\max}(\sigma \| \rho)$, which implies that
\begin{equation}\begin{aligned}
  \lim_{n\to\infty} \frac1n \DD_\Thomp(p \rho^{\otimes n} \| q \sigma^{\otimes n}) = \DD_\Thomp(\rho \| \sigma),
\end{aligned}\end{equation}
as long as the priors $p,q$ are non-zero. Theorem~\ref{thm:symmetric} then immediately gives the asymptotic regularisation of the symmetric distinguishability error.

\begin{boxed}{white}
\begin{corollary}[Asymptotic error exponent of symmetric postselected hypothesis testing]\label{cor:symmetric}
$\,$\\For all states $\rho, \sigma$ and for all priors $p,q =  1-p \in (0,1)$,
\begin{equation}\begin{aligned}
  \lim_{n\to\infty} - \frac1n \log \cperr(\rho^{\otimes n}, \sigma^{\otimes n} \pbar p,q) = \DD_{\Thomp} (\rho \| \sigma).
\end{aligned}\end{equation}
\end{corollary}
\end{boxed}

As in the asymmetric case, the expression for $- \log \cperr\left(\rho^{\otimes n},\sigma^{\otimes n} \pbar p,q\right)$ contains no $n$-dependent terms of order lower than 1. Furthermore, and again in analogy with what happens in the asymmetric case, there is a product strategy achieving the asymptotically optimal conditional error probability in symmetric discrimination. It consists in making the measurement $\{M_1,M_2,M_?\}$ detailed in~\eqref{eq:achieving_symmetric} on each copy of the state, guessing $\rho$ if the outcome `1' is obtained $n$ times, and declaring an inconclusive result in any other case.

%We
Finally, we remark that the exact expression for many-copy discrimination error simplifies slightly for the case $p=q=\frac12$, where the fact that $\DD_\Thomp(\rho^{\otimes n}\|\sigma^{\otimes n}) = n\,\DD_\Thomp(\rho\|\sigma)$ yields
\begin{equation}\begin{aligned}
    -\log \cperr\left(\rho^{\otimes n},\sigma^{\otimes n} \,\left|\, \frac12, \frac12\right.\right) = \log \big( \Thomp (\rho \| \sigma )^n + 1 \big)
\end{aligned}\end{equation}
for every $n$.

%%%%%%%%%%%%%%%%%%%%%%%%%%%%%%%%%%%%%%%%%%%%%%%%%%%%%%%%%%%%%%%%%%

\subsection{Properties of the symmetric error}\label{sec:symmetric_prop}

\begin{proposition}\label{prop:symmetric_prop}
For all states $\rho$ and $\sigma$, and all priors $p, q = 1-p \in (0,1)$:
\begin{enumerate}[(i)]
\item 
$\DD_\Thomp(p \rho \| q \sigma) \geq 0$, with equality if and only if $\rho = \sigma$ and $p = q$. Hence, $\cperr(\rho,\sigma \pbar p,q) \leq \frac12$, with equality if and only if $\rho = \sigma$ and $p = q$.

\item $\DD_\Thomp(p \rho \| q \sigma) = \infty$ if and only if $\supp \rho \neq \supp \sigma$. Hence, $\cperr(\rho,\sigma \pbar p,q) = 0$ if and only if $\rho$ and $\sigma$ have different supports.

\item For every positive linear map $\E$, the %following inequality holds 
inequality $\DD_\Thomp(p \E(\rho) \| q\E(\sigma)) \leq \DD_{\Thomp} (p \rho \| q \sigma)$ holds true.

\end{enumerate}
\end{proposition}

We note that our conclusion in point~(ii) was previously shown in~\cite{rudolph_2003}.

\begin{proof}
(i)~follows from the fact that $p \rho \leq q \sigma$ and $p \rho \geq q \sigma$ can both be true iff $p \rho = q \sigma$, which is possible if and only if $p = q$ and $\rho = \sigma$ due to $\rho$ and $\sigma$ being normalised states. (ii)~is a consequence of the fact that $D_{\max}(\rho \| \sigma) < \infty \iff \supp \rho \subseteq \supp \sigma$. (iii)~is immediate because $\E(\lambda q \sigma - p \rho) = \lambda q \E(\sigma) - p \E(\rho) \geq 0$ for every $\lambda$ such that $p \rho \leq \lambda q \sigma$ due to the positivity of $\E$. 
\end{proof}

Note that the monotonicity $\DD_{\Thomp} (p \rho \| q \sigma) \geq \DD_\Thomp(p \E(\rho) \| q \E(\sigma))$ does not mean that we can renormalise the state at the output, as we did in Proposition~\ref{prop:asymmetric_prop} for the asymmetric error. Indeed, in such a case we only get the inequality
\begin{equation}\begin{aligned}
  \Thomp\left(\left. p \frac{\E(\rho)}{\Tr \E(\rho)} \right\| q \frac{\E(\sigma)}{\Tr \E(\sigma)} \right) &\leq \max \left\{ \frac{\Tr \E(\sigma)}{\Tr \E(\rho)} 2^{D_{\max}(p \rho\|q \sigma)},\, \frac{\Tr \E( \rho)}{\Tr \E( \sigma)} 2^{D_{\max}(q \sigma\|p \rho)} \right\}\\
  &\leq \max\left\{  \frac{\Tr \E(\sigma)}{\Tr \E(\rho)},\,  \frac{\Tr \E(\rho)}{\Tr \E(\sigma)} \right\} \Thomp(p \rho \| q \sigma).
\end{aligned}\end{equation}

To see that $\Thomp\left(\left.\frac{\E(\rho)}{\Tr \E(\rho)} \right\| \frac{\E(\sigma)}{\Tr \E(\sigma)} \right) \not\leq \Thomp(\rho \| \sigma)$ in general, consider e.g.\ $\rho = \left(\begin{smallmatrix}2/3 & 0\\0 & 1/3\end{smallmatrix}\right)$ and $\sigma=\id/2$, in which case it is easy to verify that $\Thomp(\rho \|\sigma)=3/2$. Setting $\E(\cdot)\coloneqq D(\cdot)D$ with $D \coloneqq \left(\begin{smallmatrix}\sqrt{\gamma} &0\\0& 1/\sqrt{\gamma}\end{smallmatrix}\right)$ and $\gamma > 1$, we find $\Thomp\left(\left.\frac{\E(\rho)}{\Tr \E(\rho)} \right\| \frac{\E(\sigma)}{\Tr \E(\sigma)} \right) = \frac{2\gamma^2+1}{\gamma^2+1} > 3/2 = \Thomp(\rho \|\sigma)$.

%%%%%%%%%%%%%%%%%%%%%%%%%%%%%%%%%%%%%%%%%%%%%%%%%%%%%%%%%%%%%%%%%%
%%%%%%%%%%%%%%%%%%%%%%%%%%%%%%%%%%%%%%%%%%%%%%%%%%%%%%%%%%%%%%%%%%
%%%%%%%%%%%%%%%%%%%%%%%%%%%%%%%%%%%%%%%%%%%%%%%%%%%%%%%%%%%%%%%%%%

\section{Quantum channel discrimination}\label{sec:channels}

\subsection{One-shot channel hypothesis testing}

We will now consider quantum channels, that is, completely positive and trace-preserving (CPTP) linear maps acting between the spaces of operators on the Hilbert spaces of two quantum systems, $A$ and $B$. We denote such channels as $\M : A \to B$ for simplicity. Complete positivity refers here to the property that not only is the map $\mathcal{M}$ positive, i.e. it maps positive operators to positive operators, but so is $\mathrm{id} \otimes \mathcal{M}$, where the ancillary space $R$ is arbitrary. The set of quantum states of systems $A$ and $B$ will be denoted $\D_A$ and $\D_B$, respectively.

If we only have access to a single copy of two channels $\M, \N : A \to B$, there is only one thing we can do to discriminate them: feed in some state $\rho$ into the channel, and then perform a distinguishing measurement at the output. However, a crucial insight here is that using a larger, entangled state and measurement can help discriminate the channels better than simply measuring $\N(\rho)$ for some $\rho \in \D_A$~\cite{kitaev_1997}. Because of this, one performs channel hypothesis testing by measuring $\idc \otimes \M(\rho)$ for some state $\rho \in \D_{RA}$, where $\idc$ denotes the identity channel on the ancillary space $R$, with the latter space arbitrary. We therefore define the postselected asymmetric and symmetric hypothesis testing errors as
\begin{equation}\begin{aligned}
  \cbeta_\ve (\M, \N) &\coloneqq \inf_{\rho \in \D_{RA}} \cbeta_\ve ( \idc \otimes \M (\rho),\, \idc\otimes \N(\rho)),\\
  \cperr(\M, \N \pbar p, q) &\coloneqq \inf_{\rho \in \D_{RA}} \cperr ( \idc \otimes \M (\rho),\, \idc\otimes \N(\rho) \pbar p ,q ).
\end{aligned}\end{equation}

From Theorems~\ref{thm:asymmetric} and~\ref{thm:symmetric}, we immediately have that
\begin{equation}\begin{aligned}
  \cbeta_\ve (\M, \N) &= \inf_{\rho \in \D_{RA}} \left[ \frac{\ve}{1-\ve}\, \Omega( \idc \otimes \M (\rho) \| \idc\otimes \N(\rho)) + 1\right]^{-1},\\
  \cperr(\M, \N \pbar p, q) &= \inf_{\rho \in \D_{RA}} \left[ \Thomp\big( p \,\idc \otimes \M (\rho) \| q \,\idc\otimes \N(\rho)\big) + 1\right]^{-1}.
\end{aligned}\end{equation}
It may not be immediately obvious if these optimisation problems can be evaluated. However, by exploiting a property of the max-relative entropy  noticed in~\cite{wilde_2020}, we can show that it suffices to take as input the maximally entangled state $\Phi_+ = \proj{\Phi_+}$ with $\ket{\Phi_+} = \frac{1}{\sqrt{d_A}} \sum_{i=1}^{d_A} \ket{ii}$, which reduces the problem to evaluating $D_{\max}$ between two fixed states. We state this in the form of the following lemma.

\begin{boxed}{white}
\begin{lemma}\label{lem:dmax_channels}
For all channels $\M, \N : A \to B$, define
\begin{equation}\begin{aligned}
  \Omega ( \M \| \N ) &\coloneqq \sup_{\rho \in \D_{RA}} \Omega( \idc \otimes \M (\rho) \| \idc\otimes \N(\rho))\\
  \Thomp ( \M \| \N ) &\coloneqq \sup_{\rho \in \D_{RA}} \Thomp\big( \idc \otimes \M (\rho) \| \idc\otimes \N(\rho)\big).
\end{aligned}\end{equation}
Then,
\begin{equation}\begin{aligned}
  \Omega ( \M \| \N ) &= \Omega ( J_\M \| J_\N )\\
  \Thomp ( \M \| \N )  &= \Thomp (J_\M \| J_\N),
\end{aligned}\end{equation}
where $J_\M$ and $J_\N$ denote the Choi states of the channels, defined as $J_\M \coloneqq \idc \otimes \M (\Phi_+)$.
\end{lemma}
\end{boxed}
\begin{proof}
We will consider the Hilbert projective metric~$\Omega$, with the case of the Thompson metric $\Thomp$ being analogous.

First, the bound $\Omega (\M \| \N ) \geq \Omega ( J_\M \| J_\N )$ follows from the definition. For the other direction, observe that for any state $\rho$ it holds that
\begin{equation}\begin{aligned}
  D_\Omega \left(\left.  \idc \otimes \M (\rho) \right\| \idc \otimes \N (\rho) \right) &= D_{\max}\! \left(\left.  \idc \otimes \M (\rho) \right\| \idc \otimes \N (\rho) \right) \\&\quad + D_{\max} \! \left(\left.  \idc \otimes \N (\rho) \right\| \idc \otimes \M (\rho) \right)\\
  &\leq \sup_{\rho} D_{\max}\! \left(\left.  \idc \otimes \M (\rho) \right\| \idc \otimes \N (\rho) \right) \\&\quad + \sup_{\rho'} D_{\max} \! \left(\left.  \idc \otimes \N (\rho') \right\| \idc \otimes \M (\rho') \right)\\
  &= D_{\max} ( J_\M \| J_\N ) + D_{\max} ( J_\N \| J_\M )\\
  &= D_\Omega ( J_\M \| J_\N ),
\end{aligned}\end{equation}
where the second-to-last line follows because
\begin{equation}\begin{aligned}
  \sup_{\rho} D_{\max} \left(\left. \idc \otimes \M (\rho) \right\| \idc \otimes \N (\rho) \right) = D_{\max}  ( J_\M \| J_\N )
\end{aligned}\end{equation}
holds for all channels $\M, \N$~\cite[Lemma~12]{wilde_2020}. Since this upper bound applies to any $\rho$, it holds also for the supremum, thus matching the lower bound.
\end{proof}

Lemma~\ref{lem:dmax_channels} then immediately gives the following result, which establishes closed-form expressions for the asymmetric and symmetric error in postselected channel hypothesis testing.

\begin{boxed}{white}
\begin{corollary}\mbox{}\label{cor:channel_oneshot}
For all $\ve \in (0,1)$ and  $p, q = 1-p \in (0,1)$, we have that
\begin{equation}\begin{aligned}
  \cbeta_\ve (\M, \N) &= \left( \frac{\ve}{1-\ve}\, \Omega(J_\M \| J_\N) + 1 \right)^{-1} , \\
  \cperr(\M, \N \pbar p, q) &= \Big( \Thomp(p J_\M \| q J_\N) + 1\Big)^{-1}.
\end{aligned}\end{equation}
\end{corollary}
\end{boxed}

%%%%%%%%%%%%%%%%%%%%%%%%%%%%%%%%%%%%%%%%%%%%%%%

\subsection{Asymptotic channel hypothesis testing}

When given $n$ uses of a channel, the basic strategy is the so-called \emph{parallel channel discrimination}, which simply puts the channel to be discriminated in a tensor product $\M^{\otimes n}$ and optimises over all input states and measurements. Let us then define the $n$-copy parallel distinguishability errors as
\begin{equation}\begin{aligned}
  \cbetaPar (\M, \N \pbar n) &\coloneqq \cbeta_\ve (\M^{\otimes n}, \N^{\otimes n})\\
  &\phantom{:}= \inf_{\rho \in \D_{RA^{\otimes n}}} \cbeta_\ve ( \idc \otimes \M^{\otimes n} (\rho),\, \idc\otimes \N^{\otimes n}(\rho)),\\
  \cperrPar(\M, \N \pbar p, q \pbar n) &\coloneqq \cperr(\M^{\otimes n}, \N^{\otimes n} \pbar p, q)\\
 &\phantom{:}= \inf_{\rho \in \D_{RA^{\otimes n}}} \cperr ( \idc \otimes \M^{\otimes n} (\rho),\, \idc\otimes \N^{\otimes n}(\rho) \pbar p ,q ).
\end{aligned}\end{equation}
where the ancillary space $R$ is \textit{a priori} arbitrary. We stress that, although the channels are in a tensor product as $\M^{\otimes n}$ and $\N^{\otimes n}$, the input states $\rho$ may be entangled,  and the measurements at the output may be joint over all systems. Exact expressions for these quantities follow immediately from Corollary~\ref{cor:channel_oneshot}, showing also that it suffices to take the $n$-copy maximally entangled input state $\Phi_+^{\otimes n}$ in the above. In particular, no entanglement between different copies is required in the preparation of the optimal input states. Furthermore, due to the discussion in Sections~\ref{sec:asymmetric_oneshot} and~\ref{sec:symmetric_exp}, we know that also joint measurements are not necessary, and that an asymptotically optimal performance can be achieved by product strategies alone. These two facts, together, would considerably simplify any practical realisation of the above channel discrimination scheme.

To consider more general channel discrimination schemes, most works employ the formalism of \emph{adaptive} discrimination schemes~\cite{chiribella_2008-1,hayashi_2009,duan_2009,harrow_2010-1,cooney_2016,wilde_2020,pirandola_2019-2}. However, as we mentioned previously, it is known that broader types of channel manipulation protocols --- in particular, ones that use more exotic transformation structures such as superposition of causal orders --- can provide advantages over adaptive protocols~\cite{chiribella_2013,ebler_2018,quintino_2019,bavaresco_2021}. We will therefore aim to characterise the most general protocols allowed by the laws of quantum physics, without assuming \emph{anything whatsoever} about their structure.

To discuss a general $n$-channel manipulation protocol, let us consider an $n$-linear map $\Upsilon$ that takes $n$ channels as input and outputs a single quantum state. Specifically, using $(A \to B)$ to denote the space of all linear maps from $A$ to $B$, the map $\Upsilon$ can be considered as a map from the space $(A \to B)^{\times n}$ to $(\mathbb{C} \to R)$, where the latter is just the space of linear operators on some quantum system $R$ of arbitrary dimension. In order for $\Upsilon$ to be a valid channel transformation, we  assume that it satisfies
\begin{equation}\begin{aligned}\label{eq:upsilon}
  \M_1, \ldots, \M_n \in \rm{CP}(A \to B) \;\Rightarrow\; \mbox{$\Upsilon$}(\M_1, \ldots, \M_n) \geq 0,
\end{aligned}\end{equation}
with $\rm{CP}(A \to B)$ denoting the set of all completely positive maps between the two spaces. That is, the sole assumption we make is that the map $\Upsilon$ is positive: it maps any $n$-tuple of completely positive maps into a positive operator. We do not enforce any form of \emph{complete} positivity of this transformation, nor do we assume normalisation (trace preservation). By construction, any adaptive protocol is included under such a definition. We will use $\mathbb{P}_n$ to denote the set of all $n$-linear maps $\Upsilon : (A \to B)^{\times n} \to (\mathbb{C} \to R)$ satisfying~\eqref{eq:upsilon}.

 We then define
\begin{equation}\begin{aligned}
  \cbetaAny (\M, \N \pbar n) &\coloneqq \inf_{\Upsilon \in \mathbb{P}_n} \cbeta_\ve \!\left( \Upsilon(\M^{\times n}), \Upsilon(\N^{\times n}) \right)\\
  \cperrAny(\M, \N \pbar p, q \pbar n) &\coloneqq \inf_{\Upsilon \in \mathbb{P}_n} \cperr\!\left(\Upsilon(\M^{\times n}), \Upsilon(\N^{\times n}) \pbar p, q\right),
\end{aligned}\end{equation}
where $\M^{\times n}$ denotes the $n$-tuple $(\M, \ldots, \M)$. With these definitions in place, we are ready to state our main result on channel discrimination.

\begin{boxed}{white}
\begin{theorem}\label{thm:channel_asymptotic}
For all $n\in\NN$, we have that
\begin{equation}\begin{aligned}
  \cbetaAny (\M, \N \pbar n) &= \cbetaPar (\M, \N \pbar n),\\
  \cperrAny(\M, \N \pbar p, q \pbar n) &= \cperrPar(\M, \N \pbar p, q \pbar n).
\end{aligned}\end{equation}

In particular, the asymptotic exponent of asymmetric postselected channel discrimination is given by the Hilbert projective metric between the Choi operators of the channels, regardless of whether only parallel or any general discrimination schemes are employed: for all $\ve \in (0,1)$,
\begin{equation}\begin{aligned}
  \lim_{n\to\infty} - \frac1n \log \cbetaAny (\M, \N \pbar n) &= \lim_{n\to\infty} - \frac1n \log \cbetaPar (\M, \N \pbar n)\\
  &= \DD_\Omega(\M \| \N)\\
  &= \DD_\Omega(J_\M \| J_\N).
\end{aligned}\end{equation}

Similarly, for all types of discrimination protocols, the asymptotic exponent of symmetric postselected channel discrimination is given by the Thompson metric between the Choi operators of the channels. 
That is, for all priors $p,q=1-p\in(0,1)$,
\begin{equation}\begin{aligned}
   \lim_{n\to\infty} - \frac1n \log \cperrAny(\M, \N \pbar p, q \pbar n) &=  \lim_{n\to\infty} - \frac1n \log \cperrPar(\M, \N \pbar p, q \pbar n) \\
   &= \DD_\Thomp(\M \| \N)\\
   &= \DD_\Thomp(J_\M \| J_\N).
\end{aligned}\end{equation}
\end{theorem}
\end{boxed}

The result establishes a complete equivalence of parallel discrimination strategies and more general protocols, not only at the asymptotic level but also for any finite number of channel uses. The computation of the asymmetric (Stein) exponent is an analogue of a corresponding result in conventional discrimination of quantum channels~\cite{wilde_2020,wang_2019-4,fang_2020-2}, but strictly stronger: Theorem~\ref{thm:channel_asymptotic} applies even beyond adaptive discrimination strategies, and furthermore proves the strong converse property of postselected channel discrimination, which is not known to hold in the conventional setting. The evaluation of the symmetric (Chernoff) exponent has no known analogue in conventional hypothesis testing of quantum channels.

\begin{proof}
The inclusion of parallel protocols in general channel discrimination protocols implies that
\begin{equation}\begin{aligned}
\cbetaAny (\M, \N \pbar n) &\leq \cbetaPar (\M, \N \pbar n) = \left( \frac{\ve}{1-\ve}\, \Omega(J_\M^{\otimes n} \| J_\N^{\otimes n}) + 1\right)^{-1},\\
  \cperrAny(\M, \N \pbar p, q \pbar n) &\leq  \cperrPar(\M, \N \pbar p, q \pbar n) =  \left( \Thomp(p J_\M^{\otimes n} \| q J_\N^{\otimes n}) + 1\right)^{-1}.
\end{aligned}\end{equation}
The equalities in the above follow by Corollary~\ref{cor:channel_oneshot}.

For the other direction, we employ the result of~\cite[Theorem~13]{regula_2021-1}, which states that
\begin{equation}\begin{aligned}\label{eq:dmax_subadd_channels}
  D_{\max} \left(\Upsilon(\M^{\times n}) \| \Upsilon(\N^{\times n}) \right) \leq n\, D_{\max} (J_\M \| J_\N)
\end{aligned}\end{equation}
for every $\Upsilon \in \mathbb{P}_n$. 
To see that Eq.~\eqref{eq:dmax_subadd_channels} is true, and in particular that the assumptions of linearity and positivity of the channel transformations are sufficient, consider any feasible $\lambda$ such that $J_\M \leq \lambda J_\N$, i.e.\ $\lambda \N - \M$ is a completely positive map. We can then follow~\cite{regula_2021-1} by writing
\begin{align*}
  \lambda^n\, \Upsilon(\N,\ldots,\N) &= \Upsilon(\lambda \N,\ldots,\lambda \N)\\
  &= \Upsilon(\lambda \N - \M, \lambda \N, \ldots, \lambda \N) + \Upsilon(\M, \lambda \N, \ldots, \lambda \N)\\
  &= \Upsilon(\lambda \N - \M, \lambda \N, \ldots, \lambda \N) + \Upsilon(\M, \lambda \N - \M, \lambda \N, \ldots, \lambda \N) \tag{\stepcounter{equation}\theequation}\\
  &\quad +\Upsilon( \M, \M, \lambda \N, \ldots, \lambda \N)\\
  &\vdotswithin{=}\\
  &= \Upsilon(\lambda \N - \M, \lambda \N, \ldots, \lambda \N) + \ldots + \Upsilon(\M, \ldots, \M, \lambda \N - \M) + \Upsilon(\M, \ldots, \M),
\end{align*}
where we only used the $n$-linearity of $\Upsilon$. 
By the properties of $\mathbb{P}_n$ (Eq.~\eqref{eq:upsilon}), all of the terms on the right-hand side are positive operators, ensuring that $ \lambda^n\, \Upsilon(\N,\ldots,\N) \geq \Upsilon(\M, \ldots, \M)$ and thus $D_{\max} \left(\Upsilon(\M^{\times n}) \| \Upsilon(\N^{\times n}) \right) \leq n \log \lambda$. Optimising over all feasible $\lambda$ gives Eq.~\eqref{eq:dmax_subadd_channels}.

We thus get
\begin{equation}\begin{aligned}
  \cbetaAny (\M, \N \pbar n) &= \inf_{\Upsilon \in \mathbb{P}_n} \left( \frac{\ve}{1-\ve}\, \Omega\left( \Upsilon(\M^{\times n}) \| \Upsilon(\N^{\times n}) \right) + 1\right)^{-1}
  \\&\geq \left( \frac{\ve}{1-\ve}\, \Omega(J_\M \| J_\N)^n + 1\right)^{-1}\\
    &= \left( \frac{\ve}{1-\ve}\, \Omega(J_\M^{\otimes n} \| J_\N^{\otimes n}) + 1\right)^{-1}\\
  &= \cbetaPar (\M, \N \pbar n) 
\end{aligned}\end{equation}
where we used the additivity property
\begin{equation}\begin{aligned}
  D_{\max} \left( J_{\M^{\otimes n}} \| J_{\N^{\otimes n}} \right) = D_{\max} (J_\M^{\otimes n} \| J_{\N}^{\otimes n}) = n\, D_{\max} (J_\M \| J_{\N}).
\end{aligned}\end{equation}
Analogously,
\begin{equation}\begin{aligned}
  \cperrAny(\M, \N \pbar p, q \pbar n) &= \inf_{\Upsilon \in \mathbb{P}_n} \Big( \Thomp(p \Upsilon(\M^{\times n}) \| q \Upsilon(\N^{\times n})) + 1\Big)^{-1}
  \\&\geq \left( \max\left\{ \frac{p}{q} 2^{n D_{\max}(J_\M \| J_\N)},\,  \frac{q}{p} 2^{n D_{\max}(J_\M \| J_\N)} \right\} + 1\right)^{-1}\\
  &=  \left( \Thomp(p J_\M^{\otimes n} \| q J_\N^{\otimes n}) + 1\right)^{-1}\\
  &= \cperrPar(\M, \N \pbar p, q \pbar n).
\end{aligned}\end{equation}

The proof is concluded by taking the limits as in the state cases (Sections~\ref{sec:asymmetric} and~\ref{sec:symmetric}).
\end{proof}

%%%%%%%%%%%%%%%%%%%%%%%%%%%%%%%%%%%%%%%%%%%%%%%%%%%%%%%%%%%%%%%%%%

\subsection{Composite channel hypothesis testing}

The case of composite channel discrimination can be considered very similarly. Given some subset of channels $\FF \subseteq \CPTP(A\to B)$, we define
\begin{equation}\begin{aligned}
  \cbeta_{\ve,\FF}(\M) \coloneqq& \inf_{\rho \in \DD_{RA}} \inf_{\substack{M_1, M_2 \geq 0\\M_1 + M_2 \leq \id}} \lset \sup_{\N \in \FF} \frac{\Tr M_1 (\idc \otimes \N(\rho))}{\Tr \left[(M_1+M_2) (\idc \otimes \N(\rho))\right]} \bar \frac{\Tr M_2 (\idc \otimes \M(\rho))}{\Tr \left[(M_1+M_2) (\idc\otimes \M(\rho))\right]} \leq \ve \rset.
\end{aligned}\end{equation}
Using the result for the state case (Theorem~\ref{thm:postselected_steins}), we immediately get that
\begin{equation}\begin{aligned}
  \cbeta_{\ve,\FF}(\M) = \left( \frac{\ve}{1-\ve}\, \Omega_\FF(\M)  + 1\right)^{-1},
\end{aligned}\end{equation}
where
\begin{equation}\begin{aligned}
  \Omega_\FF(\M) \coloneqq \sup_{\rho \in \DD_{RA}} \inf_{\N \in \FF} \,\Omega( \idc \otimes \M(\rho) \| \idc \otimes \N(\rho)).
\end{aligned}\end{equation}
The following Lemma can then be established in full analogy with Lemma~\ref{lem:dmax_channels}.
\begin{boxed}{white}
\begin{lemma}\label{lem:dmax_channels_f}
Let $\M: A \to B$ be a channel and $\FF$ a convex and closed set of quantum channels. Then
\begin{equation}\begin{aligned}
  %\Omega_\FF ( \M \| \N )
  \Omega_\FF(\M) &= \min_{\N \in \FF}\, \Omega ( J_\M \| J_\N ).
\end{aligned}\end{equation}
\end{lemma}
\end{boxed}
As a consequence, we obtain the following generalisation of quantum Stein's lemma for parallel quantum channel discrimination.

\begin{boxed}{white}
\begin{corollary}\label{cor:channel_asymptotic_f}
For all $\ve \in (0,1)$, every channel $\mathcal{M} : A \to B$, and every family $(\FF_n)_n$ of convex and closed sets $\FF_n \subseteq \CPTP(A^{\otimes n}\to B^{\otimes n})$, we have
\begin{equation}\begin{aligned}
  \liminf_{n\to\infty} - \frac1n \log \cbeta_{\ve,\FF} (\M^{\otimes n}) = \DD_{\Omega,\FF}^\infty (\M)%(\M \| \N)
 % = \DD_{\Omega,\FF}^\infty(J_\M),
\end{aligned}\end{equation}
where
\begin{equation}\begin{aligned}
  \DD_{\Omega,\FF}^\infty (\M) %(J_\M \| J_\N)
  = \liminf_{n\to\infty} \frac1n \min_{\N \in \FF} \DD_{\Omega}\left(\left.J_\M^{\otimes n} \right\| J_\N^{\otimes n}\right).
\end{aligned}\end{equation}
If the sets $\FF_n$ are closed under tensor product, i.e.\ $\N_n \in \FF_n, \N_m \in \FF_m \Rightarrow \N_n \otimes \N_m \in \FF_{n+m}$, then $\liminf$ can be replaced with $\lim$ in the above.
\end{corollary}
\end{boxed}

%%%%%%%%%%%%%%%%%%%%%%%%%%%%%%%%%%%%%%%%%%%%%%%%%%%%%%%%%%%%%%%%%%
%%%%%%%%%%%%%%%%%%%%%%%%%%%%%%%%%%%%%%%%%%%%%%%%%%%%%%%%%%%%%%%%%%
%%%%%%%%%%%%%%%%%%%%%%%%%%%%%%%%%%%%%%%%%%%%%%%%%%%%%%%%%%%%%%%%%%

\section{Resource theory of asymmetric distinguishability}\label{sec:resource_theory}

The resource theory of asymmetric distinguishability was introduced in~\cite{matsumoto_2010,wang_2019} as a way to treat the distinguishability of quantum states operationally. Just as general resource theories~\cite{chitambar_2019} are concerned with transforming quantum states under some set of restricted, `free' operations, the theory of distinguishability studies the transformations of pairs of quantum states (`boxes') under channels that act as $(\rho, \sigma) \to (\M(\rho), \M(\sigma))$, and therefore do not increase the distinguishability between the two states. The standard resourceful box, representing a unit of `perfect distinguishability', is represented by the pair $(\proj{0}, \frac12 \id_2)$. In this context, the task of distinguishability distillation was defined as transforming a given box $(\rho,\sigma)$ into tensor powers of the standard box, $(\proj{0}^{\otimes m}, \frac{1}{2^m} \id_2^{\otimes m})$. An equivalent variant of this task can be defined by considering the transformation $(\rho, \sigma) \to (\proj{0}, \pi_{2^m})$ instead, where $\pi_k \coloneqq \frac{1}{k} \proj{0} + \frac{k-1}{k} \proj{1}$, with the advantage that the latter approach allows the parameter $m$ to vary continuously.

A result of~\cite{wang_2019} was to identify the one-shot distillable distinguishability with the hypothesis testing relative entropy $D^\ve_H$~\cite{wang_2012,buscemi_2010}, which is equal to the optimised type II error in hypothesis testing: $D^\ve_H(\rho \| \sigma) \coloneqq - \log \beta_\ve(\rho\|\sigma)$. 
A natural extension of this result to the asymptotic case where $(\rho^{\otimes n}, \sigma^{\otimes n}) \to (\proj{0}^{\otimes rn}, \frac{1}{2^{rn}} \id_2^{\otimes rn})$, coupled with the quantum Stein's lemma, shows that the value of asymptotically distillable distinguishability can be identified with the relative entropy $D(\rho\|\sigma)$, recovering and extending a finding of~\cite{matsumoto_2010} (cf.\ also~\cite{buscemi_2019}).

We can obtain an analogous result in our setting, giving an operational interpretation in the resource theory of asymmetric distinguishability to the \emph{postselected hypothesis testing relative entropy}
\begin{equation}\begin{aligned}
  \cDH^\ve (\rho \| \sigma) \coloneqq - \log \cbeta_\ve (\rho \| \sigma).
\end{aligned}\end{equation}
To do so, the transformations under quantum channels need to be replaced with transformations under probabilistic maps.

\begin{corollary}
Define the probabilistically distillable distinguishability as
\begin{equation}\begin{aligned}
    D^\ve_{d,\rm prob} (\rho , \sigma) \coloneqq \log \sup_{\E \in \rm CPTNI} \lsetr k \barr F\!\left( \frac{\E(\rho)}{\Tr \E(\rho)}, \proj{0} \right) \geq 1-\ve,\; \frac{\E(\sigma)}{\Tr \E(\sigma)} = \pi_k \rsetr,
\end{aligned}\end{equation}
where $F(\omega,\tau) = \norm{\sqrt{\omega}\sqrt{\tau}}{1}^2$ is the fidelity and $\rm CPTNI$ denotes all completely positive and trace--non-increasing maps. Then
\begin{equation}\begin{aligned}
  D^\ve_{d,\rm prob} (\rho , \sigma) = \cDH^\ve  (\rho \| \sigma).
\end{aligned}\end{equation}
\end{corollary}
We remark that the fidelity can be replaced with the trace distance $\frac12 \norm{\frac{\E(\rho)}{\Tr \E(\rho)} - \proj{0}}{1}$ in the above without affecting the result.
\begin{proof}
Applying~\cite[Corollary~14]{regula_2021-4} gives that
\begin{equation}\begin{aligned}
    D^\ve_{d,\rm prob}(\rho, \sigma) = \log \left( \frac{\ve}{1-\ve} \Omega(\rho \| \sigma) + 1 \right).
\end{aligned}\end{equation}
But in Theorem~\ref{thm:asymmetric} we have identified the right-hand side with $\cDH^\ve (\rho \| \sigma)$, and so the stated result follows.
\end{proof}

An asymptotic extension of this result can be used to show that the rate of distinguishability distillation using probabilistic (CPTNI) operations is given by $\DD_\Omega(\rho \| \sigma)$. This was also recently shown in~\cite{regula_2022-1} through a more general approach applicable to broader resource theories.

The work~\cite{wang_2019} studied the reverse task of distinguishability dilution, that is, the transformation from the maximal box to a given pair of quantum states. In this context, a connection with the max-relative entropy $D_{\max}$ was obtained. In an analogous way, in the conceptually related resource theory of \emph{symmetric} distinguishability~\cite{salzmann_2021}, the corresponding task of dilution was connected with the Thompson metric $\DD_\Thomp$. In contrast, in our probabilistic setting, the task of dilution trivialises --- any state can be obtained from a pure state probabilistically~\cite{regula_2021-4}.

%%%%%%%%%%%%%%%%%%%%%%%%%%%%%%%%%%%%%%%%%%%%%%%%%%%%%%%%%%%%%%%%%%
%%%%%%%%%%%%%%%%%%%%%%%%%%%%%%%%%%%%%%%%%%%%%%%%%%%%%%%%%%%%%%%%%%
%%%%%%%%%%%%%%%%%%%%%%%%%%%%%%%%%%%%%%%%%%%%%%%%%%%%%%%%%%%%%%%%%%

\section{General probabilistic theories}\label{sec:gpts}

The framework of general probabilistic theories (GPTs)~\cite{ludwig_1985,hartkamper_1974,davies_1970} encompasses both quantum and classical probability theory, providing a natural way to study the properties of very general physical theories, and in particular to formalise the similarities and differences between them. Although concepts such as states and measurements can be defined in any such setting, studying the optimal errors in the asymptotic regime did not appear possible --- the optimal error exponents encountered in classical and quantum theory involve quantities such as the quantum relative entropy $D(\cdot\|\cdot)$ and the Chernoff divergence $\xi(\cdot\|\cdot)$, both of which require the properties of classical or quantum theory even to be defined. Because of this, no results whatsoever have been known about asymptotic hypothesis testing in broader GPTs. Our framework, however, can be straightforwardly adapted also to such general theories, and we will see that our findings on state discrimination in Sections~\ref{sec:asymmetric} and~\ref{sec:symmetric} hold in arbitrary GPTs verbatim.

\subsection{Definitions}\label{sec:gpt_def}

To define a GPT, one first identifies the %closed
compact and convex set $\S$ of \emph{states} within a real vector space $\V$, here assumed to be finite dimensional. In analogy with the geometry of quantum theory, the affine hull of the set $\S$ is required to have co-dimension $1$ within $\V$, and to not contain the zero vector.\footnote{This seemingly ad hoc set of axioms is actually well justified from a foundational perspective. We refer the interested reader to~\cite[Chapters~1--2]{lami_2017-1} for more details on these points as well as complete derivations.} The cone generated by $\S$, namely $\C \coloneqq \lset \lambda \rho \bar \lambda \in \RR_+,\, \rho \in \S \rset$, is used to induce a partial order on the space $\V$ as $x \leq_\C y \iff y-x \in \C$. The compactness of $\S$, which implies that $\C$ is pointed, or $\C \cap (-\C) = \{0\}$, indeed means that this is a partial order. Furthermore, the fact that the affine hull of $\S$ is full dimensional entails that $\C$ is generating, i.e.\ $\operatorname{span}(\C) = \V$.

The dual space $\V\*$, i.e.\ the space of continuous linear functionals $M \colon \V \to \RR$, then has a partial order $X \leq_{\C\*} Y \iff Y -X \in {\C\*}$ induced by the dual cone $\C\*\coloneqq \lset M \in \V\* \bar \< M, \rho \> \geq 0 \; \forall \rho \in \S \rset$, where we write $\< M, x \> = M(x)$ for every $x \in \V$, $M \in \V\*$. We use the notation $X <_{\C\*} Y$ to denote $Y - X \in \operatorname{int}(\C\*)$.

In order to define measurements in this theory, we also need to identify the set of \emph{effects}, that is, elements of the dual space which yield physical measurement outcomes. First, a fixed \emph{unit effect} $U >_{\C\*} 0$ is defined as the unique element of $\V\*$ such that $\< U, \rho \> = 1 \; \forall \rho \in \S$. The so-called no restriction hypothesis~\cite{ludwig_1985,barrett_2007,chiribella_2010} then allows us to assume that every element $M \in \V\*$ such that $\< M, \rho \> \in [0,1] \; \forall \rho \in \S$ is a valid effect, and thus a valid measurement operator; this can be equivalently understood as $0 \leq_{\C\*} M  \leq_{\C\*} U$. A measurement is then a collection of effects $\{M_i\}_i$ such that $\sum_{i=1}^n M_i = U$ --- in other words, a collection such that the measurement outcome probabilities $\< M_i, \rho\>$ sum to unity for every state $\rho$.

We remark that both quantum and classical probability theory are special cases of the above.
Although the Hilbert space underlying the setting of quantum mechanics is complex, the space where density operators live --- corresponding to $\V$ here --- is nevertheless a real vector space, which is why the general assumptions of GPTs encompass quantum mechanics. Specifically, in the quantum case, $\V$ is the space of self-adjoint operators acting on some complex Hilbert space, $\C$ and $\C\*$ are both cones of postive semidefinite operators, and the unit effect $U$ is the identity.

We will use the tuple $(\V, \C, U)$ to denote a GPT as above, as the theory is uniquely determined by the three choices. Given two GPTs, $(\V_A, \C_A, U_A)$ and $(\V_B, \C_B, U_B)$, we then wish to define a bipartite structure that combines the two. Under natural assumptions including the so-called local tomography principle, which states that bipartite states should be fully determined by the statistics of local measurements, it holds that $\V_{AB}$ is isomorphic to $\V_A \otimes \V_B$~\cite{klay_1987,wilce_1992}, which means that we can identify $\V_{AB} = \V_A \otimes \V_B$ and $U_{AB} = U_A \otimes U_B$ without loss of generality. Interestingly, this does \emph{not} uniquely determine a bipartite cone $\C_{AB}$; the only thing that can be concluded is that~\cite{namioka_1969}
\begin{equation}\begin{aligned}\label{eq:minmax_cones}
  \C_A \omin \C_B \subseteq \C_{AB} \subseteq \C_A \omax \C_B
\end{aligned}\end{equation}
where the \emph{minimal tensor product} and the \emph{maximal tensor product} are defined, respectively, as
\begin{equation}\begin{aligned}
  \C_A \omin \C_B &\coloneqq \operatorname{conv} (\C_A \otimes \C_B),\\
  \C_A \omax \C_B &\coloneqq \Big( \C^\aster_A \omin \C^\aster_B\Big)\*.
\end{aligned}\end{equation}
For every pair of non-classical GPTs, it holds that $\C_A \omin \C_B \neq  \C_A \omax \C_B$~\cite{aubrun_2021}, and the difference between these two cones is paramount to understanding properties of multipartite GPTs, in particular the existence of phenomena such as entanglement~\cite{namioka_1969,aubrun_2022}.

Despite the fact that many aspects of GPTs crucially depend on the specific choice of a multipartite cone $\C_{AB}$, our results will apply to \emph{any} choice of $\C_{AB}$ as in~\eqref{eq:minmax_cones}, making them universally applicable. We remark in particular that $\C_A \otimes \C_B \subseteq \C_{AB}$ and $\C^\aster_A \otimes \C^\aster_B \subseteq \C^\aster_{AB}$ for all valid choices of $\C_{AB}$.

The extension of the above to the tensor product of more than two GPTs is immediate. 

\subsection{Extension of results}

The task of state discrimination, in both the conventional and postselected settings, can be formulated in exactly the same way as we have done it in Section~\ref{sec:intro}; specifically,
\begin{equation}\begin{aligned}
  \cbeta_\ve (\rho, \sigma) &= \inf_{M_1,M_2 \geq_{\C\*} 0} \lset \frac{\< M_1, \sigma \>}{\< M_1+M_2, \sigma \>} \bar M_1 + M_2 \leq_{\C\*} U,\; \frac{\< M_2, \rho \>}{ \< M_1+M_2,  \rho \>} \leq \ve \rset\\
  \cperr(\rho,\sigma \pbar p,q) &= \inf_{\substack{M_1, M_2 \geq_{\C\*} 0,\\ M_1 + M_2 \leq_{\C\*} U}} \, \frac{p \< M_2, \rho \> + q \< M_1, \sigma \>}{ p  \< M_1 + M_2, \rho \> + q \< M_1 + M_2, \sigma \>},
\end{aligned}\end{equation}
where $\rho, \sigma \in \S$ and $p, q = 1-p \in (0,1)$ as before. The equivalent of the max-relative entropy can also be defined in the same way,
\begin{equation}\begin{aligned}
  D_{\max} (\rho \| \sigma) \coloneqq \log \inf \lset \lambda \in \RR_+ \bar \rho \leq_{\C} \lambda \sigma \rset,
\end{aligned}\end{equation}
with the definitions of $\DD_\Omega$ and $\DD_\Thomp$ naturally extending to the GPT formalism using the above quantity.

The results of our Theorem~\ref{thm:asymmetric} (asymmetric error) and Theorem~\ref{thm:symmetric} (symmetric error) then apply to any GPT \emph{exactly as stated} under the replacement $\Tr (M \cdot) \mapsto \<M, \cdot \>$ --- there is no need to adjust any of the proofs, as we have not made use of any properties of quantum mechanics, with our proofs resting on the linearity of the trace and convex (Lagrange) duality, both of which hold for arbitrary GPTs defined as in Section~\ref{sec:gpt_def}.

This correspondence is not too surprising, given that many results on one-shot state discrimination can be mapped between quantum mechanics and GPTs~\cite{kimura_2010,bae_2016-1}. What is more remarkable about the setting of postselected hypothesis testing is that the asymptotic results --- namely, Corollary~\ref{cor:asymmetric_exp} and Corollary~\ref{cor:symmetric} --- are also valid in arbitrary GPTs, with no other assumptions needed. This is because of the easily verified additivity of $D_{\max}$, which we formalise below.

\begin{lemma}
Given a local GPT $(\V, \C, U)$, consider %the
any $n$-copy GPT %defined by
$(\V_n, \C_n, U_n)$, where $\V_n = \bigotimes_{i=1}^n \V$, $U_n = \bigotimes_{i=1}^n U$, and $\underset{\smash{\min\phantom{_{i=1}^n}}}{\bigotimes_{i=1}^n} \C \subseteq \C_n \subseteq \underset{\smash{\max\phantom{_{i=1}^n}}}{\bigotimes_{i=1}^n} \C$. Then, for all $\rho,\sigma \in \S$,
\begin{equation}\begin{aligned}
  D_{\max}(\rho^{\otimes n} \| \sigma^{\otimes n}) = n D_{\max} (\rho \| \sigma).
\end{aligned}\end{equation}
\end{lemma}
\begin{proof}
Consider any feasible $\lambda$ such that $\rho \leq_\C \lambda \sigma$. Then, $\rho^{\otimes n} \leq_{\C_n} (\lambda\sigma)^{\otimes n} = \lambda^n \sigma^{\otimes n}$; this can be seen by using the fact that $\C \omin \C \subseteq \C_2$ to show that
\begin{equation}\begin{aligned}
  \lambda^2 (\sigma \otimes \sigma) - \rho \otimes \rho =  (\lambda \sigma - \rho) \otimes \lambda \sigma + \rho \otimes (\lambda \sigma - \rho) \geq_{\C_2} 0
\end{aligned}\end{equation}
and extending to arbitrary $n$ by induction. 
This implies that $D_{\max}(\rho^{\otimes n}) \leq n \log\lambda$, and hence $D_{\max}(\rho^{\otimes n} \| \sigma^{\otimes n}) \leq n D_{\max} (\rho \| \sigma)$.

For the other direction, we write $D_{\max}$ in its dual form as
\begin{equation}\begin{aligned}
  D_{\max}(\rho\|\sigma) = \log \sup_{W \geq_{\C\*} 0} \frac{\< W, \rho \>}{\< W, \sigma \>}.
\end{aligned}\end{equation}
The fact that $\bigotimes_{i=1}^n \C\* \subseteq \C^\aster_n$ immediately implies that any feasible solution $W$ to $D_{\max}(\rho\|\sigma)$ gives a feasible solution to $D_{\max}(\rho^{\otimes n}\|\sigma^{\otimes n})$ as $W^{\otimes n}$, concluding the proof.
\end{proof}

The above lemma tells us that $\DD_\Omega(\rho^{\otimes n}\|\sigma^{\otimes n}) = n \DD_\Omega(\rho\|\sigma)$ and $\DD_\Thomp(\rho^{\otimes n} \| \sigma^{\otimes n} ) = n \DD_\Thomp(\rho \| \sigma)$ for all normalised $\rho, \sigma \in \S$. The asymptotic results of Corollaries~\ref{cor:asymmetric_exp} and~\ref{cor:symmetric} thus follow directly; namely, we have that
\begin{equation}\begin{aligned}
  \lim_{n\to\infty} - \frac1n \log \cbeta_{\ve}(\rho^{\otimes n},\sigma^{\otimes n}) = \DD_{\Omega} (\rho \| \sigma)
\end{aligned}\end{equation}
and
\begin{equation}\begin{aligned}
 \lim_{n\to\infty} - \frac1n \log \cperr(\rho^{\otimes n}, \sigma^{\otimes n} \pbar p,q) = \DD_{\Thomp} (\rho \| \sigma)
\end{aligned}\end{equation}
for all $\rho, \sigma \in \S$, all $\ve \in (0,1)$, and all $p, q=1-p \in (0,1)$.

Theorem~\ref{thm:postselected_steins}, which establishes a postselected generalisation of the quantum Stein's lemma in the composite setting, also immediately holds in every GPT. The proof is, once again, noticed not to use any assumptions save for linearity and convex duality.

We thus see that virtually all of our results on state discrimination immediately extend to GPTs, without requiring any restrictive assumptions on the latter. On the other hand, although we expect at least some of our channel discrimination results to also be generalisable to GPTs under some reasonable assumptions, we should note that this is far from immediate --- even defining  a `channel' is no trivial task in a framework so general, and tools such as the Choi--Jamiołkowski isomorphism might no longer be available to us. We thus leave the question of investigating channel discrimination in this setting as an open problem.

%%%%%%%%%%%%%%%%%%%%%%%%%%%%%%%%%%%%%%%%%%%%%%%%%%%%%%%%%%%%%%%%%%
%%%%%%%%%%%%%%%%%%%%%%%%%%%%%%%%%%%%%%%%%%%%%%%%%%%%%%%%%%%%%%%%%%
%%%%%%%%%%%%%%%%%%%%%%%%%%%%%%%%%%%%%%%%%%%%%%%%%%%%%%%%%%%%%%%%%%

\section{Discussion}

The approach of our work is inherently different from conventional quantum hypothesis testing, yet we hope that it can ultimately prove helpful in shedding light also on the latter. The distinguishing feature of our results, and perhaps an unexpected finding, is how significantly the whole framework simplifies in the postselected setting: the optimal probabilities can be given straightforward closed-form expressions; the case of composite hypothesis testing reduces to a simple extension of the i.i.d.\ case; all terms but those linear in $n$ disappear from the asymptotic expansions; and, importantly, the task of channel discrimination not only admits efficiently computable error probabilities and exponents, but in fact reduces all discrimination strategies to the much simpler case of parallel protocols. This allowed us to essentially characterise completely the task of postselected quantum hypothesis testing with just a few concise proofs --- something that, in the conventional setting, has been a multi-decade endeavour requiring significant effort and the development of highly specialised methods.

An underlying property of our framework is the need to condition on a conclusive outcome of the measurement. This is an easily realisable operation: simply repeat the measurement protocol, disregarding measurement data corresponding to inconclusive outcomes. It is, however, conceivable that the conclusive outcome may have a low probability of occurrence, meaning that many repetitions of the measurement protocol would be needed to successfully discriminate the states or channels. Such a difference indeed makes direct comparisons between $n$-copy distinguishability protocols in the postselected and conventional settings difficult to perform. However, as we argued in~\cite{regula_2022-1}, allowing postselection is much more easily justified in the limit of infinitely many i.i.d.\ copies, where the definition of a rate may be adjusted depending on what `cost function' we wish to employ:  if we are interested in error rates per number of copies of states \emph{generated}, then the conventional setting appears more appropriate; if, however, we instead focus on how many copies of states need to be \emph{manipulated} at once, then any protocol may be repeated without incurring an additional cost, and a postselected setting is thus perfectly natural. The latter assumption is not necessarily an extravagant one --- generating many copies of states is far easier than coherently manipulating them with current technologies.

The idea of postselection has previously appeared in a range of contexts in quantum information, including --- besides the aforementioned state discrimination~\cite{ivanovic_1987,dieks_1988,peres_1988} --- quantum computing~\cite{aaronson_2005}, metrology~\cite{fiurasek_2006,gendra_2012,combes_2014,combes_2015,arvidsson-shukur_2020}, connections with spacetime geometry~\cite{lloyd_2011,lloyd_2011b,brun_2012}, and manipulation of quantum entanglement~\cite{gisin_1996,kent_1998,horodecki_1999-1} as well as more general quantum resources~\cite{reeb_2011,regula_2021-4,regula_2022-1}.
In such settings, our results can be viewed as ultimate limits that cannot be exceeded even if postselection is permitted: the optimal discrimination error bounds shown in our work cannot be beat by any measurement strategy, postselected or not, even in permissive frameworks such as ones allowing the increased distinguishing power of postselected closed timelike curves~\cite{lloyd_2011,brun_2012}.

Ultimately, although it might be contentious to claim that it has directly led to practical advancements, postselection has unquestionably contributed to the understanding of the foundations of quantum mechanics~\cite{aharonov_1988,pusey_2014,oreshkov_2015}, and even formed the conceptual basis of the first proposals for quantum supremacy experiments~\cite{harrow_2017}.
In a similar manner, we hope that our results can find use both in the formalisation of the foundations of quantum information, as well as in the study of the limits of practical state and channel discrimination protocols.

%%%%%%%%%%%%%%%%%%%%%%%%%%%%%%%%%%%%%%%%%%%%%%%%%%%%%%%%%%%%%%%%%%
%%%%%%%%%%%%%%%%%%%%%%%%%%%%%%%%%%%%%%%%%%%%%%%%%%%%%%%%%%%%%%%%%%

\subsection*{Acknowledgments}

We thank Yonglong Li and Marco Tomamichel for discussions on related settings studied in classical hypothesis testing~\cite{gutman_1989}. 
B.R.\ was partially supported by the Japan Society for the Promotion of Science (JSPS) KAKENHI Grant No.\ 21F21015 and the JSPS Postdoctoral Fellowship for Research in Japan. L.L.\ was supported by the Alexander von Humboldt Foundation. M.M.W.\ acknowledges support from the NSF under grant no.~1907615.

%%%%%%%%%%%%%%%%%%%%%%%%%%%%%%%%%%%%%%%%%%%%%%%%%%%%%%%%%%%%%%%%%%
%%%%%%%%%%%%%%%%%%%%%%%%%%%%%%%%%%%%%%%%%%%%%%%%%%%%%%%%%%%%%%%%%%
%\clearpage
 \bibliographystyle{apsc}
 \bibliography{main}

%merlin.mbs apsrev4-1.bst 2010-07-25 4.21a (PWD, AO, DPC) hacked
%Control: key (0)
%Control: author (72) initials jnrlst
%Control: editor formatted (1) identically to author
%Control: production of article title (1) required
%Control: page (0) single
%Control: production of eprint (0) enabled
%Control: year (1) truncated
\begin{thebibliography}{103}%
\makeatletter
\providecommand \@ifxundefined [1]{%
 \@ifx{#1\undefined}
}%
\providecommand \@ifnum [1]{%
 \ifnum #1\expandafter \@firstoftwo
 \else \expandafter \@secondoftwo
 \fi
}%
\providecommand \@ifx [1]{%
 \ifx #1\expandafter \@firstoftwo
 \else \expandafter \@secondoftwo
 \fi
}%
\providecommand \natexlab [1]{#1}%
\providecommand \emph  [1]{``#1''}%
\providecommand \bibnamefont  [1]{#1}%
\providecommand \bibfnamefont [1]{#1}%
\providecommand \citenamefont [1]{#1}%
\providecommand \href@noop [0]{\@secondoftwo}%
\providecommand \href [0]{\begingroup \@sanitize@url \@href}%
\providecommand \@href[1]{\@@startlink{#1}\@@href}%
\providecommand \@@href[1]{\endgroup#1\@@endlink}%
\providecommand \@sanitize@url [0]{\catcode `\\12\catcode `\$12\catcode
  `\&12\catcode `\#12\catcode `\^12\catcode `\_12\catcode `\%12\relax}%
\providecommand \@@startlink[1]{}%
\providecommand \@@endlink[0]{}%
\providecommand \url  [0]{\begingroup\@sanitize@url \@url }%
\providecommand \@url [1]{\endgroup\@href {#1}{\urlprefix }}%
\providecommand \urlprefix  [0]{URL }%
\providecommand \Eprint [0]{\href }%
\providecommand \doibase [0]{http://dx.doi.org/}%
\providecommand \selectlanguage [0]{\@gobble}%
\providecommand \bibinfo  [0]{\@secondoftwo}%
\providecommand \bibfield  [0]{\@secondoftwo}%
\providecommand \translation [1]{[#1]}%
\providecommand \BibitemOpen [0]{}%
\providecommand \bibitemStop [0]{}%
\providecommand \bibitemNoStop [0]{.\EOS\space}%
\providecommand \EOS [0]{\spacefactor3000\relax}%
\providecommand \BibitemShut  [1]{\csname bibitem#1\endcsname}%
\let\auto@bib@innerbib\@empty
%</preamble>
\bibitem [{\citenamefont {Hayashi}(2017)}]{hayashi_2016}%
  \BibitemOpen
  \bibfield  {author} {\bibinfo {author} {\bibfnamefont {M.}~\bibnamefont
  {Hayashi}},\ }\href@noop {} {\emph {\bibinfo {title} {Quantum {{Information
  Theory}}: {{Mathematical Foundation}}}}}\ (\bibinfo  {publisher}
  {{Springer}},\ \bibinfo {year} {2017})\BibitemShut {NoStop}%
\bibitem [{\citenamefont {Watrous}(2018)}]{watrous_2018}%
  \BibitemOpen
  \bibfield  {author} {\bibinfo {author} {\bibfnamefont {J.}~\bibnamefont
  {Watrous}},\ }\href@noop {} {\emph {\bibinfo {title} {The {{Theory}} of
  {{Quantum Information}}}}}\ (\bibinfo  {publisher} {{Cambridge University
  Press}},\ \bibinfo {address} {{Cambridge}},\ \bibinfo {year}
  {2018})\BibitemShut {NoStop}%
\bibitem [{\citenamefont {Wilde}(2017)}]{wilde_2017}%
  \BibitemOpen
  \bibfield  {author} {\bibinfo {author} {\bibfnamefont {M.~M.}\ \bibnamefont
  {Wilde}},\ }\href@noop {} {\emph {\bibinfo {title} {Quantum {{Information
  Theory}}}}},\ \bibinfo {edition} {2nd}\ ed.\ (\bibinfo  {publisher}
  {{Cambridge University Press}},\ \bibinfo {address} {{Cambridge}},\ \bibinfo
  {year} {2017})\BibitemShut {NoStop}%
\bibitem [{\citenamefont {Helstrom}(1969)}]{helstrom_1969}%
  \BibitemOpen
  \bibfield  {author} {\bibinfo {author} {\bibfnamefont {C.~W.}\ \bibnamefont
  {Helstrom}},\ }\bibfield  {title} {\emph {\bibinfo {title} {Quantum detection
  and estimation theory},}\ }\href {http://dx.doi.org/10.1007/BF01007479}
  {\bibfield  {journal} {\bibinfo  {journal} {J. Stat. Phys.}\ }\textbf
  {\bibinfo {volume} {1}},\ \bibinfo {pages} {231} (\bibinfo {year}
  {1969})}\BibitemShut {NoStop}%
\bibitem [{\citenamefont {Holevo}(1972)}]{holevo_1972}%
  \BibitemOpen
  \bibfield  {author} {\bibinfo {author} {\bibfnamefont {A.~S.}\ \bibnamefont
  {Holevo}},\ }\bibfield  {title} {\emph {\bibinfo {title} {{An analogue of
  statistical decision theory and noncommutative probability theory}},}\
  }\href@noop {} {\bibfield  {journal} {\bibinfo  {journal} {Trudy Moskov. Mat.
  Obsc.}\ }\textbf {\bibinfo {volume} {26}},\ \bibinfo {pages} {133} (\bibinfo
  {year} {1972})}\BibitemShut {NoStop}%
\bibitem [{\citenamefont {Umegaki}(1962)}]{umegaki_1962}%
  \BibitemOpen
  \bibfield  {author} {\bibinfo {author} {\bibfnamefont {H.}~\bibnamefont
  {Umegaki}},\ }\bibfield  {title} {\emph {\bibinfo {title} {Conditional
  expectation in an operator algebra, {{IV}} ({{Entropy}} and information)},}\
  }\href {http://dx.doi.org/10.2996/kmj/1138844604} {\bibfield  {journal}
  {\bibinfo  {journal} {Kodai Math. Sem. Rep.}\ }\textbf {\bibinfo {volume}
  {14}},\ \bibinfo {pages} {59} (\bibinfo {year} {1962})}\BibitemShut {NoStop}%
\bibitem [{\citenamefont {Hiai}\ and\ \citenamefont {Petz}(1991)}]{hiai_1991}%
  \BibitemOpen
  \bibfield  {author} {\bibinfo {author} {\bibfnamefont {F.}~\bibnamefont
  {Hiai}}\ and\ \bibinfo {author} {\bibfnamefont {D.}~\bibnamefont {Petz}},\
  }\bibfield  {title} {\emph {\bibinfo {title} {The proper formula for relative
  entropy and its asymptotics in quantum probability},}\ }\href
  {http://dx.doi.org/10.1007/BF02100287} {\bibfield  {journal} {\bibinfo
  {journal} {Commun. Math. Phys.}\ }\textbf {\bibinfo {volume} {143}},\
  \bibinfo {pages} {99} (\bibinfo {year} {1991})}\BibitemShut {NoStop}%
\bibitem [{\citenamefont {Ogawa}\ and\ \citenamefont
  {Nagaoka}(2000)}]{ogawa_2000}%
  \BibitemOpen
  \bibfield  {author} {\bibinfo {author} {\bibfnamefont {T.}~\bibnamefont
  {Ogawa}}\ and\ \bibinfo {author} {\bibfnamefont {H.}~\bibnamefont
  {Nagaoka}},\ }\bibfield  {title} {\emph {\bibinfo {title} {Strong converse
  and {{Stein}}'s lemma in quantum hypothesis testing},}\ }\href
  {http://dx.doi.org/10.1109/18.887855} {\bibfield  {journal} {\bibinfo
  {journal} {IEEE Trans. Inf. Theory}\ }\textbf {\bibinfo {volume} {46}},\
  \bibinfo {pages} {2428} (\bibinfo {year} {2000})}\BibitemShut {NoStop}%
\bibitem [{\citenamefont {Stein}(shed)}]{stein_unpublished}%
  \BibitemOpen
  \bibfield  {author} {\bibinfo {author} {\bibfnamefont {C.}~\bibnamefont
  {Stein}},\ }\href@noop {} {\emph {\bibinfo {title} {Information and
  {{Comparison}} of {{Experiments}}},}\ } (\bibinfo {year} {unpublished}).\EOS\
  \bibinfo {note} {Charles Stein papers (SC1224), Box 12, Folder 7, Department
  of Special Collections and University Archives, Stanford University
  Libraries.}\BibitemShut {Stop}%
\bibitem [{\citenamefont {Chernoff}(1956)}]{chernoff_1956}%
  \BibitemOpen
  \bibfield  {author} {\bibinfo {author} {\bibfnamefont {H.}~\bibnamefont
  {Chernoff}},\ }\bibfield  {title} {\emph {\bibinfo {title} {Large-{{Sample
  Theory}}: {{Parametric Case}}},}\ }\href
  {http://dx.doi.org/10.1214/aoms/1177728347} {\bibfield  {journal} {\bibinfo
  {journal} {Ann. Math. Stat.}\ }\textbf {\bibinfo {volume} {27}},\ \bibinfo
  {pages} {1} (\bibinfo {year} {1956})}\BibitemShut {NoStop}%
\bibitem [{\citenamefont {Kullback}\ and\ \citenamefont
  {Leibler}(1951)}]{kullback_1951}%
  \BibitemOpen
  \bibfield  {author} {\bibinfo {author} {\bibfnamefont {S.}~\bibnamefont
  {Kullback}}\ and\ \bibinfo {author} {\bibfnamefont {R.~A.}\ \bibnamefont
  {Leibler}},\ }\bibfield  {title} {\emph {\bibinfo {title} {On {{Information}}
  and {{Sufficiency}}},}\ }\href {http://dx.doi.org/10.1214/aoms/1177729694}
  {\bibfield  {journal} {\bibinfo  {journal} {Ann. Math. Stat.}\ }\textbf
  {\bibinfo {volume} {22}},\ \bibinfo {pages} {79} (\bibinfo {year}
  {1951})}\BibitemShut {NoStop}%
\bibitem [{\citenamefont {Chernoff}(1952)}]{chernoff_1952}%
  \BibitemOpen
  \bibfield  {author} {\bibinfo {author} {\bibfnamefont {H.}~\bibnamefont
  {Chernoff}},\ }\bibfield  {title} {\emph {\bibinfo {title} {A {{Measure}} of
  {{Asymptotic Efficiency}} for {{Tests}} of a {{Hypothesis Based}} on the sum
  of {{Observations}}},}\ }\href {http://dx.doi.org/10.1214/aoms/1177729330}
  {\bibfield  {journal} {\bibinfo  {journal} {Ann. Math. Stat.}\ }\textbf
  {\bibinfo {volume} {23}},\ \bibinfo {pages} {493} (\bibinfo {year}
  {1952})}\BibitemShut {NoStop}%
\bibitem [{\citenamefont {Nussbaum}\ and\ \citenamefont
  {Szko{\l}a}(2009)}]{nussbaum_2009}%
  \BibitemOpen
  \bibfield  {author} {\bibinfo {author} {\bibfnamefont {M.}~\bibnamefont
  {Nussbaum}}\ and\ \bibinfo {author} {\bibfnamefont {A.}~\bibnamefont
  {Szko{\l}a}},\ }\bibfield  {title} {\emph {\bibinfo {title} {The {{Chernoff}}
  lower bound for symmetric quantum hypothesis testing},}\ }\href
  {http://dx.doi.org/10.1214/08-AOS593} {\bibfield  {journal} {\bibinfo
  {journal} {Ann. Stat.}\ }\textbf {\bibinfo {volume} {37}},\ \bibinfo {pages}
  {1040} (\bibinfo {year} {2009})}\BibitemShut {NoStop}%
\bibitem [{\citenamefont {Audenaert}\ \emph {et~al.}(2007)\citenamefont
  {Audenaert}, \citenamefont {Calsamiglia}, \citenamefont {{Mu{\~n}oz-Tapia}},
  \citenamefont {Bagan}, \citenamefont {Masanes}, \citenamefont {Acin},\ and\
  \citenamefont {Verstraete}}]{audenaert_2007}%
  \BibitemOpen
  \bibfield  {author} {\bibinfo {author} {\bibfnamefont {K.~M.~R.}\
  \bibnamefont {Audenaert}}, \bibinfo {author} {\bibfnamefont {J.}~\bibnamefont
  {Calsamiglia}}, \bibinfo {author} {\bibfnamefont {R.}~\bibnamefont
  {{Mu{\~n}oz-Tapia}}}, \bibinfo {author} {\bibfnamefont {E.}~\bibnamefont
  {Bagan}}, \bibinfo {author} {\bibfnamefont {L.}~\bibnamefont {Masanes}},
  \bibinfo {author} {\bibfnamefont {A.}~\bibnamefont {Acin}}, \ and\ \bibinfo
  {author} {\bibfnamefont {F.}~\bibnamefont {Verstraete}},\ }\bibfield  {title}
  {\emph {\bibinfo {title} {Discriminating {{States}}: {{The Quantum Chernoff
  Bound}}},}\ }\href {http://dx.doi.org/10.1103/PhysRevLett.98.160501}
  {\bibfield  {journal} {\bibinfo  {journal} {Phys. Rev. Lett.}\ }\textbf
  {\bibinfo {volume} {98}},\ \bibinfo {pages} {160501} (\bibinfo {year}
  {2007})}\BibitemShut {NoStop}%
\bibitem [{\citenamefont {Brand{\~a}o}\ and\ \citenamefont
  {Plenio}(2010{\natexlab{a}})}]{brandao_2010-1}%
  \BibitemOpen
  \bibfield  {author} {\bibinfo {author} {\bibfnamefont {F.~G. S.~L.}\
  \bibnamefont {Brand{\~a}o}}\ and\ \bibinfo {author} {\bibfnamefont {M.~B.}\
  \bibnamefont {Plenio}},\ }\bibfield  {title} {\emph {\bibinfo {title} {A
  {{Generalization}} of {{Quantum Stein}}'s {{Lemma}}},}\ }\href
  {http://dx.doi.org/10.1007/s00220-010-1005-z} {\bibfield  {journal} {\bibinfo
   {journal} {Commun. Math. Phys.}\ }\textbf {\bibinfo {volume} {295}},\
  \bibinfo {pages} {791} (\bibinfo {year} {2010}{\natexlab{a}})}\BibitemShut
  {NoStop}%
\bibitem [{\citenamefont {Berta}\ \emph {et~al.}(2022)\citenamefont {Berta},
  \citenamefont {Brand{\~a}o}, \citenamefont {Gour}, \citenamefont {Lami},
  \citenamefont {Plenio}, \citenamefont {Regula},\ and\ \citenamefont
  {Tomamichel}}]{berta_2022}%
  \BibitemOpen
  \bibfield  {author} {\bibinfo {author} {\bibfnamefont {M.}~\bibnamefont
  {Berta}}, \bibinfo {author} {\bibfnamefont {F.~G. S.~L.}\ \bibnamefont
  {Brand{\~a}o}}, \bibinfo {author} {\bibfnamefont {G.}~\bibnamefont {Gour}},
  \bibinfo {author} {\bibfnamefont {L.}~\bibnamefont {Lami}}, \bibinfo {author}
  {\bibfnamefont {M.~B.}\ \bibnamefont {Plenio}}, \bibinfo {author}
  {\bibfnamefont {B.}~\bibnamefont {Regula}}, \ and\ \bibinfo {author}
  {\bibfnamefont {M.}~\bibnamefont {Tomamichel}},\ }\bibfield  {title} {\emph
  {\bibinfo {title} {On a gap in the proof of the generalised quantum
  {{Stein}}'s lemma and its consequences for the reversibility of quantum
  resources},}\ }\href {http://arxiv.org/abs/2205.02813} {\Eprint
  {http://arxiv.org/abs/2205.02813} {arXiv:2205.02813}  (\bibinfo {year}
  {2022})}\BibitemShut {NoStop}%
\bibitem [{\citenamefont {Tomamichel}\ and\ \citenamefont
  {Hayashi}(2018)}]{tomamichel_2018}%
  \BibitemOpen
  \bibfield  {author} {\bibinfo {author} {\bibfnamefont {M.}~\bibnamefont
  {Tomamichel}}\ and\ \bibinfo {author} {\bibfnamefont {M.}~\bibnamefont
  {Hayashi}},\ }\bibfield  {title} {\emph {\bibinfo {title} {Operational
  {{Interpretation}} of {{R\'enyi Information Measures}} via {{Composite
  Hypothesis Testing Against Product}} and {{Markov Distributions}}},}\ }\href
  {http://dx.doi.org/10.1109/TIT.2017.2776900} {\bibfield  {journal} {\bibinfo
  {journal} {IEEE Trans. Inf. Theory}\ }\textbf {\bibinfo {volume} {64}},\
  \bibinfo {pages} {1064} (\bibinfo {year} {2018})}\BibitemShut {NoStop}%
\bibitem [{\citenamefont {Berta}\ \emph {et~al.}(2021)\citenamefont {Berta},
  \citenamefont {Brand{\~a}o},\ and\ \citenamefont {Hirche}}]{berta_2021}%
  \BibitemOpen
  \bibfield  {author} {\bibinfo {author} {\bibfnamefont {M.}~\bibnamefont
  {Berta}}, \bibinfo {author} {\bibfnamefont {F.~G. S.~L.}\ \bibnamefont
  {Brand{\~a}o}}, \ and\ \bibinfo {author} {\bibfnamefont {C.}~\bibnamefont
  {Hirche}},\ }\bibfield  {title} {\emph {\bibinfo {title} {On {{Composite
  Quantum Hypothesis Testing}}},}\ }\href
  {http://dx.doi.org/10.1007/s00220-021-04133-8} {\bibfield  {journal}
  {\bibinfo  {journal} {Commun. Math. Phys.}\ }\textbf {\bibinfo {volume}
  {385}},\ \bibinfo {pages} {55} (\bibinfo {year} {2021})}\BibitemShut
  {NoStop}%
\bibitem [{\citenamefont {Brand{\~a}o}\ \emph {et~al.}(2020)\citenamefont
  {Brand{\~a}o}, \citenamefont {Harrow}, \citenamefont {Lee},\ and\
  \citenamefont {Peres}}]{brandao_2020}%
  \BibitemOpen
  \bibfield  {author} {\bibinfo {author} {\bibfnamefont {F.~G. S.~L.}\
  \bibnamefont {Brand{\~a}o}}, \bibinfo {author} {\bibfnamefont {A.~W.}\
  \bibnamefont {Harrow}}, \bibinfo {author} {\bibfnamefont {J.~R.}\
  \bibnamefont {Lee}}, \ and\ \bibinfo {author} {\bibfnamefont
  {Y.}~\bibnamefont {Peres}},\ }\bibfield  {title} {\emph {\bibinfo {title}
  {Adversarial {{Hypothesis Testing}} and a {{Quantum Stein}}'s {{Lemma}} for
  {{Restricted Measurements}}},}\ }\href
  {http://dx.doi.org/10.1109/TIT.2020.2979704} {\bibfield  {journal} {\bibinfo
  {journal} {IEEE Trans. Inf. Theory}\ }\textbf {\bibinfo {volume} {66}},\
  \bibinfo {pages} {5037} (\bibinfo {year} {2020})}\BibitemShut {NoStop}%
\bibitem [{\citenamefont {Mosonyi}\ \emph {et~al.}(2022)\citenamefont
  {Mosonyi}, \citenamefont {Szil{\'a}gyi},\ and\ \citenamefont
  {Weiner}}]{mosonyi_2021}%
  \BibitemOpen
  \bibfield  {author} {\bibinfo {author} {\bibfnamefont {M.}~\bibnamefont
  {Mosonyi}}, \bibinfo {author} {\bibfnamefont {Z.}~\bibnamefont
  {Szil{\'a}gyi}}, \ and\ \bibinfo {author} {\bibfnamefont {M.}~\bibnamefont
  {Weiner}},\ }\bibfield  {title} {\emph {\bibinfo {title} {On the {{Error
  Exponents}} of {{Binary State Discrimination With Composite Hypotheses}}},}\
  }\href {http://dx.doi.org/10.1109/TIT.2021.3125683} {\bibfield  {journal}
  {\bibinfo  {journal} {IEEE Trans. Inf. Theory}\ }\textbf {\bibinfo {volume}
  {68}},\ \bibinfo {pages} {1032} (\bibinfo {year} {2022})}\BibitemShut
  {NoStop}%
\bibitem [{\citenamefont {Chiribella}\ \emph {et~al.}(2008)\citenamefont
  {Chiribella}, \citenamefont {D'Ariano},\ and\ \citenamefont
  {Perinotti}}]{chiribella_2008-1}%
  \BibitemOpen
  \bibfield  {author} {\bibinfo {author} {\bibfnamefont {G.}~\bibnamefont
  {Chiribella}}, \bibinfo {author} {\bibfnamefont {G.~M.}\ \bibnamefont
  {D'Ariano}}, \ and\ \bibinfo {author} {\bibfnamefont {P.}~\bibnamefont
  {Perinotti}},\ }\bibfield  {title} {\emph {\bibinfo {title} {Quantum
  {{Circuit Architecture}}},}\ }\href
  {http://dx.doi.org/10.1103/PhysRevLett.101.060401} {\bibfield  {journal}
  {\bibinfo  {journal} {Phys. Rev. Lett.}\ }\textbf {\bibinfo {volume} {101}},\
  \bibinfo {pages} {060401} (\bibinfo {year} {2008})}\BibitemShut {NoStop}%
\bibitem [{\citenamefont {Wilde}\ \emph {et~al.}(2020)\citenamefont {Wilde},
  \citenamefont {Berta}, \citenamefont {Hirche},\ and\ \citenamefont
  {Kaur}}]{wilde_2020}%
  \BibitemOpen
  \bibfield  {author} {\bibinfo {author} {\bibfnamefont {M.~M.}\ \bibnamefont
  {Wilde}}, \bibinfo {author} {\bibfnamefont {M.}~\bibnamefont {Berta}},
  \bibinfo {author} {\bibfnamefont {C.}~\bibnamefont {Hirche}}, \ and\ \bibinfo
  {author} {\bibfnamefont {E.}~\bibnamefont {Kaur}},\ }\bibfield  {title}
  {\emph {\bibinfo {title} {Amortized channel divergence for asymptotic quantum
  channel discrimination},}\ }\href
  {http://dx.doi.org/10.1007/s11005-020-01297-7} {\bibfield  {journal}
  {\bibinfo  {journal} {Lett. Math. Phys.}\ }\textbf {\bibinfo {volume}
  {110}},\ \bibinfo {pages} {2277} (\bibinfo {year} {2020})}\BibitemShut
  {NoStop}%
\bibitem [{\citenamefont {Wang}\ and\ \citenamefont
  {Wilde}(2019{\natexlab{a}})}]{wang_2019-4}%
  \BibitemOpen
  \bibfield  {author} {\bibinfo {author} {\bibfnamefont {X.}~\bibnamefont
  {Wang}}\ and\ \bibinfo {author} {\bibfnamefont {M.~M.}\ \bibnamefont
  {Wilde}},\ }\bibfield  {title} {\emph {\bibinfo {title} {Resource theory of
  asymmetric distinguishability for quantum channels},}\ }\href
  {http://dx.doi.org/10.1103/PhysRevResearch.1.033169} {\bibfield  {journal}
  {\bibinfo  {journal} {Phys. Rev. Research}\ }\textbf {\bibinfo {volume}
  {1}},\ \bibinfo {pages} {033169} (\bibinfo {year}
  {2019}{\natexlab{a}})}\BibitemShut {NoStop}%
\bibitem [{\citenamefont {Fang}\ \emph {et~al.}(2020)\citenamefont {Fang},
  \citenamefont {Fawzi}, \citenamefont {Renner},\ and\ \citenamefont
  {Sutter}}]{fang_2020-2}%
  \BibitemOpen
  \bibfield  {author} {\bibinfo {author} {\bibfnamefont {K.}~\bibnamefont
  {Fang}}, \bibinfo {author} {\bibfnamefont {O.}~\bibnamefont {Fawzi}},
  \bibinfo {author} {\bibfnamefont {R.}~\bibnamefont {Renner}}, \ and\ \bibinfo
  {author} {\bibfnamefont {D.}~\bibnamefont {Sutter}},\ }\bibfield  {title}
  {\emph {\bibinfo {title} {Chain {{Rule}} for the {{Quantum Relative
  Entropy}}},}\ }\href {http://dx.doi.org/10.1103/PhysRevLett.124.100501}
  {\bibfield  {journal} {\bibinfo  {journal} {Phys. Rev. Lett.}\ }\textbf
  {\bibinfo {volume} {124}},\ \bibinfo {pages} {100501} (\bibinfo {year}
  {2020})}\BibitemShut {NoStop}%
\bibitem [{\citenamefont {Harrow}\ \emph {et~al.}(2010)\citenamefont {Harrow},
  \citenamefont {Hassidim}, \citenamefont {Leung},\ and\ \citenamefont
  {Watrous}}]{harrow_2010-1}%
  \BibitemOpen
  \bibfield  {author} {\bibinfo {author} {\bibfnamefont {A.~W.}\ \bibnamefont
  {Harrow}}, \bibinfo {author} {\bibfnamefont {A.}~\bibnamefont {Hassidim}},
  \bibinfo {author} {\bibfnamefont {D.~W.}\ \bibnamefont {Leung}}, \ and\
  \bibinfo {author} {\bibfnamefont {J.}~\bibnamefont {Watrous}},\ }\bibfield
  {title} {\emph {\bibinfo {title} {Adaptive versus nonadaptive strategies for
  quantum channel discrimination},}\ }\href
  {http://dx.doi.org/10.1103/PhysRevA.81.032339} {\bibfield  {journal}
  {\bibinfo  {journal} {Phys. Rev. A}\ }\textbf {\bibinfo {volume} {81}},\
  \bibinfo {pages} {032339} (\bibinfo {year} {2010})}\BibitemShut {NoStop}%
\bibitem [{\citenamefont {Salek}\ \emph {et~al.}(2022)\citenamefont {Salek},
  \citenamefont {Hayashi},\ and\ \citenamefont {Winter}}]{salek_2022}%
  \BibitemOpen
  \bibfield  {author} {\bibinfo {author} {\bibfnamefont {F.}~\bibnamefont
  {Salek}}, \bibinfo {author} {\bibfnamefont {M.}~\bibnamefont {Hayashi}}, \
  and\ \bibinfo {author} {\bibfnamefont {A.}~\bibnamefont {Winter}},\
  }\bibfield  {title} {\emph {\bibinfo {title} {Usefulness of adaptive
  strategies in asymptotic quantum channel discrimination},}\ }\href
  {http://dx.doi.org/10.1103/PhysRevA.105.022419} {\bibfield  {journal}
  {\bibinfo  {journal} {Phys. Rev. A}\ }\textbf {\bibinfo {volume} {105}},\
  \bibinfo {pages} {022419} (\bibinfo {year} {2022})}\BibitemShut {NoStop}%
\bibitem [{\citenamefont {Yu}\ and\ \citenamefont {Zhou}(2021)}]{yu_2021}%
  \BibitemOpen
  \bibfield  {author} {\bibinfo {author} {\bibfnamefont {N.}~\bibnamefont
  {Yu}}\ and\ \bibinfo {author} {\bibfnamefont {L.}~\bibnamefont {Zhou}},\
  }\bibfield  {title} {\emph {\bibinfo {title} {When is the {{Chernoff
  Exponent}} for {{Quantum Operations Finite}}?}}\ }\href
  {http://dx.doi.org/10.1109/TIT.2021.3067924} {\bibfield  {journal} {\bibinfo
  {journal} {IEEE Trans. Inf. Theory}\ }\textbf {\bibinfo {volume} {67}},\
  \bibinfo {pages} {4517} (\bibinfo {year} {2021})}\BibitemShut {NoStop}%
\bibitem [{\citenamefont {Chiribella}\ \emph {et~al.}(2013)\citenamefont
  {Chiribella}, \citenamefont {D'Ariano}, \citenamefont {Perinotti},\ and\
  \citenamefont {Valiron}}]{chiribella_2013}%
  \BibitemOpen
  \bibfield  {author} {\bibinfo {author} {\bibfnamefont {G.}~\bibnamefont
  {Chiribella}}, \bibinfo {author} {\bibfnamefont {G.~M.}\ \bibnamefont
  {D'Ariano}}, \bibinfo {author} {\bibfnamefont {P.}~\bibnamefont {Perinotti}},
  \ and\ \bibinfo {author} {\bibfnamefont {B.}~\bibnamefont {Valiron}},\
  }\bibfield  {title} {\emph {\bibinfo {title} {Quantum computations without
  definite causal structure},}\ }\href
  {http://dx.doi.org/10.1103/PhysRevA.88.022318} {\bibfield  {journal}
  {\bibinfo  {journal} {Phys. Rev. A}\ }\textbf {\bibinfo {volume} {88}},\
  \bibinfo {pages} {022318} (\bibinfo {year} {2013})}\BibitemShut {NoStop}%
\bibitem [{\citenamefont {Oreshkov}\ \emph {et~al.}(2012)\citenamefont
  {Oreshkov}, \citenamefont {Costa},\ and\ \citenamefont
  {Brukner}}]{oreshkov_2012}%
  \BibitemOpen
  \bibfield  {author} {\bibinfo {author} {\bibfnamefont {O.}~\bibnamefont
  {Oreshkov}}, \bibinfo {author} {\bibfnamefont {F.}~\bibnamefont {Costa}}, \
  and\ \bibinfo {author} {\bibfnamefont {{\v C}.}~\bibnamefont {Brukner}},\
  }\bibfield  {title} {\emph {\bibinfo {title} {Quantum correlations with no
  causal order},}\ }\href {http://dx.doi.org/10.1038/ncomms2076} {\bibfield
  {journal} {\bibinfo  {journal} {Nat. Commun.}\ }\textbf {\bibinfo {volume}
  {3}},\ \bibinfo {pages} {1092} (\bibinfo {year} {2012})}\BibitemShut
  {NoStop}%
\bibitem [{\citenamefont {Ebler}\ \emph {et~al.}(2018)\citenamefont {Ebler},
  \citenamefont {Salek},\ and\ \citenamefont {Chiribella}}]{ebler_2018}%
  \BibitemOpen
  \bibfield  {author} {\bibinfo {author} {\bibfnamefont {D.}~\bibnamefont
  {Ebler}}, \bibinfo {author} {\bibfnamefont {S.}~\bibnamefont {Salek}}, \ and\
  \bibinfo {author} {\bibfnamefont {G.}~\bibnamefont {Chiribella}},\ }\bibfield
   {title} {\emph {\bibinfo {title} {Enhanced {{Communication}} with the
  {{Assistance}} of {{Indefinite Causal Order}}},}\ }\href
  {http://dx.doi.org/10.1103/PhysRevLett.120.120502} {\bibfield  {journal}
  {\bibinfo  {journal} {Phys. Rev. Lett.}\ }\textbf {\bibinfo {volume} {120}},\
  \bibinfo {pages} {120502} (\bibinfo {year} {2018})}\BibitemShut {NoStop}%
\bibitem [{\citenamefont {Quintino}\ \emph {et~al.}(2019)\citenamefont
  {Quintino}, \citenamefont {Dong}, \citenamefont {Shimbo}, \citenamefont
  {Soeda},\ and\ \citenamefont {Murao}}]{quintino_2019}%
  \BibitemOpen
  \bibfield  {author} {\bibinfo {author} {\bibfnamefont {M.~T.}\ \bibnamefont
  {Quintino}}, \bibinfo {author} {\bibfnamefont {Q.}~\bibnamefont {Dong}},
  \bibinfo {author} {\bibfnamefont {A.}~\bibnamefont {Shimbo}}, \bibinfo
  {author} {\bibfnamefont {A.}~\bibnamefont {Soeda}}, \ and\ \bibinfo {author}
  {\bibfnamefont {M.}~\bibnamefont {Murao}},\ }\bibfield  {title} {\emph
  {\bibinfo {title} {Reversing {{Unknown Quantum Transformations}}: {{Universal
  Quantum Circuit}} for {{Inverting General Unitary Operations}}},}\ }\href
  {http://dx.doi.org/10.1103/PhysRevLett.123.210502} {\bibfield  {journal}
  {\bibinfo  {journal} {Phys. Rev. Lett.}\ }\textbf {\bibinfo {volume} {123}},\
  \bibinfo {pages} {210502} (\bibinfo {year} {2019})}\BibitemShut {NoStop}%
\bibitem [{\citenamefont {Bavaresco}\ \emph {et~al.}(2021)\citenamefont
  {Bavaresco}, \citenamefont {Murao},\ and\ \citenamefont
  {Quintino}}]{bavaresco_2021}%
  \BibitemOpen
  \bibfield  {author} {\bibinfo {author} {\bibfnamefont {J.}~\bibnamefont
  {Bavaresco}}, \bibinfo {author} {\bibfnamefont {M.}~\bibnamefont {Murao}}, \
  and\ \bibinfo {author} {\bibfnamefont {M.~T.}\ \bibnamefont {Quintino}},\
  }\bibfield  {title} {\emph {\bibinfo {title} {Strict {{Hierarchy}} between
  {{Parallel}}, {{Sequential}}, and {{Indefinite-Causal-Order Strategies}} for
  {{Channel Discrimination}}},}\ }\href
  {http://dx.doi.org/10.1103/PhysRevLett.127.200504} {\bibfield  {journal}
  {\bibinfo  {journal} {Phys. Rev. Lett.}\ }\textbf {\bibinfo {volume} {127}},\
  \bibinfo {pages} {200504} (\bibinfo {year} {2021})}\BibitemShut {NoStop}%
\bibitem [{\citenamefont {Blahut}(1974)}]{blahut_1974}%
  \BibitemOpen
  \bibfield  {author} {\bibinfo {author} {\bibfnamefont {R.}~\bibnamefont
  {Blahut}},\ }\bibfield  {title} {\emph {\bibinfo {title} {Hypothesis testing
  and information theory},}\ }\href
  {http://dx.doi.org/10.1109/TIT.1974.1055254} {\bibfield  {journal} {\bibinfo
  {journal} {IEEE Trans. Inf. Theory}\ }\textbf {\bibinfo {volume} {20}},\
  \bibinfo {pages} {405} (\bibinfo {year} {1974})}\BibitemShut {NoStop}%
\bibitem [{\citenamefont {Hayashi}(2009)}]{hayashi_2009}%
  \BibitemOpen
  \bibfield  {author} {\bibinfo {author} {\bibfnamefont {M.}~\bibnamefont
  {Hayashi}},\ }\bibfield  {title} {\emph {\bibinfo {title} {Discrimination of
  {{Two Channels}} by {{Adaptive Methods}} and {{Its Application}} to {{Quantum
  System}}},}\ }\href {http://dx.doi.org/10.1109/TIT.2009.2023726} {\bibfield
  {journal} {\bibinfo  {journal} {IEEE Trans. Inf. Theory}\ }\textbf {\bibinfo
  {volume} {55}},\ \bibinfo {pages} {3807} (\bibinfo {year}
  {2009})}\BibitemShut {NoStop}%
\bibitem [{\citenamefont {Ivanovic}(1987)}]{ivanovic_1987}%
  \BibitemOpen
  \bibfield  {author} {\bibinfo {author} {\bibfnamefont {I.~D.}\ \bibnamefont
  {Ivanovic}},\ }\bibfield  {title} {\emph {\bibinfo {title} {How to
  differentiate between non-orthogonal states},}\ }\href
  {http://dx.doi.org/10.1016/0375-9601(87)90222-2} {\bibfield  {journal}
  {\bibinfo  {journal} {Phys. Lett. A}\ }\textbf {\bibinfo {volume} {123}},\
  \bibinfo {pages} {257} (\bibinfo {year} {1987})}\BibitemShut {NoStop}%
\bibitem [{\citenamefont {Dieks}(1988)}]{dieks_1988}%
  \BibitemOpen
  \bibfield  {author} {\bibinfo {author} {\bibfnamefont {D.}~\bibnamefont
  {Dieks}},\ }\bibfield  {title} {\emph {\bibinfo {title} {Overlap and
  distinguishability of quantum states},}\ }\href
  {http://dx.doi.org/10.1016/0375-9601(88)90840-7} {\bibfield  {journal}
  {\bibinfo  {journal} {Phys. Lett. A}\ }\textbf {\bibinfo {volume} {126}},\
  \bibinfo {pages} {303} (\bibinfo {year} {1988})}\BibitemShut {NoStop}%
\bibitem [{\citenamefont {Peres}(1988)}]{peres_1988}%
  \BibitemOpen
  \bibfield  {author} {\bibinfo {author} {\bibfnamefont {A.}~\bibnamefont
  {Peres}},\ }\bibfield  {title} {\emph {\bibinfo {title} {How to differentiate
  between non-orthogonal states},}\ }\href
  {http://dx.doi.org/10.1016/0375-9601(88)91034-1} {\bibfield  {journal}
  {\bibinfo  {journal} {Phys. Lett. A}\ }\textbf {\bibinfo {volume} {128}},\
  \bibinfo {pages} {19} (\bibinfo {year} {1988})}\BibitemShut {NoStop}%
\bibitem [{\citenamefont {Chefles}\ and\ \citenamefont
  {Barnett}(1998)}]{chefles_1998}%
  \BibitemOpen
  \bibfield  {author} {\bibinfo {author} {\bibfnamefont {A.}~\bibnamefont
  {Chefles}}\ and\ \bibinfo {author} {\bibfnamefont {S.~M.}\ \bibnamefont
  {Barnett}},\ }\bibfield  {title} {\emph {\bibinfo {title} {Strategies for
  discriminating between non-orthogonal quantum states},}\ }\href
  {http://dx.doi.org/10.1080/09500349808230919} {\bibfield  {journal} {\bibinfo
   {journal} {J. Mod. Opt.}\ }\textbf {\bibinfo {volume} {45}},\ \bibinfo
  {pages} {1295} (\bibinfo {year} {1998})}\BibitemShut {NoStop}%
\bibitem [{\citenamefont {Fiur{\'a}{\v s}ek}\ and\ \citenamefont {Je{\v
  z}ek}(2003)}]{fiurasek_2003}%
  \BibitemOpen
  \bibfield  {author} {\bibinfo {author} {\bibfnamefont {J.}~\bibnamefont
  {Fiur{\'a}{\v s}ek}}\ and\ \bibinfo {author} {\bibfnamefont {M.}~\bibnamefont
  {Je{\v z}ek}},\ }\bibfield  {title} {\emph {\bibinfo {title} {Optimal
  discrimination of mixed quantum states involving inconclusive results},}\
  }\href {http://dx.doi.org/10.1103/PhysRevA.67.012321} {\bibfield  {journal}
  {\bibinfo  {journal} {Phys. Rev. A}\ }\textbf {\bibinfo {volume} {67}},\
  \bibinfo {pages} {012321} (\bibinfo {year} {2003})}\BibitemShut {NoStop}%
\bibitem [{\citenamefont {Rudolph}\ \emph {et~al.}(2003)\citenamefont
  {Rudolph}, \citenamefont {Spekkens},\ and\ \citenamefont
  {Turner}}]{rudolph_2003}%
  \BibitemOpen
  \bibfield  {author} {\bibinfo {author} {\bibfnamefont {T.}~\bibnamefont
  {Rudolph}}, \bibinfo {author} {\bibfnamefont {R.~W.}\ \bibnamefont
  {Spekkens}}, \ and\ \bibinfo {author} {\bibfnamefont {P.~S.}\ \bibnamefont
  {Turner}},\ }\bibfield  {title} {\emph {\bibinfo {title} {Unambiguous
  discrimination of mixed states},}\ }\href
  {http://dx.doi.org/10.1103/PhysRevA.68.010301} {\bibfield  {journal}
  {\bibinfo  {journal} {Phys. Rev. A}\ }\textbf {\bibinfo {volume} {68}},\
  \bibinfo {pages} {010301} (\bibinfo {year} {2003})}\BibitemShut {NoStop}%
\bibitem [{\citenamefont {Croke}\ \emph {et~al.}(2006)\citenamefont {Croke},
  \citenamefont {Andersson}, \citenamefont {Barnett}, \citenamefont {Gilson},\
  and\ \citenamefont {Jeffers}}]{croke_2006}%
  \BibitemOpen
  \bibfield  {author} {\bibinfo {author} {\bibfnamefont {S.}~\bibnamefont
  {Croke}}, \bibinfo {author} {\bibfnamefont {E.}~\bibnamefont {Andersson}},
  \bibinfo {author} {\bibfnamefont {S.~M.}\ \bibnamefont {Barnett}}, \bibinfo
  {author} {\bibfnamefont {C.~R.}\ \bibnamefont {Gilson}}, \ and\ \bibinfo
  {author} {\bibfnamefont {J.}~\bibnamefont {Jeffers}},\ }\bibfield  {title}
  {\emph {\bibinfo {title} {Maximum {{Confidence Quantum Measurements}}},}\
  }\href {http://dx.doi.org/10.1103/PhysRevLett.96.070401} {\bibfield
  {journal} {\bibinfo  {journal} {Phys. Rev. Lett.}\ }\textbf {\bibinfo
  {volume} {96}},\ \bibinfo {pages} {070401} (\bibinfo {year}
  {2006})}\BibitemShut {NoStop}%
\bibitem [{\citenamefont {Herzog}(2009)}]{herzog_2009}%
  \BibitemOpen
  \bibfield  {author} {\bibinfo {author} {\bibfnamefont {U.}~\bibnamefont
  {Herzog}},\ }\bibfield  {title} {\emph {\bibinfo {title} {Discrimination of
  two mixed quantum states with maximum confidence and minimum probability of
  inconclusive results},}\ }\href
  {http://dx.doi.org/10.1103/PhysRevA.79.032323} {\bibfield  {journal}
  {\bibinfo  {journal} {Phys. Rev. A}\ }\textbf {\bibinfo {volume} {79}},\
  \bibinfo {pages} {032323} (\bibinfo {year} {2009})}\BibitemShut {NoStop}%
\bibitem [{\citenamefont {Bagan}\ \emph {et~al.}(2012)\citenamefont {Bagan},
  \citenamefont {{Mu{\~n}oz-Tapia}}, \citenamefont {{Olivares-Renter{\'i}a}},\
  and\ \citenamefont {Bergou}}]{bagan_2012}%
  \BibitemOpen
  \bibfield  {author} {\bibinfo {author} {\bibfnamefont {E.}~\bibnamefont
  {Bagan}}, \bibinfo {author} {\bibfnamefont {R.}~\bibnamefont
  {{Mu{\~n}oz-Tapia}}}, \bibinfo {author} {\bibfnamefont {G.~A.}\ \bibnamefont
  {{Olivares-Renter{\'i}a}}}, \ and\ \bibinfo {author} {\bibfnamefont {J.~A.}\
  \bibnamefont {Bergou}},\ }\bibfield  {title} {\emph {\bibinfo {title}
  {Optimal discrimination of quantum states with a fixed rate of inconclusive
  outcomes},}\ }\href {http://dx.doi.org/10.1103/PhysRevA.86.040303} {\bibfield
   {journal} {\bibinfo  {journal} {Phys. Rev. A}\ }\textbf {\bibinfo {volume}
  {86}},\ \bibinfo {pages} {040303} (\bibinfo {year} {2012})}\BibitemShut
  {NoStop}%
\bibitem [{\citenamefont {Zhuang}(2020)}]{zhuang_2020}%
  \BibitemOpen
  \bibfield  {author} {\bibinfo {author} {\bibfnamefont {Q.}~\bibnamefont
  {Zhuang}},\ }\bibfield  {title} {\emph {\bibinfo {title} {Ultimate limits of
  approximate unambiguous discrimination},}\ }\href
  {http://dx.doi.org/10.1103/PhysRevResearch.2.043276} {\bibfield  {journal}
  {\bibinfo  {journal} {Phys. Rev. Research}\ }\textbf {\bibinfo {volume}
  {2}},\ \bibinfo {pages} {043276} (\bibinfo {year} {2020})}\BibitemShut
  {NoStop}%
\bibitem [{\citenamefont {Barnett}\ and\ \citenamefont
  {Croke}(2009)}]{barnett_2009}%
  \BibitemOpen
  \bibfield  {author} {\bibinfo {author} {\bibfnamefont {S.~M.}\ \bibnamefont
  {Barnett}}\ and\ \bibinfo {author} {\bibfnamefont {S.}~\bibnamefont
  {Croke}},\ }\bibfield  {title} {\emph {\bibinfo {title} {Quantum state
  discrimination},}\ }\href {http://dx.doi.org/10.1364/AOP.1.000238} {\bibfield
   {journal} {\bibinfo  {journal} {Adv. Opt. Photon.}\ }\textbf {\bibinfo
  {volume} {1}},\ \bibinfo {pages} {238} (\bibinfo {year} {2009})}\BibitemShut
  {NoStop}%
\bibitem [{\citenamefont {Bae}\ and\ \citenamefont {Kwek}(2015)}]{bae_2015}%
  \BibitemOpen
  \bibfield  {author} {\bibinfo {author} {\bibfnamefont {J.}~\bibnamefont
  {Bae}}\ and\ \bibinfo {author} {\bibfnamefont {L.-C.}\ \bibnamefont {Kwek}},\
  }\bibfield  {title} {\emph {\bibinfo {title} {Quantum state discrimination
  and its applications},}\ }\href
  {http://dx.doi.org/10.1088/1751-8113/48/8/083001} {\bibfield  {journal}
  {\bibinfo  {journal} {J. Phys. A: Math. Theor.}\ }\textbf {\bibinfo {volume}
  {48}},\ \bibinfo {pages} {083001} (\bibinfo {year} {2015})}\BibitemShut
  {NoStop}%
\bibitem [{\citenamefont {Nikulin}(1986)}]{nikulin_1986}%
  \BibitemOpen
  \bibfield  {author} {\bibinfo {author} {\bibfnamefont {M.~S.}\ \bibnamefont
  {Nikulin}},\ }\bibfield  {title} {\emph {\bibinfo {title} {A result of {L}.\
  {N}.\ {B}ol'shev from the theory of the statistical testing of hypotheses},}\
  }\href@noop {} {\bibfield  {journal} {\bibinfo  {journal} {Zap. Nauchn. Sem.
  Leningr. Otd. Mat. Inst.}\ }\textbf {\bibinfo {volume} {153}},\ \bibinfo
  {pages} {129} (\bibinfo {year} {1986})}\BibitemShut {NoStop}%
\bibitem [{\citenamefont {Gutman}(1989)}]{gutman_1989}%
  \BibitemOpen
  \bibfield  {author} {\bibinfo {author} {\bibfnamefont {M.}~\bibnamefont
  {Gutman}},\ }\bibfield  {title} {\emph {\bibinfo {title} {Asymptotically
  optimal classification for multiple tests with empirically observed
  statistics},}\ }\href {http://dx.doi.org/10.1109/18.32134} {\bibfield
  {journal} {\bibinfo  {journal} {IEEE Trans. Inf. Theory}\ }\textbf {\bibinfo
  {volume} {35}},\ \bibinfo {pages} {401} (\bibinfo {year} {1989})}\BibitemShut
  {NoStop}%
\bibitem [{\citenamefont {Datta}(2009)}]{datta_2009}%
  \BibitemOpen
  \bibfield  {author} {\bibinfo {author} {\bibfnamefont {N.}~\bibnamefont
  {Datta}},\ }\bibfield  {title} {\emph {\bibinfo {title} {Min- and
  {{Max-Relative Entropies}} and a {{New Entanglement Monotone}}},}\ }\href
  {http://dx.doi.org/10.1109/TIT.2009.2018325} {\bibfield  {journal} {\bibinfo
  {journal} {IEEE Trans. Inf. Theory}\ }\textbf {\bibinfo {volume} {55}},\
  \bibinfo {pages} {2816} (\bibinfo {year} {2009})}\BibitemShut {NoStop}%
\bibitem [{\citenamefont {{M{\"u}ller-Lennert}}\ \emph
  {et~al.}(2013)\citenamefont {{M{\"u}ller-Lennert}}, \citenamefont {Dupuis},
  \citenamefont {Szehr}, \citenamefont {Fehr},\ and\ \citenamefont
  {Tomamichel}}]{muller-lennert_2013}%
  \BibitemOpen
  \bibfield  {author} {\bibinfo {author} {\bibfnamefont {M.}~\bibnamefont
  {{M{\"u}ller-Lennert}}}, \bibinfo {author} {\bibfnamefont {F.}~\bibnamefont
  {Dupuis}}, \bibinfo {author} {\bibfnamefont {O.}~\bibnamefont {Szehr}},
  \bibinfo {author} {\bibfnamefont {S.}~\bibnamefont {Fehr}}, \ and\ \bibinfo
  {author} {\bibfnamefont {M.}~\bibnamefont {Tomamichel}},\ }\bibfield  {title}
  {\emph {\bibinfo {title} {On quantum {{R\'enyi}} entropies: {{A}} new
  generalization and some properties},}\ }\href
  {http://dx.doi.org/10.1063/1.4838856} {\bibfield  {journal} {\bibinfo
  {journal} {J. Math. Phys.}\ }\textbf {\bibinfo {volume} {54}},\ \bibinfo
  {pages} {122203} (\bibinfo {year} {2013})}\BibitemShut {NoStop}%
\bibitem [{\citenamefont {Bushell}(1973)}]{bushell_1973}%
  \BibitemOpen
  \bibfield  {author} {\bibinfo {author} {\bibfnamefont {P.~J.}\ \bibnamefont
  {Bushell}},\ }\bibfield  {title} {\emph {\bibinfo {title} {Hilbert's metric
  and positive contraction mappings in a {{Banach}} space},}\ }\href
  {http://dx.doi.org/10.1007/BF00247467} {\bibfield  {journal} {\bibinfo
  {journal} {Arch. Rat. Mech. Anal.}\ }\textbf {\bibinfo {volume} {52}},\
  \bibinfo {pages} {330} (\bibinfo {year} {1973})}\BibitemShut {NoStop}%
\bibitem [{\citenamefont {Thompson}(1963)}]{thompson_1963}%
  \BibitemOpen
  \bibfield  {author} {\bibinfo {author} {\bibfnamefont {A.~C.}\ \bibnamefont
  {Thompson}},\ }\bibfield  {title} {\emph {\bibinfo {title} {On certain
  contraction mappings in a partially ordered vector space.}}\ }\href
  {http://dx.doi.org/10.1090/S0002-9939-1963-0149237-7} {\bibfield  {journal}
  {\bibinfo  {journal} {Proc. Amer. Math. Soc.}\ }\textbf {\bibinfo {volume}
  {14}},\ \bibinfo {pages} {438} (\bibinfo {year} {1963})}\BibitemShut
  {NoStop}%
\bibitem [{\citenamefont {Nussbaum}\ and\ \citenamefont
  {Walsh}(2004)}]{nussbaum_2004}%
  \BibitemOpen
  \bibfield  {author} {\bibinfo {author} {\bibfnamefont {R.~D.}\ \bibnamefont
  {Nussbaum}}\ and\ \bibinfo {author} {\bibfnamefont {C.}~\bibnamefont
  {Walsh}},\ }\bibfield  {title} {\emph {\bibinfo {title} {A metric inequality
  for the {{Thompson}} and {{Hilbert}} geometries.}}\ }\href
  {https://eudml.org/doc/124572} {\bibfield  {journal} {\bibinfo  {journal} {J.
  Inequal. Pure Appl. Math.}\ }\textbf {\bibinfo {volume} {5}},\ \bibinfo
  {pages} {54} (\bibinfo {year} {2004})}\BibitemShut {NoStop}%
\bibitem [{\citenamefont {Reeb}\ \emph {et~al.}(2011)\citenamefont {Reeb},
  \citenamefont {Kastoryano},\ and\ \citenamefont {Wolf}}]{reeb_2011}%
  \BibitemOpen
  \bibfield  {author} {\bibinfo {author} {\bibfnamefont {D.}~\bibnamefont
  {Reeb}}, \bibinfo {author} {\bibfnamefont {M.~J.}\ \bibnamefont
  {Kastoryano}}, \ and\ \bibinfo {author} {\bibfnamefont {M.~M.}\ \bibnamefont
  {Wolf}},\ }\bibfield  {title} {\emph {\bibinfo {title} {Hilbert's projective
  metric in quantum information theory},}\ }\href
  {http://dx.doi.org/10.1063/1.3615729} {\bibfield  {journal} {\bibinfo
  {journal} {J. Math. Phys.}\ }\textbf {\bibinfo {volume} {52}},\ \bibinfo
  {pages} {082201} (\bibinfo {year} {2011})}\BibitemShut {NoStop}%
\bibitem [{\citenamefont {Regula}(2022{\natexlab{a}})}]{regula_2022}%
  \BibitemOpen
  \bibfield  {author} {\bibinfo {author} {\bibfnamefont {B.}~\bibnamefont
  {Regula}},\ }\bibfield  {title} {\emph {\bibinfo {title} {Probabilistic
  {{Transformations}} of {{Quantum Resources}}},}\ }\href
  {http://dx.doi.org/10.1103/PhysRevLett.128.110505} {\bibfield  {journal}
  {\bibinfo  {journal} {Phys. Rev. Lett.}\ }\textbf {\bibinfo {volume} {128}},\
  \bibinfo {pages} {110505} (\bibinfo {year} {2022}{\natexlab{a}})}\BibitemShut
  {NoStop}%
\bibitem [{\citenamefont {Regula}\ \emph {et~al.}(2023)\citenamefont {Regula},
  \citenamefont {Lami},\ and\ \citenamefont {Wilde}}]{regula_2022-1}%
  \BibitemOpen
  \bibfield  {author} {\bibinfo {author} {\bibfnamefont {B.}~\bibnamefont
  {Regula}}, \bibinfo {author} {\bibfnamefont {L.}~\bibnamefont {Lami}}, \ and\
  \bibinfo {author} {\bibfnamefont {M.~M.}\ \bibnamefont {Wilde}},\ }\bibfield
  {title} {\emph {\bibinfo {title} {Overcoming entropic limitations on
  asymptotic state transformations through probabilistic protocols},}\ }\href
  {http://dx.doi.org/10.1103/PhysRevA.107.042401} {\bibfield  {journal}
  {\bibinfo  {journal} {Phys. Rev. A}\ }\textbf {\bibinfo {volume} {107}},\
  \bibinfo {pages} {042401} (\bibinfo {year} {2023})}\BibitemShut {NoStop}%
\bibitem [{\citenamefont {Regula}(2022{\natexlab{b}})}]{regula_2021-4}%
  \BibitemOpen
  \bibfield  {author} {\bibinfo {author} {\bibfnamefont {B.}~\bibnamefont
  {Regula}},\ }\bibfield  {title} {\emph {\bibinfo {title} {Tight constraints
  on probabilistic convertibility of quantum states},}\ }\href
  {http://dx.doi.org/10.22331/q-2022-09-22-817} {\bibfield  {journal} {\bibinfo
   {journal} {Quantum}\ }\textbf {\bibinfo {volume} {6}},\ \bibinfo {pages}
  {817} (\bibinfo {year} {2022}{\natexlab{b}})}\BibitemShut {NoStop}%
\bibitem [{\citenamefont {Salzmann}\ \emph {et~al.}(2021)\citenamefont
  {Salzmann}, \citenamefont {Datta}, \citenamefont {Gour}, \citenamefont
  {Wang},\ and\ \citenamefont {Wilde}}]{salzmann_2021}%
  \BibitemOpen
  \bibfield  {author} {\bibinfo {author} {\bibfnamefont {R.}~\bibnamefont
  {Salzmann}}, \bibinfo {author} {\bibfnamefont {N.}~\bibnamefont {Datta}},
  \bibinfo {author} {\bibfnamefont {G.}~\bibnamefont {Gour}}, \bibinfo {author}
  {\bibfnamefont {X.}~\bibnamefont {Wang}}, \ and\ \bibinfo {author}
  {\bibfnamefont {M.~M.}\ \bibnamefont {Wilde}},\ }\bibfield  {title} {\emph
  {\bibinfo {title} {Symmetric distinguishability as a quantum resource},}\
  }\href {http://dx.doi.org/10.1088/1367-2630/ac14aa} {\bibfield  {journal}
  {\bibinfo  {journal} {New J. Phys.}\ }\textbf {\bibinfo {volume} {23}},\
  \bibinfo {pages} {083016} (\bibinfo {year} {2021})}\BibitemShut {NoStop}%
\bibitem [{\citenamefont {Li}(2014)}]{li_2014}%
  \BibitemOpen
  \bibfield  {author} {\bibinfo {author} {\bibfnamefont {K.}~\bibnamefont
  {Li}},\ }\bibfield  {title} {\emph {\bibinfo {title} {Second-order
  asymptotics for quantum hypothesis testing},}\ }\href
  {http://dx.doi.org/10.1214/13-AOS1185} {\bibfield  {journal} {\bibinfo
  {journal} {Ann. Stat.}\ }\textbf {\bibinfo {volume} {42}},\ \bibinfo {pages}
  {171} (\bibinfo {year} {2014})}\BibitemShut {NoStop}%
\bibitem [{\citenamefont {Tomamichel}\ and\ \citenamefont
  {Hayashi}(2013)}]{tomamichel_2013}%
  \BibitemOpen
  \bibfield  {author} {\bibinfo {author} {\bibfnamefont {M.}~\bibnamefont
  {Tomamichel}}\ and\ \bibinfo {author} {\bibfnamefont {M.}~\bibnamefont
  {Hayashi}},\ }\bibfield  {title} {\emph {\bibinfo {title} {A {{Hierarchy}} of
  {{Information Quantities}} for {{Finite Block Length Analysis}} of {{Quantum
  Tasks}}},}\ }\href {http://dx.doi.org/10.1109/TIT.2013.2276628} {\bibfield
  {journal} {\bibinfo  {journal} {IEEE Trans. Inf. Theory}\ }\textbf {\bibinfo
  {volume} {59}},\ \bibinfo {pages} {7693} (\bibinfo {year}
  {2013})}\BibitemShut {NoStop}%
\bibitem [{\citenamefont {Fekete}(1923)}]{fekete_1923}%
  \BibitemOpen
  \bibfield  {author} {\bibinfo {author} {\bibfnamefont {M.}~\bibnamefont
  {Fekete}},\ }\bibfield  {title} {\emph {\bibinfo {title} {{\"Uber die
  Verteilung der Wurzeln bei gewissen algebraischen Gleichungen mit
  ganzzahligen Koeffizienten}},}\ }\href {http://dx.doi.org/10.1007/BF01504345}
  {\bibfield  {journal} {\bibinfo  {journal} {Math Z}\ }\textbf {\bibinfo
  {volume} {17}},\ \bibinfo {pages} {228} (\bibinfo {year} {1923})}\BibitemShut
  {NoStop}%
\bibitem [{\citenamefont {Brand{\~a}o}\ and\ \citenamefont
  {Plenio}(2010{\natexlab{b}})}]{brandao_2010}%
  \BibitemOpen
  \bibfield  {author} {\bibinfo {author} {\bibfnamefont {F.~G. S.~L.}\
  \bibnamefont {Brand{\~a}o}}\ and\ \bibinfo {author} {\bibfnamefont {M.~B.}\
  \bibnamefont {Plenio}},\ }\bibfield  {title} {\emph {\bibinfo {title} {A
  {{Reversible Theory}} of {{Entanglement}} and its {{Relation}} to the
  {{Second Law}}},}\ }\href {http://dx.doi.org/10.1007/s00220-010-1003-1}
  {\bibfield  {journal} {\bibinfo  {journal} {Commun. Math. Phys.}\ }\textbf
  {\bibinfo {volume} {295}},\ \bibinfo {pages} {829} (\bibinfo {year}
  {2010}{\natexlab{b}})}\BibitemShut {NoStop}%
\bibitem [{\citenamefont {Rockafellar}(1970)}]{rockafellar_1970}%
  \BibitemOpen
  \bibfield  {author} {\bibinfo {author} {\bibfnamefont {R.~T.}\ \bibnamefont
  {Rockafellar}},\ }\href@noop {} {\emph {\bibinfo {title} {Convex
  {{Analysis}}}}}\ (\bibinfo  {publisher} {{Princeton University Press}},\
  \bibinfo {address} {{Princeton}},\ \bibinfo {year} {1970})\BibitemShut
  {NoStop}%
\bibitem [{\citenamefont {Kitaev}(1997)}]{kitaev_1997}%
  \BibitemOpen
  \bibfield  {author} {\bibinfo {author} {\bibfnamefont {A.~Y.}\ \bibnamefont
  {Kitaev}},\ }\bibfield  {title} {\emph {\bibinfo {title} {Quantum
  computations: Algorithms and error correction},}\ }\href
  {http://dx.doi.org/10.1070/RM1997v052n06ABEH002155} {\bibfield  {journal}
  {\bibinfo  {journal} {Russ. Math. Surv.}\ }\textbf {\bibinfo {volume} {52}},\
  \bibinfo {pages} {1191} (\bibinfo {year} {1997})}\BibitemShut {NoStop}%
\bibitem [{\citenamefont {Duan}\ \emph {et~al.}(2009)\citenamefont {Duan},
  \citenamefont {Feng},\ and\ \citenamefont {Ying}}]{duan_2009}%
  \BibitemOpen
  \bibfield  {author} {\bibinfo {author} {\bibfnamefont {R.}~\bibnamefont
  {Duan}}, \bibinfo {author} {\bibfnamefont {Y.}~\bibnamefont {Feng}}, \ and\
  \bibinfo {author} {\bibfnamefont {M.}~\bibnamefont {Ying}},\ }\bibfield
  {title} {\emph {\bibinfo {title} {Perfect {{Distinguishability}} of {{Quantum
  Operations}}},}\ }\href {http://dx.doi.org/10.1103/PhysRevLett.103.210501}
  {\bibfield  {journal} {\bibinfo  {journal} {Phys. Rev. Lett.}\ }\textbf
  {\bibinfo {volume} {103}},\ \bibinfo {pages} {210501} (\bibinfo {year}
  {2009})}\BibitemShut {NoStop}%
\bibitem [{\citenamefont {Cooney}\ \emph {et~al.}(2016)\citenamefont {Cooney},
  \citenamefont {Mosonyi},\ and\ \citenamefont {Wilde}}]{cooney_2016}%
  \BibitemOpen
  \bibfield  {author} {\bibinfo {author} {\bibfnamefont {T.}~\bibnamefont
  {Cooney}}, \bibinfo {author} {\bibfnamefont {M.}~\bibnamefont {Mosonyi}}, \
  and\ \bibinfo {author} {\bibfnamefont {M.~M.}\ \bibnamefont {Wilde}},\
  }\bibfield  {title} {\emph {\bibinfo {title} {Strong {{Converse Exponents}}
  for a {{Quantum Channel Discrimination Problem}} and
  {{Quantum-Feedback-Assisted Communication}}},}\ }\href
  {http://dx.doi.org/10.1007/s00220-016-2645-4} {\bibfield  {journal} {\bibinfo
   {journal} {Commun. Math. Phys.}\ }\textbf {\bibinfo {volume} {344}},\
  \bibinfo {pages} {797} (\bibinfo {year} {2016})}\BibitemShut {NoStop}%
\bibitem [{\citenamefont {Pirandola}\ \emph {et~al.}(2019)\citenamefont
  {Pirandola}, \citenamefont {Laurenza}, \citenamefont {Lupo},\ and\
  \citenamefont {Pereira}}]{pirandola_2019-2}%
  \BibitemOpen
  \bibfield  {author} {\bibinfo {author} {\bibfnamefont {S.}~\bibnamefont
  {Pirandola}}, \bibinfo {author} {\bibfnamefont {R.}~\bibnamefont {Laurenza}},
  \bibinfo {author} {\bibfnamefont {C.}~\bibnamefont {Lupo}}, \ and\ \bibinfo
  {author} {\bibfnamefont {J.~L.}\ \bibnamefont {Pereira}},\ }\bibfield
  {title} {\emph {\bibinfo {title} {Fundamental limits to quantum channel
  discrimination},}\ }\href {http://dx.doi.org/10.1038/s41534-019-0162-y}
  {\bibfield  {journal} {\bibinfo  {journal} {npj Quantum Inf.}\ }\textbf
  {\bibinfo {volume} {5}},\ \bibinfo {pages} {1} (\bibinfo {year}
  {2019})}\BibitemShut {NoStop}%
\bibitem [{\citenamefont {Regula}\ and\ \citenamefont
  {Takagi}(2021)}]{regula_2021-1}%
  \BibitemOpen
  \bibfield  {author} {\bibinfo {author} {\bibfnamefont {B.}~\bibnamefont
  {Regula}}\ and\ \bibinfo {author} {\bibfnamefont {R.}~\bibnamefont
  {Takagi}},\ }\bibfield  {title} {\emph {\bibinfo {title} {Fundamental
  limitations on distillation of quantum channel resources},}\ }\href
  {http://dx.doi.org/10.1038/s41467-021-24699-0} {\bibfield  {journal}
  {\bibinfo  {journal} {Nat. Commun.}\ }\textbf {\bibinfo {volume} {12}},\
  \bibinfo {pages} {4411} (\bibinfo {year} {2021})}\BibitemShut {NoStop}%
\bibitem [{\citenamefont {Matsumoto}(2010)}]{matsumoto_2010}%
  \BibitemOpen
  \bibfield  {author} {\bibinfo {author} {\bibfnamefont {K.}~\bibnamefont
  {Matsumoto}},\ }\bibfield  {title} {\emph {\bibinfo {title} {Reverse {{Test}}
  and {{Characterization}} of {{Quantum Relative Entropy}}},}\ }\href@noop {}
  {\Eprint {http://arxiv.org/abs/1010.1030} {arXiv:1010.1030}  (\bibinfo {year}
  {2010})}\BibitemShut {NoStop}%
\bibitem [{\citenamefont {Wang}\ and\ \citenamefont
  {Wilde}(2019{\natexlab{b}})}]{wang_2019}%
  \BibitemOpen
  \bibfield  {author} {\bibinfo {author} {\bibfnamefont {X.}~\bibnamefont
  {Wang}}\ and\ \bibinfo {author} {\bibfnamefont {M.~M.}\ \bibnamefont
  {Wilde}},\ }\bibfield  {title} {\emph {\bibinfo {title} {Resource theory of
  asymmetric distinguishability},}\ }\href
  {http://dx.doi.org/10.1103/PhysRevResearch.1.033170} {\bibfield  {journal}
  {\bibinfo  {journal} {Phys. Rev. Research}\ }\textbf {\bibinfo {volume}
  {1}},\ \bibinfo {pages} {033170} (\bibinfo {year}
  {2019}{\natexlab{b}})}\BibitemShut {NoStop}%
\bibitem [{\citenamefont {Chitambar}\ and\ \citenamefont
  {Gour}(2019)}]{chitambar_2019}%
  \BibitemOpen
  \bibfield  {author} {\bibinfo {author} {\bibfnamefont {E.}~\bibnamefont
  {Chitambar}}\ and\ \bibinfo {author} {\bibfnamefont {G.}~\bibnamefont
  {Gour}},\ }\bibfield  {title} {\emph {\bibinfo {title} {Quantum resource
  theories},}\ }\href {http://dx.doi.org/10.1103/RevModPhys.91.025001}
  {\bibfield  {journal} {\bibinfo  {journal} {Rev. Mod. Phys.}\ }\textbf
  {\bibinfo {volume} {91}},\ \bibinfo {pages} {025001} (\bibinfo {year}
  {2019})}\BibitemShut {NoStop}%
\bibitem [{\citenamefont {Wang}\ and\ \citenamefont
  {Renner}(2012)}]{wang_2012}%
  \BibitemOpen
  \bibfield  {author} {\bibinfo {author} {\bibfnamefont {L.}~\bibnamefont
  {Wang}}\ and\ \bibinfo {author} {\bibfnamefont {R.}~\bibnamefont {Renner}},\
  }\bibfield  {title} {\emph {\bibinfo {title} {One-{{Shot Classical-Quantum
  Capacity}} and {{Hypothesis Testing}}},}\ }\href
  {http://dx.doi.org/10.1103/PhysRevLett.108.200501} {\bibfield  {journal}
  {\bibinfo  {journal} {Phys. Rev. Lett.}\ }\textbf {\bibinfo {volume} {108}},\
  \bibinfo {pages} {200501} (\bibinfo {year} {2012})}\BibitemShut {NoStop}%
\bibitem [{\citenamefont {Buscemi}\ and\ \citenamefont
  {Datta}(2010)}]{buscemi_2010}%
  \BibitemOpen
  \bibfield  {author} {\bibinfo {author} {\bibfnamefont {F.}~\bibnamefont
  {Buscemi}}\ and\ \bibinfo {author} {\bibfnamefont {N.}~\bibnamefont
  {Datta}},\ }\bibfield  {title} {\emph {\bibinfo {title} {The {{Quantum
  Capacity}} of {{Channels With Arbitrarily Correlated Noise}}},}\ }\href
  {http://dx.doi.org/10.1109/TIT.2009.2039166} {\bibfield  {journal} {\bibinfo
  {journal} {IEEE Trans. Inf. Theory}\ }\textbf {\bibinfo {volume} {56}},\
  \bibinfo {pages} {1447} (\bibinfo {year} {2010})}\BibitemShut {NoStop}%
\bibitem [{\citenamefont {Buscemi}\ \emph {et~al.}(2019)\citenamefont
  {Buscemi}, \citenamefont {Sutter},\ and\ \citenamefont
  {Tomamichel}}]{buscemi_2019}%
  \BibitemOpen
  \bibfield  {author} {\bibinfo {author} {\bibfnamefont {F.}~\bibnamefont
  {Buscemi}}, \bibinfo {author} {\bibfnamefont {D.}~\bibnamefont {Sutter}}, \
  and\ \bibinfo {author} {\bibfnamefont {M.}~\bibnamefont {Tomamichel}},\
  }\bibfield  {title} {\emph {\bibinfo {title} {An information-theoretic
  treatment of quantum dichotomies},}\ }\href
  {http://dx.doi.org/10.22331/q-2019-12-09-209} {\bibfield  {journal} {\bibinfo
   {journal} {Quantum}\ }\textbf {\bibinfo {volume} {3}},\ \bibinfo {pages}
  {209} (\bibinfo {year} {2019})}\BibitemShut {NoStop}%
\bibitem [{\citenamefont {Ludwig}(1985)}]{ludwig_1985}%
  \BibitemOpen
  \bibfield  {author} {\bibinfo {author} {\bibfnamefont {G.}~\bibnamefont
  {Ludwig}},\ }\href@noop {} {\emph {\bibinfo {title} {An {{Axiomatic Basis}}
  for {{Quantum Mechanics}}: {{Volume}} 1 {{Derivation}} of {{Hilbert Space
  Structure}}}}}\ (\bibinfo  {publisher} {{Springer-Verlag}},\ \bibinfo
  {address} {{Berlin Heidelberg}},\ \bibinfo {year} {1985})\BibitemShut
  {NoStop}%
\bibitem [{\citenamefont {Hartk{\"a}mper}\ and\ \citenamefont
  {Neumann}(1974)}]{hartkamper_1974}%
  \BibitemOpen
  \bibinfo {editor} {\bibfnamefont {A.}~\bibnamefont {Hartk{\"a}mper}}\ and\
  \bibinfo {editor} {\bibfnamefont {H.}~\bibnamefont {Neumann}},\ eds.,\
  \href@noop {} {\emph {\bibinfo {title} {Foundations of {{Quantum Mechanics}}
  and {{Ordered Linear Spaces}}}}}\ (\bibinfo  {publisher} {{Springer}},\
  \bibinfo {year} {1974})\BibitemShut {NoStop}%
\bibitem [{\citenamefont {Davies}\ and\ \citenamefont
  {Lewis}(1970)}]{davies_1970}%
  \BibitemOpen
  \bibfield  {author} {\bibinfo {author} {\bibfnamefont {E.~B.}\ \bibnamefont
  {Davies}}\ and\ \bibinfo {author} {\bibfnamefont {J.~T.}\ \bibnamefont
  {Lewis}},\ }\bibfield  {title} {\emph {\bibinfo {title} {An operational
  approach to quantum probability},}\ }\href
  {http://dx.doi.org/10.1007/BF01647093} {\bibfield  {journal} {\bibinfo
  {journal} {Commun. Math. Phys.}\ }\textbf {\bibinfo {volume} {17}},\ \bibinfo
  {pages} {239} (\bibinfo {year} {1970})}\BibitemShut {NoStop}%
\bibitem [{\citenamefont {Lami}(2017)}]{lami_2017-1}%
  \BibitemOpen
  \bibfield  {author} {\bibinfo {author} {\bibfnamefont {L.}~\bibnamefont
  {Lami}},\ }\emph {\bibinfo {title} {Non-Classical Correlations in Quantum
  Mechanics and Beyond}},\ \href@noop {} {Ph.D. thesis},\ \bibinfo  {school}
  {Universitat Aut\`onoma de Barcelona} (\bibinfo {year} {2017}),\ \Eprint
  {http://arxiv.org/abs/1803.02902} {arXiv:1803.02902} \BibitemShut {NoStop}%
\bibitem [{\citenamefont {Barrett}(2007)}]{barrett_2007}%
  \BibitemOpen
  \bibfield  {author} {\bibinfo {author} {\bibfnamefont {J.}~\bibnamefont
  {Barrett}},\ }\bibfield  {title} {\emph {\bibinfo {title} {Information
  processing in generalized probabilistic theories},}\ }\href
  {http://dx.doi.org/10.1103/PhysRevA.75.032304} {\bibfield  {journal}
  {\bibinfo  {journal} {Phys. Rev. A}\ }\textbf {\bibinfo {volume} {75}},\
  \bibinfo {pages} {032304} (\bibinfo {year} {2007})}\BibitemShut {NoStop}%
\bibitem [{\citenamefont {Chiribella}\ \emph {et~al.}(2010)\citenamefont
  {Chiribella}, \citenamefont {D'Ariano},\ and\ \citenamefont
  {Perinotti}}]{chiribella_2010}%
  \BibitemOpen
  \bibfield  {author} {\bibinfo {author} {\bibfnamefont {G.}~\bibnamefont
  {Chiribella}}, \bibinfo {author} {\bibfnamefont {G.~M.}\ \bibnamefont
  {D'Ariano}}, \ and\ \bibinfo {author} {\bibfnamefont {P.}~\bibnamefont
  {Perinotti}},\ }\bibfield  {title} {\emph {\bibinfo {title} {Probabilistic
  theories with purification},}\ }\href
  {http://dx.doi.org/10.1103/PhysRevA.81.062348} {\bibfield  {journal}
  {\bibinfo  {journal} {Phys. Rev. A}\ }\textbf {\bibinfo {volume} {81}},\
  \bibinfo {pages} {062348} (\bibinfo {year} {2010})}\BibitemShut {NoStop}%
\bibitem [{\citenamefont {Kl{\"a}y}\ \emph {et~al.}(1987)\citenamefont
  {Kl{\"a}y}, \citenamefont {Randall},\ and\ \citenamefont
  {Foulis}}]{klay_1987}%
  \BibitemOpen
  \bibfield  {author} {\bibinfo {author} {\bibfnamefont {M.}~\bibnamefont
  {Kl{\"a}y}}, \bibinfo {author} {\bibfnamefont {C.}~\bibnamefont {Randall}}, \
  and\ \bibinfo {author} {\bibfnamefont {D.}~\bibnamefont {Foulis}},\
  }\bibfield  {title} {\emph {\bibinfo {title} {Tensor products and probability
  weights},}\ }\href {http://dx.doi.org/10.1007/BF00668911} {\bibfield
  {journal} {\bibinfo  {journal} {Int. J. Theor. Phys.}\ }\textbf {\bibinfo
  {volume} {26}},\ \bibinfo {pages} {199} (\bibinfo {year} {1987})}\BibitemShut
  {NoStop}%
\bibitem [{\citenamefont {Wilce}(1992)}]{wilce_1992}%
  \BibitemOpen
  \bibfield  {author} {\bibinfo {author} {\bibfnamefont {A.}~\bibnamefont
  {Wilce}},\ }\bibfield  {title} {\emph {\bibinfo {title} {Tensor products in
  generalized measure theory},}\ }\href {http://dx.doi.org/10.1007/BF00671964}
  {\bibfield  {journal} {\bibinfo  {journal} {Int. J. Theor. Phys.}\ }\textbf
  {\bibinfo {volume} {31}},\ \bibinfo {pages} {1915} (\bibinfo {year}
  {1992})}\BibitemShut {NoStop}%
\bibitem [{\citenamefont {Namioka}\ and\ \citenamefont
  {Phelps}(1969)}]{namioka_1969}%
  \BibitemOpen
  \bibfield  {author} {\bibinfo {author} {\bibfnamefont {I.}~\bibnamefont
  {Namioka}}\ and\ \bibinfo {author} {\bibfnamefont {R.~R.}\ \bibnamefont
  {Phelps}},\ }\bibfield  {title} {\emph {\bibinfo {title} {Tensor products of
  compact convex sets.}}\ }\href
  {https://projecteuclid.org/journals/pacific-journal-of-mathematics/volume-31/issue-2/Tensor-products-of-compact-convex-sets/pjm/1102977882.full}
  {\bibfield  {journal} {\bibinfo  {journal} {Pac. J. Math.}\ }\textbf
  {\bibinfo {volume} {31}},\ \bibinfo {pages} {469} (\bibinfo {year}
  {1969})}\BibitemShut {NoStop}%
\bibitem [{\citenamefont {Aubrun}\ \emph {et~al.}(2021)\citenamefont {Aubrun},
  \citenamefont {Lami}, \citenamefont {Palazuelos},\ and\ \citenamefont
  {Pl{\'a}vala}}]{aubrun_2021}%
  \BibitemOpen
  \bibfield  {author} {\bibinfo {author} {\bibfnamefont {G.}~\bibnamefont
  {Aubrun}}, \bibinfo {author} {\bibfnamefont {L.}~\bibnamefont {Lami}},
  \bibinfo {author} {\bibfnamefont {C.}~\bibnamefont {Palazuelos}}, \ and\
  \bibinfo {author} {\bibfnamefont {M.}~\bibnamefont {Pl{\'a}vala}},\
  }\bibfield  {title} {\emph {\bibinfo {title} {Entangleability of cones},}\
  }\href {http://dx.doi.org/10.1007/s00039-021-00565-5} {\bibfield  {journal}
  {\bibinfo  {journal} {Geom. Funct. Anal.}\ }\textbf {\bibinfo {volume}
  {31}},\ \bibinfo {pages} {181} (\bibinfo {year} {2021})}\BibitemShut
  {NoStop}%
\bibitem [{\citenamefont {Aubrun}\ \emph {et~al.}(2022)\citenamefont {Aubrun},
  \citenamefont {Lami}, \citenamefont {Palazuelos},\ and\ \citenamefont
  {Pl{\'a}vala}}]{aubrun_2022}%
  \BibitemOpen
  \bibfield  {author} {\bibinfo {author} {\bibfnamefont {G.}~\bibnamefont
  {Aubrun}}, \bibinfo {author} {\bibfnamefont {L.}~\bibnamefont {Lami}},
  \bibinfo {author} {\bibfnamefont {C.}~\bibnamefont {Palazuelos}}, \ and\
  \bibinfo {author} {\bibfnamefont {M.}~\bibnamefont {Pl{\'a}vala}},\
  }\bibfield  {title} {\emph {\bibinfo {title} {Entanglement and
  {{Superposition Are Equivalent Concepts}} in {{Any Physical Theory}}},}\
  }\href {http://dx.doi.org/10.1103/PhysRevLett.128.160402} {\bibfield
  {journal} {\bibinfo  {journal} {Phys. Rev. Lett.}\ }\textbf {\bibinfo
  {volume} {128}},\ \bibinfo {pages} {160402} (\bibinfo {year}
  {2022})}\BibitemShut {NoStop}%
\bibitem [{\citenamefont {Kimura}\ \emph {et~al.}(2010)\citenamefont {Kimura},
  \citenamefont {Nuida},\ and\ \citenamefont {Imai}}]{kimura_2010}%
  \BibitemOpen
  \bibfield  {author} {\bibinfo {author} {\bibfnamefont {G.}~\bibnamefont
  {Kimura}}, \bibinfo {author} {\bibfnamefont {K.}~\bibnamefont {Nuida}}, \
  and\ \bibinfo {author} {\bibfnamefont {H.}~\bibnamefont {Imai}},\ }\bibfield
  {title} {\emph {\bibinfo {title} {Distinguishability measures and entropies
  for general probabilistic theories},}\ }\href
  {http://dx.doi.org/10.1016/S0034-4877(10)00025-X} {\bibfield  {journal}
  {\bibinfo  {journal} {Rep. Math. Phys.}\ }\textbf {\bibinfo {volume} {66}},\
  \bibinfo {pages} {175} (\bibinfo {year} {2010})}\BibitemShut {NoStop}%
\bibitem [{\citenamefont {Bae}\ \emph {et~al.}(2016)\citenamefont {Bae},
  \citenamefont {Kim},\ and\ \citenamefont {Kwek}}]{bae_2016-1}%
  \BibitemOpen
  \bibfield  {author} {\bibinfo {author} {\bibfnamefont {J.}~\bibnamefont
  {Bae}}, \bibinfo {author} {\bibfnamefont {D.-G.}\ \bibnamefont {Kim}}, \ and\
  \bibinfo {author} {\bibfnamefont {L.-C.}\ \bibnamefont {Kwek}},\ }\bibfield
  {title} {\emph {\bibinfo {title} {Structure of {{Optimal State
  Discrimination}} in {{Generalized Probabilistic Theories}}},}\ }\href
  {http://dx.doi.org/10.3390/e18020039} {\bibfield  {journal} {\bibinfo
  {journal} {Entropy}\ }\textbf {\bibinfo {volume} {18}},\ \bibinfo {pages}
  {39} (\bibinfo {year} {2016})}\BibitemShut {NoStop}%
\bibitem [{\citenamefont {Aaronson}(2005)}]{aaronson_2005}%
  \BibitemOpen
  \bibfield  {author} {\bibinfo {author} {\bibfnamefont {S.}~\bibnamefont
  {Aaronson}},\ }\bibfield  {title} {\emph {\bibinfo {title} {Quantum
  computing, postselection, and probabilistic polynomial-time},}\ }\href
  {http://dx.doi.org/10.1098/rspa.2005.1546} {\bibfield  {journal} {\bibinfo
  {journal} {Proc. R. Soc. A}\ }\textbf {\bibinfo {volume} {461}},\ \bibinfo
  {pages} {3473} (\bibinfo {year} {2005})}\BibitemShut {NoStop}%
\bibitem [{\citenamefont {Fiur{\'a}{\v s}ek}(2006)}]{fiurasek_2006}%
  \BibitemOpen
  \bibfield  {author} {\bibinfo {author} {\bibfnamefont {J.}~\bibnamefont
  {Fiur{\'a}{\v s}ek}},\ }\bibfield  {title} {\emph {\bibinfo {title} {Optimal
  probabilistic estimation of quantum states},}\ }\href
  {http://dx.doi.org/10.1088/1367-2630/8/9/192} {\bibfield  {journal} {\bibinfo
   {journal} {New J. Phys.}\ }\textbf {\bibinfo {volume} {8}},\ \bibinfo
  {pages} {192} (\bibinfo {year} {2006})}\BibitemShut {NoStop}%
\bibitem [{\citenamefont {Gendra}\ \emph {et~al.}(2012)\citenamefont {Gendra},
  \citenamefont {{Ronco-Bonvehi}}, \citenamefont {Calsamiglia}, \citenamefont
  {{Mu{\~n}oz-Tapia}},\ and\ \citenamefont {Bagan}}]{gendra_2012}%
  \BibitemOpen
  \bibfield  {author} {\bibinfo {author} {\bibfnamefont {B.}~\bibnamefont
  {Gendra}}, \bibinfo {author} {\bibfnamefont {E.}~\bibnamefont
  {{Ronco-Bonvehi}}}, \bibinfo {author} {\bibfnamefont {J.}~\bibnamefont
  {Calsamiglia}}, \bibinfo {author} {\bibfnamefont {R.}~\bibnamefont
  {{Mu{\~n}oz-Tapia}}}, \ and\ \bibinfo {author} {\bibfnamefont
  {E.}~\bibnamefont {Bagan}},\ }\bibfield  {title} {\emph {\bibinfo {title}
  {Beating noise with abstention in state estimation},}\ }\href
  {http://dx.doi.org/10.1088/1367-2630/14/10/105015} {\bibfield  {journal}
  {\bibinfo  {journal} {New J. Phys.}\ }\textbf {\bibinfo {volume} {14}},\
  \bibinfo {pages} {105015} (\bibinfo {year} {2012})}\BibitemShut {NoStop}%
\bibitem [{\citenamefont {Combes}\ \emph {et~al.}(2014)\citenamefont {Combes},
  \citenamefont {Ferrie}, \citenamefont {Jiang},\ and\ \citenamefont
  {Caves}}]{combes_2014}%
  \BibitemOpen
  \bibfield  {author} {\bibinfo {author} {\bibfnamefont {J.}~\bibnamefont
  {Combes}}, \bibinfo {author} {\bibfnamefont {C.}~\bibnamefont {Ferrie}},
  \bibinfo {author} {\bibfnamefont {Z.}~\bibnamefont {Jiang}}, \ and\ \bibinfo
  {author} {\bibfnamefont {C.~M.}\ \bibnamefont {Caves}},\ }\bibfield  {title}
  {\emph {\bibinfo {title} {Quantum limits on postselected, probabilistic
  quantum metrology},}\ }\href {http://dx.doi.org/10.1103/PhysRevA.89.052117}
  {\bibfield  {journal} {\bibinfo  {journal} {Phys. Rev. A}\ }\textbf {\bibinfo
  {volume} {89}},\ \bibinfo {pages} {052117} (\bibinfo {year}
  {2014})}\BibitemShut {NoStop}%
\bibitem [{\citenamefont {Combes}\ and\ \citenamefont
  {Ferrie}(2015)}]{combes_2015}%
  \BibitemOpen
  \bibfield  {author} {\bibinfo {author} {\bibfnamefont {J.}~\bibnamefont
  {Combes}}\ and\ \bibinfo {author} {\bibfnamefont {C.}~\bibnamefont
  {Ferrie}},\ }\bibfield  {title} {\emph {\bibinfo {title} {Cost of
  postselection in decision theory},}\ }\href
  {http://dx.doi.org/10.1103/PhysRevA.92.022117} {\bibfield  {journal}
  {\bibinfo  {journal} {Phys. Rev. A}\ }\textbf {\bibinfo {volume} {92}},\
  \bibinfo {pages} {022117} (\bibinfo {year} {2015})}\BibitemShut {NoStop}%
\bibitem [{\citenamefont {{Arvidsson-Shukur}}\ \emph
  {et~al.}(2020)\citenamefont {{Arvidsson-Shukur}}, \citenamefont
  {Yunger~Halpern}, \citenamefont {Lepage}, \citenamefont {Lasek},
  \citenamefont {Barnes},\ and\ \citenamefont {Lloyd}}]{arvidsson-shukur_2020}%
  \BibitemOpen
  \bibfield  {author} {\bibinfo {author} {\bibfnamefont {D.~R.~M.}\
  \bibnamefont {{Arvidsson-Shukur}}}, \bibinfo {author} {\bibfnamefont
  {N.}~\bibnamefont {Yunger~Halpern}}, \bibinfo {author} {\bibfnamefont
  {H.~V.}\ \bibnamefont {Lepage}}, \bibinfo {author} {\bibfnamefont {A.~A.}\
  \bibnamefont {Lasek}}, \bibinfo {author} {\bibfnamefont {C.~H.~W.}\
  \bibnamefont {Barnes}}, \ and\ \bibinfo {author} {\bibfnamefont
  {S.}~\bibnamefont {Lloyd}},\ }\bibfield  {title} {\emph {\bibinfo {title}
  {Quantum advantage in postselected metrology},}\ }\href
  {http://dx.doi.org/10.1038/s41467-020-17559-w} {\bibfield  {journal}
  {\bibinfo  {journal} {Nat. Commun.}\ }\textbf {\bibinfo {volume} {11}},\
  \bibinfo {pages} {3775} (\bibinfo {year} {2020})}\BibitemShut {NoStop}%
\bibitem [{\citenamefont {Lloyd}\ \emph
  {et~al.}(2011{\natexlab{a}})\citenamefont {Lloyd}, \citenamefont {Maccone},
  \citenamefont {{Garcia-Patron}}, \citenamefont {Giovannetti}, \citenamefont
  {Shikano}, \citenamefont {Pirandola}, \citenamefont {Rozema}, \citenamefont
  {Darabi}, \citenamefont {Soudagar}, \citenamefont {Shalm},\ and\
  \citenamefont {Steinberg}}]{lloyd_2011}%
  \BibitemOpen
  \bibfield  {author} {\bibinfo {author} {\bibfnamefont {S.}~\bibnamefont
  {Lloyd}}, \bibinfo {author} {\bibfnamefont {L.}~\bibnamefont {Maccone}},
  \bibinfo {author} {\bibfnamefont {R.}~\bibnamefont {{Garcia-Patron}}},
  \bibinfo {author} {\bibfnamefont {V.}~\bibnamefont {Giovannetti}}, \bibinfo
  {author} {\bibfnamefont {Y.}~\bibnamefont {Shikano}}, \bibinfo {author}
  {\bibfnamefont {S.}~\bibnamefont {Pirandola}}, \bibinfo {author}
  {\bibfnamefont {L.~A.}\ \bibnamefont {Rozema}}, \bibinfo {author}
  {\bibfnamefont {A.}~\bibnamefont {Darabi}}, \bibinfo {author} {\bibfnamefont
  {Y.}~\bibnamefont {Soudagar}}, \bibinfo {author} {\bibfnamefont {L.~K.}\
  \bibnamefont {Shalm}}, \ and\ \bibinfo {author} {\bibfnamefont {A.~M.}\
  \bibnamefont {Steinberg}},\ }\bibfield  {title} {\emph {\bibinfo {title}
  {Closed {{Timelike Curves}} via {{Postselection}}: {{Theory}} and
  {{Experimental Test}} of {{Consistency}}},}\ }\href
  {http://dx.doi.org/10.1103/PhysRevLett.106.040403} {\bibfield  {journal}
  {\bibinfo  {journal} {Phys. Rev. Lett.}\ }\textbf {\bibinfo {volume} {106}},\
  \bibinfo {pages} {040403} (\bibinfo {year} {2011}{\natexlab{a}})}\BibitemShut
  {NoStop}%
\bibitem [{\citenamefont {Lloyd}\ \emph
  {et~al.}(2011{\natexlab{b}})\citenamefont {Lloyd}, \citenamefont {Maccone},
  \citenamefont {Garcia-Patron}, \citenamefont {Giovannetti},\ and\
  \citenamefont {Shikano}}]{lloyd_2011b}%
  \BibitemOpen
  \bibfield  {author} {\bibinfo {author} {\bibfnamefont {S.}~\bibnamefont
  {Lloyd}}, \bibinfo {author} {\bibfnamefont {L.}~\bibnamefont {Maccone}},
  \bibinfo {author} {\bibfnamefont {R.}~\bibnamefont {Garcia-Patron}}, \bibinfo
  {author} {\bibfnamefont {V.}~\bibnamefont {Giovannetti}}, \ and\ \bibinfo
  {author} {\bibfnamefont {Y.}~\bibnamefont {Shikano}},\ }\bibfield  {title}
  {\emph {\bibinfo {title} {Quantum mechanics of time travel through
  post-selected teleportation},}\ }\href
  {http://dx.doi.org/10.1103/PhysRevD.84.025007} {\bibfield  {journal}
  {\bibinfo  {journal} {Phys. Rev. D}\ }\textbf {\bibinfo {volume} {84}},\
  \bibinfo {pages} {025007} (\bibinfo {year} {2011}{\natexlab{b}})}\BibitemShut
  {NoStop}%
\bibitem [{\citenamefont {Brun}\ and\ \citenamefont {Wilde}(2012)}]{brun_2012}%
  \BibitemOpen
  \bibfield  {author} {\bibinfo {author} {\bibfnamefont {T.~A.}\ \bibnamefont
  {Brun}}\ and\ \bibinfo {author} {\bibfnamefont {M.~M.}\ \bibnamefont
  {Wilde}},\ }\bibfield  {title} {\emph {\bibinfo {title} {Perfect {{State
  Distinguishability}} and {{Computational Speedups}} with {{Postselected
  Closed Timelike Curves}}},}\ }\href
  {http://dx.doi.org/10.1007/s10701-011-9601-0} {\bibfield  {journal} {\bibinfo
   {journal} {Found. Phys.}\ }\textbf {\bibinfo {volume} {42}},\ \bibinfo
  {pages} {341} (\bibinfo {year} {2012})}\BibitemShut {NoStop}%
\bibitem [{\citenamefont {Gisin}(1996)}]{gisin_1996}%
  \BibitemOpen
  \bibfield  {author} {\bibinfo {author} {\bibfnamefont {N.}~\bibnamefont
  {Gisin}},\ }\bibfield  {title} {\emph {\bibinfo {title} {Hidden quantum
  nonlocality revealed by local filters},}\ }\href
  {http://dx.doi.org/10.1016/S0375-9601(96)80001-6} {\bibfield  {journal}
  {\bibinfo  {journal} {Phys. Lett. A}\ }\textbf {\bibinfo {volume} {210}},\
  \bibinfo {pages} {151} (\bibinfo {year} {1996})}\BibitemShut {NoStop}%
\bibitem [{\citenamefont {Kent}(1998)}]{kent_1998}%
  \BibitemOpen
  \bibfield  {author} {\bibinfo {author} {\bibfnamefont {A.}~\bibnamefont
  {Kent}},\ }\bibfield  {title} {\emph {\bibinfo {title} {Entangled {{Mixed
  States}} and {{Local Purification}}},}\ }\href
  {http://dx.doi.org/10.1103/PhysRevLett.81.2839} {\bibfield  {journal}
  {\bibinfo  {journal} {Phys. Rev. Lett.}\ }\textbf {\bibinfo {volume} {81}},\
  \bibinfo {pages} {2839} (\bibinfo {year} {1998})}\BibitemShut {NoStop}%
\bibitem [{\citenamefont {Horodecki}\ \emph {et~al.}(1999)\citenamefont
  {Horodecki}, \citenamefont {Horodecki},\ and\ \citenamefont
  {Horodecki}}]{horodecki_1999-1}%
  \BibitemOpen
  \bibfield  {author} {\bibinfo {author} {\bibfnamefont {M.}~\bibnamefont
  {Horodecki}}, \bibinfo {author} {\bibfnamefont {P.}~\bibnamefont
  {Horodecki}}, \ and\ \bibinfo {author} {\bibfnamefont {R.}~\bibnamefont
  {Horodecki}},\ }\bibfield  {title} {\emph {\bibinfo {title} {General
  teleportation channel, singlet fraction, and quasidistillation},}\ }\href
  {http://dx.doi.org/10.1103/PhysRevA.60.1888} {\bibfield  {journal} {\bibinfo
  {journal} {Phys. Rev. A}\ }\textbf {\bibinfo {volume} {60}},\ \bibinfo
  {pages} {1888} (\bibinfo {year} {1999})}\BibitemShut {NoStop}%
\bibitem [{\citenamefont {Aharonov}\ \emph {et~al.}(1988)\citenamefont
  {Aharonov}, \citenamefont {Albert},\ and\ \citenamefont
  {Vaidman}}]{aharonov_1988}%
  \BibitemOpen
  \bibfield  {author} {\bibinfo {author} {\bibfnamefont {Y.}~\bibnamefont
  {Aharonov}}, \bibinfo {author} {\bibfnamefont {D.~Z.}\ \bibnamefont
  {Albert}}, \ and\ \bibinfo {author} {\bibfnamefont {L.}~\bibnamefont
  {Vaidman}},\ }\bibfield  {title} {\emph {\bibinfo {title} {How the result of
  a measurement of a component of the spin of a spin-1/2 particle can turn out
  to be 100},}\ }\href {http://dx.doi.org/10.1103/PhysRevLett.60.1351}
  {\bibfield  {journal} {\bibinfo  {journal} {Phys. Rev. Lett.}\ }\textbf
  {\bibinfo {volume} {60}},\ \bibinfo {pages} {1351} (\bibinfo {year}
  {1988})}\BibitemShut {NoStop}%
\bibitem [{\citenamefont {Pusey}(2014)}]{pusey_2014}%
  \BibitemOpen
  \bibfield  {author} {\bibinfo {author} {\bibfnamefont {M.~F.}\ \bibnamefont
  {Pusey}},\ }\bibfield  {title} {\emph {\bibinfo {title} {Anomalous {{Weak
  Values Are Proofs}} of {{Contextuality}}},}\ }\href
  {http://dx.doi.org/10.1103/PhysRevLett.113.200401} {\bibfield  {journal}
  {\bibinfo  {journal} {Phys. Rev. Lett.}\ }\textbf {\bibinfo {volume} {113}},\
  \bibinfo {pages} {200401} (\bibinfo {year} {2014})}\BibitemShut {NoStop}%
\bibitem [{\citenamefont {Oreshkov}\ and\ \citenamefont
  {Cerf}(2015)}]{oreshkov_2015}%
  \BibitemOpen
  \bibfield  {author} {\bibinfo {author} {\bibfnamefont {O.}~\bibnamefont
  {Oreshkov}}\ and\ \bibinfo {author} {\bibfnamefont {N.~J.}\ \bibnamefont
  {Cerf}},\ }\bibfield  {title} {\emph {\bibinfo {title} {Operational
  formulation of time reversal in quantum theory},}\ }\href
  {http://dx.doi.org/10.1038/nphys3414} {\bibfield  {journal} {\bibinfo
  {journal} {Nat. Phys.}\ }\textbf {\bibinfo {volume} {11}},\ \bibinfo {pages}
  {853} (\bibinfo {year} {2015})}\BibitemShut {NoStop}%
\bibitem [{\citenamefont {Harrow}\ and\ \citenamefont
  {Montanaro}(2017)}]{harrow_2017}%
  \BibitemOpen
  \bibfield  {author} {\bibinfo {author} {\bibfnamefont {A.~W.}\ \bibnamefont
  {Harrow}}\ and\ \bibinfo {author} {\bibfnamefont {A.}~\bibnamefont
  {Montanaro}},\ }\bibfield  {title} {\emph {\bibinfo {title} {Quantum
  computational supremacy},}\ }\href {http://dx.doi.org/10.1038/nature23458}
  {\bibfield  {journal} {\bibinfo  {journal} {Nature}\ }\textbf {\bibinfo
  {volume} {549}},\ \bibinfo {pages} {203} (\bibinfo {year}
  {2017})}\BibitemShut {NoStop}%
\end{thebibliography}%

%%%%%%%%%%%%%%%%%%%%%%%%%%%%%%%%%%%%%%%%%%%%%%%%%%%%%%%%%%%%%%%%%%
%%%%%%%%%%%%%%%%%%%%%%%%%%%%%%%%%%%%%%%%%%%%%%%%%%%%%%%%%%%%%%%%%%
%%%%%%%%%%%%%%%%%%%%%%%%%%%%%%%%%%%%%%%%%%%%%%%%%%%%%%%%%%%%%%%%%%
%%%%%%%%%%%%%%%%%%%%%%%%%%%%%%%%%%%%%%%%%%%%%%%%%%%%%%%%%%%%%%%%%%
%%%%%%%%%%%%%%%%%%%%%%%%%%%%%%%%%%%%%%%%%%%%%%%%%%%%%%%%%%%%%%%%%%
%%%%%%%%%%%%%%%%%%%%%%%%%%%%%%%%%%%%%%%%%%%%%%%%%%%%%%%%%%%%%%%%%%
%%%%%%%%%%%%%%%%%%%%%%%%%%%%%%%%%%%%%%%%%%%%%%%%%%%%%%%%%%%%%%%%%%

%%%%%%%%%%%%%%%%%%%%%%%%%%%%%%%%%%%%%%%%%%%%%%%%%%%%%%%%%%%%%
\end{document}